\documentclass[10pt]{article}
\usepackage{amsfonts}
\usepackage{amsmath}
\usepackage{amsthm}
\usepackage{authblk}
\usepackage{blindtext}
\usepackage[british]{babel}
\usepackage{mathtools}
\usepackage[a4paper]{geometry}
\usepackage{graphicx}
\usepackage[utf8]{inputenc}

\usepackage[square,numbers]{natbib}
\usepackage{pdflscape}
\usepackage{xcolor}
\usepackage{amssymb}
\usepackage{float}

\usepackage{bm}
\usepackage{bbm}
\usepackage{rotating}
\usepackage{physics}

\usepackage{subcaption}
\usepackage{caption}
\usepackage{mathrsfs}
\usepackage{enumerate}

\newcommand{\ZZ}{\mathbb{Z}}

\newcommand{\RR}{\mathbb{R}}

\newcommand{\PP}{\mathbb{P}}
\newcommand{\EE}{\mathbb{E}}

\usepackage{bigints}

\DeclareMathOperator*{\Var}{Var}

\DeclareMathOperator*{\Cov}{Cov}

\DeclareMathOperator*{\Covu}{Cov_\pi}
\DeclareMathOperator*{\Varu}{Var_\pi}

\DeclareMathOperator*{\Covnu}{Cov_\nu}
\DeclareMathOperator*{\Varnu}{Var_\nu}
\DeclareMathOperator*{\Vollie}{Vol}
\DeclareMathOperator*{\Bern}{Bernoulli}

\newtheorem{assumption}{Assumption}

\newtheorem{drift}{Drift Condition}
\newtheorem{minorisation}{Minorisation Condition}
 
\newtheorem{theorem}{Theorem}[section]

\newtheorem{corollary}[theorem]{Corollary}
\newtheorem{lemma}[theorem]{Lemma}
\newtheorem{proposition}[theorem]{Proposition}
\theoremstyle{remark}
\usepackage[colorlinks=true,urlcolor=blue,citecolor=black,linkcolor=blue]{hyperref}

\newcommand{\forceindent}{\leavevmode{\parindent=1em\indent}}

\newtheorem{remarkx}[theorem]{Remark}
\newenvironment{remark}
{\pushQED{\qed}\remarkx}
{\popQED\endremarkx}
\renewcommand{\i}{\mathrm i}
\newcommand{\E}{\mathbb E}
\renewcommand{\P}{\mathbb P}
\newcommand{\R}{\mathbb R}

\newcommand{\floor}[1]{\lfloor #1 \rfloor}
\newcommand{\ceil}[1]{\lceil #1 \rceil}

\makeatletter
\newcommand{\undersim}[1]{\mathrel{\mathpalette\@undersim{#1}}}
\newcommand{\@undersim}[2]{%
  \vcenter{%
    \ialign{%
      ##\cr
      $\m@th#1#2$\cr
      \noalign{\nointerlineskip\kern.2ex}
      $\m@th#1\sim$\cr
      \noalign{\kern-.4ex}
    }%
  }%
}

\newcommand{\lsim}{\undersim{<}}
\makeatother
\newcommand{\bigO}{\mathcal{O}}
\makeindex
\usepackage{sectsty}
\usepackage{titlesec}
\titlelabel{\thetitle.\quad}
\sectionfont{\fontsize{13}{12}\selectfont}
\setcitestyle{semicolon}
\setcitestyle{notesep={; },square,aysep={},yysep={;}}

\title{Gaussian Approximation and Output Analysis for High-Dimensional MCMC}

\author[1]{Ardjen Pengel}
\author[2]{Jun Yang}
\author[3]{Zhou Zhou}
\affil[1]{\small Department of Pure Mathematics and Mathematical Statistics,  University of Cambridge}
\affil[2]{\small Department of Mathematical Sciences, University of Copenhagen}
\affil[3]{\small Department of Statistical Sciences, University of Toronto}
\affil[ ]{\small\textit{E-mail:} \url{alp98@cam.ac.uk}, \url{jy@math.ku.dk}, \url{zhou@utstat.toronto.edu}}

\begin{document}
\date{} 
\maketitle

\setlength\parindent{15pt}
\setcitestyle{semicolon}

\begin{abstract} 
\noindent
The widespread use of Markov Chain Monte Carlo (MCMC) methods for high-dimensional applications has motivated research into the scalability of these algorithms with respect to the dimension of the problem.  Despite this, numerous problems concerning output analysis in high-dimensional settings have remained unaddressed.  We present novel quantitative Gaussian approximation results for a broad range of MCMC algorithms. Notably, we analyse the dependency of the obtained approximation errors on the dimension of both the target distribution and the feature space.
We demonstrate how these Gaussian approximations can be applied in output analysis. This includes determining the simulation effort required to guarantee Markov chain central limit theorems and consistent  estimation of the variance and effective sample size in high-dimensional settings.
We give quantitative convergence bounds for termination criteria and show that the termination time of a wide class of MCMC algorithms scales polynomially in dimension while ensuring a desired level of precision. Our results offer guidance to practitioners for obtaining appropriate standard errors and deciding the minimum simulation effort of MCMC algorithms in both multivariate and high-dimensional settings.
\end{abstract}

\setlength\parindent{0pt}
\section{Introduction}

Markov Chain Monte Carlo (MCMC) methods are widely 
applied in various high-dimensional settings, such as those encountered in computational Bayesian statistics and machine learning, see for example \cite{brooks2011handbook, gawlikowski2023survey,springenberg2016bayesian, welling2011bayesian}. MCMC methods are generally acknowledged to be the most versatile
algorithms for simulating a probability distribution of interest. 
%
%
%
We consider the problem of
sampling from a high-dimensional probability distribution $\pi$ defined on $E\subseteq \mathbb{R}^N$.
Typically, the objective is to compute some collection of features of this distribution, that can usually be expressed as expectations with respect to $\pi$, in other words, we are interested in $\pi(f)=\int f(x) \pi(dx)$, for some appropriately integrable function $f: E\rightarrow E'$ with $E'\subseteq \mathbb{R}^{
{d}}$.  We will refer to $E$ and $E'$ as the state space and feature space respectively.\\

The fundamental idea behind MCMC is to construct a Markov chain $X=(X_t)_{t\in \mathbb{N}}$ such that its stationary distribution is given by the distribution of interest. There are two practical questions that need to be addressed for every application of MCMC. Firstly, when is it reasonable to assume that the sampling algorithm is exploring the state space according to its stationary distribution? And when is it justified to terminate the simulation? Additionally, in high-dimensional settings, the scale of the problem introduces further challenges to these key inquiries. In \cite{qin2021limitations} and \cite{rajaratnam2015mcmc} it is shown that many results regarding the convergence of MCMC samplers, that do not explicitly take the impact of dimensionality into consideration are inapplicable to high-dimensional scenarios. This highlighted the importance of the so-called convergence complexity of MCMC, which entails understanding how the convergence properties scale as the dimension of the problem grows.
While there has been a significant amount of research done to provide dimension-dependent performance guarantees for MCMC algorithms, the focus has been primarily on the first question, see for example \cite{altmeyer2022polynomial,bou2023mixing,  dalalyan2017theoretical,durmus2017nonasymptotic,durmus2019high,hairer2014spectral,nickl2022polynomial, qin2019convergence, qin2022wasserstein,roberts2001optimal,yang2023complexity}.
The issue of when to terminate the simulation has not been examined as thoroughly. As noted by \cite{gong2016}, many output analysis tools used for addressing the termination question, such as visual inspection of trace plots and classical convergence diagnostics, are only appropriate for problems of moderate dimension. This paper seeks to fill that gap by giving a rigorous exploration of the convergence complexity of MCMC output analysis and providing theoretical guarantees for termination criteria in high-dimensional settings.  \\
%
%
%
%

In \cite{glynn}, asymptotic validity of several sequential stopping rules is established under the assumption of a functional central limit theorem (FCLT) for the simulation process and a functional weak law of large numbers (FWLLN) for an estimator of the asymptotic covariance.
Consider the rescaled partial-sum process of our Markov chain $X$  given by  
$S^{(T)}:=(S^{(T)}_t)_{t\in[0,1]}$ where 
\begin{equation}
S^{(T)}_t:=\frac{1}{\sqrt{T}}\sum_{t=1}^{\floor{Tt}}(f(X_t)-\pi(f)), \ t\in [0,1],
\end{equation}
 and let $Z:=(Z_t)_{t\in[0,1]}$ be defined as $Z_t:=\Sigma_f^{1/2}W_t$, where $\Sigma_f^{1/2}$ denotes the square root of the time-average covariance matrix of $X$ and $W=(W_t)_{t\in[0,1]}$ is a standard $d$--dimensional Brownian motion, and let $D[0,1]$ denote the Skorokhod space, which consists of all $\R^d$-valued c\`adl\`ag functions 
 with domain $[0,1]$, and let $\mathscr{D}$
 denote the Borel $\sigma$--algebra generated by the Skorokhod topology. The functional central limit theorem states that for every continuity set $A\in \mathscr{D}$ of $Z$,  we have that 
  \begin{equation}
 \label{intro_fclt1}
     \abs{\P\left(S^{(T)} \in A \right)-\P\left(Z \in A \right)}=o(1)  \ \textrm{as}\ T\rightarrow \infty.
     \end{equation}
In order to generalise the results of \cite{glynn} and obtain asymptotic validity of termination criteria in high-dimensional settings, we require quantitative Gaussian approximations of our Markov chain $X$. These Gaussian approximation results quantify the rate at which the trajectories of the partial sum process can be approximated by the appropriately scaled trajectories of a Gaussian process and are thus a refinement of the functional central limit theorem, as shown in for example \cite[Theorem 1.E]{philippSIP}.
%
%
%
%
%
%
%
We say that a weak Gaussian approximation holds for $X=(X_t)_{t\in \mathbb{N}}$ if the process can be defined on a probability space, together with a Brownian motion $W$, such that 
\begin{align}
\label{weak_GA}
\lim_{T\rightarrow \infty} \mathbb{P} \left(\frac{1}{\Psi_T} \abs{\sum_{t=0}^Tf(X_t)- T\pi(f) - \Sigma_f^{1/2} W_T} \leqslant K\bar{\psi}_N \psi_d  \right)=1,   
\end{align}
where $|\cdot|$ denotes the Euclidean norm, $K$ denotes a dimension-independent almost surely finite constant, and $\bar{\psi}_N, \psi_d$, and $\Psi_{T}$ denote the dependence of the approximation error on the dimension of the state space, the dimension of the feature space, and the sampling time respectively. Similarly, we say that a strong Gaussian approximation holds for $X=(X_t)_{t\in \mathbb{N}}$ if
\begin{align}
\mathbb{P}\left(\limsup_{T\rightarrow \infty}  \frac{1}{\Psi_T} \abs{\sum_{t=0}^Tf(X_t)- T\pi(f) - \Sigma_f^{1/2} W_T}\leqslant K\bar{\psi}_N \psi_d  \right)=1.    
\end{align}
We use the customary notation for Gaussian approximations, i.e.,  
\begin{equation}
\label{GA_introduction}
\abs{\sum_{t=0}^Tf(X_t)-T\pi(f) - \Sigma_f^{1/2} W_T}=  \left\{
            \begin{array}{lll}                                          
                  \bigO_{P}\left(\bar{\psi}_N \psi_d \Psi_{T}\right)  \\
                                   \textcolor{white}{.}   \\ 
                 \bigO_{a.s.}\left(\bar{\psi}_N \psi_d \Psi_{T}\right) 
                 
                \end{array},
              \right.  
\end{equation} 

where $\bigO_{P}$ and $\bigO_{a.s.}$ denote the weak and strong approximation respectively. These Gaussian approximation results are closely related to the convergence rate of the FCLT, as shown in  \citet[Theorem 1.16 and Theorem 1.17]{weighted_approx}. These approximation results are powerful tools used to obtain numerous results in both probability and statistics as seen in, e.g.,  \citet{applications_sip} and \citet{shorack_empirical}. Gaussian approximation results also play a central role in the analysis of estimators of the asymptotic variance and MCMC output analysis, see for example \cite{damerdji1991,damerdji1994,flegal_bm,  honestjones, jones_fixed,multivariate_consistency,multivariate_output}.
 In the one-dimensional case, Gaussian approximation results for MCMC were obtained by \cite{ flegal_bm, sip_regenerative,jones_fixed} and \cite{merlevede2015}. These results were extended to multivariate and continuous-time settings by \cite{banerjee2022,li2024multivariate} and \cite{pengel2024strong} respectively. For a more extensive overview of Gaussian approximation results obtained in settings with a dependence structure, we refer to \cite{lu2022}. This paper contributes to the field by introducing novel explicit and dimension-dependent bounds for the  Gaussian approximation of a wide class of MCMC samplers. Furthermore, we are the first to obtain the optimal approximation rate for MCMC samplers in a multivariate setting. These approximation results in turn enable us to obtain new theoretical guarantees regarding MCMC output analysis in both multivariate and high-dimensional settings. \\

%
%
%
%
%
We follow the framework of \cite{qin2019convergence,yang2023complexity,zhou2022dimension} where a drift-and-minorisation approach is used to obtain quantitative convergence
bounds in high-dimensional settings. A drift condition describes how fast the Markov chain moves through the state space, while the minorisation condition controls how fast the Markov chain forgets its past.  Our analysis considers both geometric and polynomial drift conditions, which characterise the varying speeds at which the Markov chain can move towards subsets of the state space.
We also consider both the one-step and multi-step minorisation case, which covers all widely used MCMC algorithms, see \cite[Proposition 5.4.5]{meyn_tweedie_2012} and \cite{orey1971limit}. We will further impose some regularity conditions on the asymptotic covariance matrix $\Sigma_f$ and mainly consider component-wise moment conditions of the form
$\sup_{i\in \{1,\cdots,d\}}\pi(\abs{f_i}^{p+\epsilon})< \infty$ for given $p>2$ and $\epsilon>0$.
Table \ref{table:table_approx_1} summarises the results of Theorems \ref{theorem_approx_discrete_one} and \ref{theorem_approx_discrete_multi} which describe the dependence of the Gaussian approximation rate on the simulation time.
%
\begin{table}[ht]
\centering
\begin{center}
\begin{tabular}{|c|c|c|}
    \hline
    \textbf{} & \textbf{one-step minorisation} & \textbf{multi-step minorisation} \\
    \hline
    \rule{0pt}{16pt}
        geometric drift & $\displaystyle T^{1/p}\log(T)$ & $\displaystyle T^{1/4+1/4(p-1)}\log(T)$\\
   \hline
       \rule{0pt}{16pt}
    polynomial drift & $ \displaystyle T^{1/p_0}\log(T)$ & $\displaystyle T^{1/4+1/4(p_0-1)}\log(T)$  \\
   \hline 
\end{tabular}
\caption{Gaussian approximation rate $\Psi_T$}
\label{table:table_approx_1}
\end{center}
\end{table}

\vspace{-\baselineskip}
Here $p_0$ is specified in Equation \eqref{p_0definition1} of Theorem \ref{theorem_approx_discrete_one} and depends explicitly on both the moment condition and the degree of polynomial drift. Our Gaussian approximation results cover a larger class of polynomially ergodic Markov chains than aforementioned works and are the first to quantify the influence of the polynomial drift on the approximation rate.  Furthermore, we note that besides the result of \cite{merlevede2015}, we are the first to obtain the optimal rate in a Markov chain setting. The proof of the Gaussian approximation for the one-step minorisation case builds upon the approach of \cite{merlevede2015}.  \cite{merlevede2015} also utilises the regenerative structure obtained from the one-step minorisation condition and subsequently applies the Gaussian approximation of \cite{zaitsev1998multidimensional}. However, that result is not applicable to our setting since it requires the existence of the moment-generating function. Furthermore, the result of \cite{merlevede2015} assumes geometric ergodicity and bounded one-dimensional functions $f$. Currently, under the moment condition $\pi(\abs{f}^p)$, the best-known Gaussian approximation error for MCMC algorithms is $\max \left\{T^{1/p },T^{1/4}\right\}\log(T)$, see the results of  \cite{sip_regenerative,flegal_bm,jones_fixed} and \cite{banerjee2022,li2024multivariate} for the one-dimensional and multivariate setting respectively.
Under the same moment conditions, we now obtain the optimal Koml\'os--Major--Tusn\'ady approximation rate, up to a logarithmic factor, for the one-step minorisation case. For the multi-step minorisation case, we do not recover the optimal rate; however, to the best of our knowledge, our rate is the best rate currently available for the class of processes considered.\\

Additionally, our Gaussian approximation results provide explicit bounds for $\bar{\psi}_N$ and $\psi_d$, the dependence of the dimension of the state space and feature space on the approximation rate. For the weak Gaussian approximation, we obtain $\psi_d=\sqrt{d}\pi(\abs{f}^p)^{1/p}$.
Since $\pi(\abs{f}^p)^{1/p}$ is of order $\bigO(\sqrt{d})$, our results in general imply a linear dependence in the dimension of the feature space.
However, this can be  improved under growth conditions on the moments of $\abs{f}$ under $\pi$.
This is reasonable in settings with sparsity and regularisation. In \cite{zhai2018high} it is shown that for i.i.d. random vectors and bounded functions $f$  
the convergence rate of the CLT is lower bounded by $\sqrt{d}\norm{f}_\infty$, where $\norm{\cdot}_\infty$ denotes the usual supremum norm. Our obtained results obtain the same dependence on the dimension of the feature space as the currently best-known weak Gaussian approximation results for independent random vectors satisfying only $p$ moments given in \cite{eldan2020clt,mies2023sequential}. The dependence of the dimension of the state space is expressed in terms of the parameters of the drift and minorisation conditions. Let $q$ denote the speed of polynomial drift, given stability conditions on the drift and minorisation, our obtained weak Gaussian approximation results guarantee a CLT for Markov processes, provided that the dimension of the feature space grows as given in Table \ref{table:table_dim_growth} for any $\bar{\varepsilon}>0$.
\begin{table}[ht]
\begin{center}
\begin{tabular}{|c|c|c|}
    \hline
    \textbf{} & \textbf{one-step minorisation} & \textbf{multi-step minorisation} \\
    \hline
    \rule{0pt}{16pt}
        geometric drift & $o\Big(T^{1/2-\bar{\varepsilon}}\Big)$ & $o\Big(T^{1/4-\bar{\varepsilon}}\Big)$\\
   \hline
       \rule{0pt}{16pt}
    polynomial drift & $\displaystyle o\left(T^{\frac{q-2}{2q}-\bar{\varepsilon}}\right)$ & $\displaystyle o\left(T^{\frac{q-2}{4(q-1)}-\bar{\varepsilon}}\right) $ \\
   \hline 
\end{tabular}
\caption{Growth rate $d$ such that CLT holds for sufficiently large $p$}
\label{table:table_dim_growth}
\end{center}
\end{table}

As the speed of polynomial drift tends to the geometric case, which corresponds to $q\rightarrow \infty$, we show that the corresponding approximation errors in Table \ref{table:table_approx_1} coincide and consequently also the growth conditions of the dimension given in Table \ref{table:table_dim_growth}.\\

In order to construct valid confidence ellipsoids for our features of interest,
we require the weak Gaussian approximation to hold. However, the usual approach in statistics
where the dimension is allowed to grow with the sample size, which in this context corresponds
to the simulation time, is not appropriate. In the MCMC setting, the dimension of the target
distribution and the set of features of interest are known prior to the simulation. Consequently,
the critical question becomes how large the simulation time needs to be to ensure that the Gaussian approximations remain reliable in the setting where the dimension of the problem grows.
In Table \ref{table:table_sim_growth_intro} below we see how the simulation time is required to grow with the dimension of the feature space in order to guarantee a CLT. For two sequences $(a_n)$ and $(b_n)$ we write $a_n=\Omega(b_n)$ if $\limsup_{n\rightarrow \infty} \frac{a_n}{b_n}>0.$  \\

\begin{table}[ht]
\begin{center}
\begin{tabular}{|c|c|c|}
    \hline
    \textbf{} & \textbf{one-step minorisation} & \textbf{multi-step minorisation} \\
    \hline
    \rule{0pt}{18pt}
        geometric drift & $\displaystyle \Omega\left(d^{2+\bar{\varepsilon}}\right)$ & $\displaystyle \Omega\left(d^{4+\bar{\varepsilon}}\right) $ \\
   \hline
       \rule{0pt}{18pt}
    polynomial drift & $\displaystyle \Omega\left(d^{\frac{2q}{q-2}+\bar{\varepsilon}}\right)$ & $\displaystyle \Omega\left(d^{\frac{4(q-1)}{q-2}+\bar{\varepsilon}}\right) $ \\
   \hline 
\end{tabular}
\caption{Simulation time $T$ such that CLT holds given sufficiently large $p$}
\label{table:table_sim_growth_intro}
\end{center}
\end{table}
\vspace{-\baselineskip}

 For applications in Bayesian statistics, our results provide a direct link between the statistical model complexity and the computational complexity of the MCMC algorithm. Any theoretical guarantees for the posterior distribution, such as large-sample concentration results and prior induced regularisation, have direct implications on the required running times for sampling algorithms. We note that under additional assumptions on the target distribution, the growth rates presented in Table \ref{table:table_approx_1} and Table \ref{table:table_dim_growth} can be greatly improved, see Remark \ref{remark_growth} for a more extensive discussion on this result.  \\

By taking the dimensionality into account, multiple insights that do not arise in the finite-dimensional setting are revealed. Current results on MCMC output analysis often hold for every initial distribution. However, in the high-dimensional setting, a good initialisation of the Markov chain is crucial, as improper initial states can lead to exponentially increasing bounds. Additionally, also the minorisation volume plays a vital role, since contrary to the low-dimensional setting, a naive minorisation lower bound can cause to the Gaussian approximation error rate to increase exponentially in dimension. For the mixing time of high-dimensional MCMC algorithms, similar observations have been made, see for example \cite{qin2019convergence,rajaratnam2015mcmc,yang2023complexity}.
 Furthermore, also the isometry of the asymptotic covariance matrix plays a larger role, since for isotropic target distributions we are able to give better convergence guarantees. \\

Through our obtained Gaussian approximation results, we are able to extend various results concerning output analysis to a wider array of applications. Firstly, we consider the estimation of the asymptotic covariance matrix. Many key aspects of output analysis for MCMC depend on the uncertainty quantification of our sampling algorithm. Estimating the Monte Carlo standard error is essential for ensuring the credibility of our simulation results and for computing many convergence diagnostics.  In \cite{multivariate_consistency} and \cite{multivariate_output} the consistency of the spectral variance and batch means estimator are proven respectively, under the assumption of a  Gaussian approximation with an implicit rate. Consequently, our results are immediately applicable and can be used to adapt the tuning parameters of the considered variance estimation methods to high-dimensional settings. In accordance with empirical findings, we observe noteworthy differences between the simulation requirements for the polynomial and geometric drift cases, as well as between low and high-dimensional scenarios. These findings are summarised in Table \ref{table:batch_size_multi_dim} and Table \ref{table:batch_size_high_dim} respectively.\\

Termination criteria are usually defined as stopping times that follow a confidence set for the features of interest over the simulation trajectory. They allow termination of the simulation when a specified precision requirement is met. An alternative stopping criterion is to terminate the simulation when the effective sample size is sufficiently large. It is well-known that stopping according to the effective sample size is equivalent to termination rules based on confidence sets, see for example \cite{multivariate_output}. Our results enable us to study the convergence complexity of a broad range of termination criteria, and show that the termination rules introduced in \cite{glynn} and \cite{multivariate_output} can be applied to high-dimensional settings. These results also provide us with insights regarding the choice of termination rule. The analysis of these termination rules relies on the FCLT and consistency of variance estimators. Our results allow us to take the error of the Gaussian approximation and the convergence rate of the variance estimator into account, and thus generalise the analysis of aforementioned results. Finally, we obtain novel requirements for the minimum simulation threshold that guarantee the validity of variance estimation and termination criteria in high-dimensional settings. We give conditions that guarantee that the termination time of an MCMC algorithm scales polynomially in dimension while ensuring a desired level of precision. \\

%
%
%
%
%
%
%
%
%
%
%
%
%
%
%
%
%

This article is organised as follows. In Section \ref{section:2}, we review some preliminary results regarding MCMC. In Section \ref{section:3}, we give our obtained Gaussian approximation results. In Section \ref{section:4}, we apply our results to MCMC output analysis and provide guarantees for the termination time. In Section \ref{section:5},  we present the proofs of our results. \\

\textbf{Notation:} For $a,b \in \RR$, we denote  $a \wedge b =\min(a,b)$, $a \vee b =\max(a,b)$, and $a\lesssim b$ to denote an inequality up to a universal constant. For $x > 0$, we denote
$\log^{*}(x) = \log(x) \vee 1$. For a vector $x \in \RR^d$, we denote the Euclidean norm with $\abs{x}$. For a matrix $A$, we denote the Frobenius norm and the spectral norm with $\abs{A}$ and $\abs{A}_*$ respectively. Furthermore, the trace of matrix $A$ is denoted as $\tr(A)$ and let $\sigma_d(A)$ and $\sigma_1(A)$, denote the largest and smallest eigenvalue respectively.

\section{Preliminaries MCMC}
\label{section:2}
 We consider a Markov chain $X=(X_t)_{t\in \mathbb{N}}$ on $(E,\mathcal{E})$ where a $E\subseteq \mathbb{R}^N$ denotes the state-space and $\mathcal{E}$ the corresponding Borel $\sigma$-algebra, with $m$-step transition kernel $P^m(x,\cdot)$, defined by
\[
P^m(x,B)=\mathbb{P}(X_m \in B|X_0=x), \ \  m \in \mathbb{N},  x \in E, B\in \mathcal{E}.
\]
We say that $\pi$ is the stationary distribution of the Markov chain if
\[
\int_E P(x,B)\pi(dx)=\pi(B) \ \ \ \forall B\in \mathcal{E}.
\]

Drift and minorisation conditions are widely used for obtaining quantitative bounds for the mixing time of Markov chains. We follow the framework of \cite{qin2019convergence,yang2023complexity,zhou2022dimension} where a  drift-and-minorisation approach is used to obtain quantitative convergence
bounds in high-dimensional settings. The drift condition describes how fast the Markov chain moves towards subsets of the state space, while the minorisation condition controls how fast the Markov chain forgets its past. More specifically, the function $V$ in Drift Condition \ref{exp_drift_discrete_one} describes how fast the chain in expectation will move towards the set $C$ given that the chain is currently in state $x$ and how long the chain is expected to stay in this set $C$. An appropriate drift function should have low values in high-probability regions of the state space. Note that 
Drift Condition \ref{exp_drift_discrete_one} implies that while the chain is not in $C$, the value of the drift function will decrease geometrically. 
We say that the Markov chain satisfies a geometric drift condition if Drift Condition \ref{exp_drift_discrete_one} holds.
\begin{drift} 
\label{exp_drift_discrete_one}
Let there exist a function $V:E\rightarrow\R^+$,  constants $\lambda \in (0,1)$ and $0<b, \upsilon_C< \infty$  such that $\upsilon_C=\sup_{x\in C}V(x)$ and 
    \[P V(x)= \int_E V(y)P(x,dy) \leq \lambda V(x)  +b \mathbbm{1}_{C}(x),\]
    for some set $C \in \mathcal{E}$.
\end{drift}

In many applications, we can only guarantee that the drift function decays at a polynomial rate while the process is not in $C$. This corresponds to the following polynomial drift condition.
\begin{drift} 
\label{poly_drift_discrete_one}
Let there exist a function $V:E\rightarrow\R^+$,   constants $0<c,b,\upsilon_C< \infty$, such that $\upsilon_C=\sup_{x\in C}V(x)$ and $\eta \in (0,1)$ such that
\[
PV(x)\leq V(x)-cV(x)^\eta+b\mathbbm{1}_{C}(x),
\] 
for some set $C \in \mathcal{E}$ with $\pi(C)>0.$
\end{drift}

It will also be useful to consider the geometric and polynomial drift condition for the $m_0$-skeleton of the Markov chain.

\begin{drift} 
\label{exp_drift_discrete_skeleton}
Let there exist a function $V:E\rightarrow\R^+$,  constants $\lambda \in (0,1)$ and $0<b, \upsilon_C< \infty$  such that $\upsilon_C=\sup_{x\in C}V(x)$ and 
    \[P^{m_0} V(x)= \int_E V(y)P^{m_0}(x,dy) \leq \lambda V(x)  +b \mathbbm{1}_{C}(x).\]
\end{drift}

\begin{drift}
\label{poly_drift_discrete_skeleton}
Let there exist a function $V:E\rightarrow\R^+$,   constants $0<c,b,\upsilon_C< \infty$, such that $\upsilon_C=\sup_{x\in C}V(x)$ and $\eta \in (0,1)$ such that
\[
P^{m_0}V(x)\leq V(x)-cV(x)^\eta+b\mathbbm{1}_{C}(x).
\] 
\end{drift}
Note that we can obtain a drift condition for the $m_0$--skeleton by iterating the one-step drift condition. However, in order for the parameters of the drift condition to remain tractable, we give some sufficient conditions. 
 Propositions \ref{exp_condition_superset} and \ref{exp_drift_discrete_subset} and Propositions \ref{poly_condition_superset} and \ref{poly_drift_discrete_subset} give conditions such that the one-step drift condition implies the desired drift condition for the skeleton chain in case of the geometric and polynomial drift condition respectively.\\

We say that an associated local $m_0$-step minorisation condition holds for the Markov chain if the following holds.
\begin{minorisation} 
Let $\nu$ be some probability measure defined on $C$ such that
\label{minorisation}
\begin{equation*}
P^{m_0}(x,\cdot)\geq \alpha \mathbbm{1}_C(x) \nu (\cdot),
\label{minorization_skeleton}
\end{equation*}
the minorisation volume $\alpha \in (0,1]$, $m_0 \in \mathbbm{N}$ and small set $C$ with $\pi(C)>0$.
\end{minorisation}

It is known that a multi-step minorisation condition holds for all widely used MCMC algorithms, see for example \cite[Proposition 5.4.5]{meyn_tweedie_2012}. Often it can even be shown that the Markov chain satisfies a one-step minorisation condition. For a more detailed explanation regarding the interpretation of the drift and minorisation conditions and how they relate to one another, we refer to Section \ref{section:5} and the given references. In the framework of \cite{qin2019convergence,yang2023complexity,zhou2022dimension}, it is demonstrated that for the drift and minorisation conditions to behave well as the dimension of the state-space increases, it is useful to consider a family of drift functions. Let $\pi_N$ denote the stationary distribution of the Markov chain and let $N$ be the dimension of the corresponding state-space. We say that a family of non-negative functions, $\{V_N(\cdot)\}_{N\in \mathbb{N}}$, is a family of drift functions if for each $N$ they all satisfy a certain class of drift conditions. 
We say that a family of geometric drift conditions is stable if $\lambda:=\limsup_{N\rightarrow \infty} \lambda_N<1$ and $b:=\limsup_{N\rightarrow \infty} b_N<\infty$. Similarly, we say that a family of polynomial drift conditions is stable if $\eta:=\liminf_{N\rightarrow \infty} \eta_N>0$, $c:=\liminf_{N\rightarrow \infty} c_N>0$ and $b:=\limsup_{N\rightarrow \infty} b_N<\infty$. Furthermore, a minorisation condition is stable provided that $\alpha:=\liminf_{N\rightarrow \infty} \alpha_N>0$.
In high-dimensional settings, 
 the drift function can exhibit undesirable concentration behaviour, causing the small set $C$ to grow too fast, which in turn leads the minorisation volume $\alpha$ to tend to zero at an uncontrolled rate. Consequently, as shown in for example  \cite{rajaratnam2015mcmc}, many existing convergence bounds deteriorate as the dimension of the state space grows large.  In the remainder of the paper, we will often simply write $\pi$  and $V$ instead of  $\pi_N$ and $V_N$. However, it is important to keep in mind that when we study the behaviour of the chain as the state-space grows, $N\rightarrow \infty$, we are considering a family of Markov chains $\{X_N\}_{N\in \mathbb{N}}$ with a corresponding family of target distributions and drift and minorisation conditions. Furthermore, throughout the paper, we will consider all drift and minorisation conditions to be stable.


\section{Gaussian approximation for MCMC samplers}

\label{section:3}
We consider the following component-wise moment conditions on our features:

\begin{assumption}
\label{assumption_moment_condition}
    Let $f: E\rightarrow \mathbb{R}^d$ be a measurable function such that either of the following moment conditions holds
    \begin{enumerate}
        \item  $\sup_{i\in \{1,\cdots,d\}}\pi(\abs{f_i}^{p+\varepsilon})< \infty$ holds for given $p>2$ and some $\varepsilon \in (0,1/p] $,
        \item $\sup_{i\in \{1,\cdots,d\}}\pi( e^{tf_i})< \infty$ holds  for  $t$ in some neighbourhood of 0.
        \item $\sup_{i\in \{1,\cdots,d\}}\norm{f_i}_\infty< \infty$.
    \end{enumerate}

\end{assumption}

Most of our results will rely on Assumption \ref{assumption_moment_condition}.1, which is the most reasonable to assume in practice. Furthermore, we will impose the following regularity conditions on the covariance matrix of our MCMC sampler.

\begin{assumption}
\label{assumption_covariance}
Suppose that the smallest eigenvalues of  $$\Sigma_{f_T}=\Covu\left(\frac{1}{\sqrt{T}}\sum_{t=1}^T\{f(X_t)-\pi(f)\}\right)$$ for sufficiently large $T$ and $\Sigma_f:=\lim_{T\rightarrow \infty }\Sigma_{f_T}$ are larger than some constant $\sigma_0>0$. Furthermore, suppose that 
\[
\sup_{i\in \{1,\cdots,d\}}\sup_{j\in \{1,\cdots,d\}}\abs{\Sigma_{f_{ij}}} < \infty,
\]
where $\Sigma_{f_{ij}}$ denotes the $(i,j)$--th entry of the matrix $\Sigma_f$ for $1\leq i,j \leq d.$
\end{assumption}

Assumption \ref{assumption_covariance} guarantees that the asymptotic covariance matrix is well defined, and that both the empirical and asymptotic covariance matrix are non-singular. Firstly, we state our obtained weak Gaussian approximation results for the one-step minorisation case under both a geometric and polynomial drift condition.

\begin{theorem}
\label{theorem_approx_discrete_one}
Let $(X_t)_{t \in \mathbb{N}}$ be an irreducible aperiodic stationary Markov chain, assume that a one-step minorisation condition holds and that drift condition \ref{exp_drift_discrete_one} is satisfied. Then for all functions $f$  for which Assumptions \ref{assumption_moment_condition} and \ref{assumption_covariance} hold we can, on an enriched probability space, define a process that is equal in law to $X$ and a standard $d$-dimensional Brownian motion $W$  such that for any $\theta_0>0$ we have that
\begin{equation}
\label{main_sip_discrete_one_exp}
\abs{\sum_{t=1}^T f(X_t)-T\pi(f)-\Sigma_f^{1/2}W_T} =   \left\{
            \begin{array}{lll}                                          
                 \displaystyle \bigO_{P} \left(     d^{} \bar{\psi}_{N}^{1/p}  T^{1/p}\log(T)\right)  \\
                                   \textcolor{white}{.}   \\ 
              \displaystyle   \bigO_{a.s.} \left(    d^{25/4+\theta_0}\log^* (d)\bar{\psi}_{N}^{2/p} \left(\frac{\sigma_d}{\sigma_0}\right)^{1/2}T^{1/p}\log(T)\right)
                 
                \end{array},   
              \right.   
\end{equation}

where

\begin{equation}
\label{C_one_step_exp}
    \bar{\psi}_{N}:=\psi(\alpha,\lambda,b)=   \alpha^{-1}\left(\frac{b}{\alpha(1-\lambda)} \right)^{1+\varepsilon/p} \left(\frac{p}{\ln(1/\lambda)e}\right)^p \sup_{i\in \{1,\dots,d\}}\pi(\abs{f_i}^{p+\varepsilon})
\end{equation}

If we assume that drift condition \ref{poly_drift_discrete_one} holds instead of \ref{exp_drift_discrete_one}, then we have that 

\begin{equation}
\label{main_sip_discrete_one_poly}
\abs{\sum_{t=1}^T f(X_t)-T\pi(f)-\Sigma_f^{1/2}W_T} =   \left\{
            \begin{array}{lll}                                          
                 \displaystyle \bigO_{P} \left(     d^{} \tilde{\psi}_{N}^{1/p_0}  T^{1/p_0}\log(T)\right)  \\
                                   \textcolor{white}{.}   \\ 
              \displaystyle   \bigO_{a.s.} \left(    d^{25/4+\theta_0}\log^* (d)\tilde{\psi}_{N}^{2/p} \left(\frac{\sigma_d}{\sigma_0}\right)^{1/2}T^{1/p_0}\log(T)\right)
                 
                \end{array},   
              \right.   
\end{equation}

where 

         \begin{equation}
         \label{C_one_step_poly}
    \tilde{\psi}_N:=\tilde{\psi}_N(\alpha,b,c,\upsilon_C)=                    
                  \alpha^{-1}\left(1+\frac{b }{c \alpha}+\frac{\upsilon_c-c+b}{1-\alpha} \right)^{1+\varepsilon/p_0} \sup_{i\in \{1,\dots,d\}}\pi(\abs{f_i}^{p+\varepsilon})        
         \end{equation}
 and
      \begin{equation}
  \label{p_0definition1}
    p_0=   \left\{
                \begin{array}{ll}
                  \frac{pq(\eta)}{p+q(\eta)+\varepsilon},   & \  \textrm{if } \frac{2p}{3p-2} < \eta \leq p(p+\varepsilon)/(p(p+\varepsilon)+\varepsilon) ,\\
                 
                  p,   & \  \textrm{if } \eta > p(p+\varepsilon)/(p(p+\varepsilon)+\varepsilon) ,\\

                  q(\eta)-\bar{\epsilon},   & \  \textrm{if } \eta> 1/2 \  \textrm{  and A\ref{assumption_moment_condition}.2 holds},

                \end{array}
              \right.
      \end{equation}
with $q(\eta)=\eta/(1-\eta)$  and the         
  entries of $\Sigma_f$ are given by 
  \begin{equation}
  \label{covvie}
 \Sigma_f= \Varu(f(X_0))+ \sum_{k=1}^\infty \Covu(f(X_0),f(X_k))+  \sum_{k=1}^\infty \Covu(f(X_k),f(X_0)).
  \end{equation}
\end{theorem}

In Theorem \ref{theorem_approx_discrete_multi}, we formulate our obtained Gaussian approximation results for the multi-step minorisation case under both a geometric and polynomial drift condition.

\begin{theorem}
\label{theorem_approx_discrete_multi}
Let $(X_t)_{t \in \mathbb{N}}$ be an irreducible aperiodic stationary Markov chain, assume that an $m_0$-step minorisation condition holds and that drift condition \ref{exp_drift_discrete_skeleton} is satisfied. Then for all functions $f$  for which Assumptions \ref{assumption_moment_condition}.1 and \ref{assumption_covariance} hold we can, on an enriched probability space, define a process that is equal in law to $X$ and a standard $d$-dimensional Brownian motion $W$ such that for any $\theta_0>0$ we have that
\begin{equation}
\label{main_sip_discrete_multi_exp}
\abs{\sum_{t=1}^T f(X_t)-T\pi(f)-\Sigma_f^{1/2}W_T} =   \left\{
            \begin{array}{lll}                      \displaystyle \bigO_{P}\left(  d^{} \bar{\psi}_{N}^{2/p} T^{\frac{1}{4}+\frac{1}{4(p-1)}} \right)    \\
                                   \textcolor{white}{.}   \\ 
              \displaystyle  \bigO_{a.s.}\left(  d^{25/4+\theta_0}\log^* (d)\bar{\psi}_{N}^{2/p} \left(\frac{\sigma_d}{\sigma_0}\right)^{1/2}T^{\frac{1}{4}+\frac{1}{4(p-1)}} \right)  
                \end{array},
              \right.   
         \end{equation}

where  
         \[
  \bar{\psi}_N:=\psi(\alpha,\lambda,b,m_0)=  \alpha^{-1}\left(\frac{b m_0 }{\alpha(1-\lambda)}\right)^{\varepsilon/p}\sup_{i\in \{1,\dots,d\}}\pi(\abs{f_i}^{p_0+\varepsilon})  
  \]
  

                 

If we assume that drift condition \ref{poly_drift_discrete_skeleton} holds instead of \ref{exp_drift_discrete_skeleton}, then we have that 

\begin{equation}
\label{main_sip_discrete_multi_poly}
\abs{\sum_{t=1}^T f(X_t)-T\pi(f)-\Sigma_f^{1/2}W_T} =       \left\{
            \begin{array}{lll}                              %
                 \displaystyle         \bigO_{P}\left( d^{} \tilde{\psi}_{N}^{2/p} \left(\frac{\sigma_d}{\sigma_0}\right) T^{\frac{1}{4}+\frac{1}{4(p_0-1)}} \right),   \\
                                   \textcolor{white}{.}   \\ 
              \displaystyle         \bigO_{a.s.}\left( d^{25/4+\theta_0}\log^* (d)\tilde{\psi}_{N}^{2/p} \left(\frac{\sigma_d}{\sigma_0}\right)^{1/2} T^{\frac{1}{4}+\frac{1}{4(p_0-1)}} \right),  
                 
                \end{array},
              \right.   
         \end{equation}

where  
         \[
    \tilde{\psi}_N:=\tilde{\psi}(\alpha,b,c,\upsilon_c,m_0)=               
            \alpha^{-1}m_0^{q(\eta)/{p_0^2}}\left(1+\frac{b }{c \alpha}+\frac{\upsilon_c-c+b}{1-\alpha} \right)^{({p-p_0+\varepsilon})/{p}  } \sup_{i\in \{1,\dots,d\}}\pi(\abs{f_i}^{p_0+\varepsilon})    
              \]
 
       \[
    p_0=   \left\{
                \begin{array}{ll}
                  \frac{pq(\eta)}{p+q(\eta)+\varepsilon},   & \  \textrm{if } \frac{2p}{3p-2} < \eta \leq p(p+\varepsilon)/(p(p+\varepsilon)+\varepsilon) ,\\
                 
                  p,   & \  \textrm{if } \eta > p(p+\varepsilon)/(p(p+\varepsilon)+\varepsilon) ,\\

                  q(\eta)-\bar{\epsilon},   & \  \textrm{if } \eta> 1/2 \  \textrm{  and A\ref{assumption_moment_condition}.2 holds},

                \end{array}
              \right.
    \]  
            with $q(\eta)=\eta/(1-\eta)$  and the         
  entries of $\Sigma_f$ are given by 
 \begin{equation}
  \label{covvie2}
 \Sigma_f= \Varu(f(X_0))+ \sum_{k=1}^\infty \Covu(f(X_0),f(X_k))+  \sum_{k=1}^\infty \Covu(f(X_k),f(X_0)).
  \end{equation}
\end{theorem}

\begin{remark}
\label{remark_growth}Firstly, our obtained weak Gaussian approximation results given in Theorems \ref{theorem_approx_discrete_one} and \ref{theorem_approx_discrete_multi} guarantee a central limit theorem for Markov chains, provided that the dimension of the feature space grows at most as described in Table \ref{table:table_dim_growth_true} below for any $\bar{\varepsilon}>0$.\\

\begin{table}[ht]
\begin{center}
\begin{tabular}{|c|c|c|}
    \hline
    \textbf{} & \textbf{one-step minorisation} & \textbf{multi-step minorisation} \\
    \hline
    \rule{0pt}{18pt}
        geometric drift & $\displaystyle o\left(T^{\frac{p-2}{2p}-\bar{\varepsilon}}\right)$ & $\displaystyle o\left(T^{\frac{p-2}{4(p-1)}-\bar{\varepsilon}}\right) $ \\
   \hline
       \rule{0pt}{18pt}
    polynomial drift & $\displaystyle o\left(T^{\frac{p_0-2}{2p_0}-\bar{\varepsilon}}\right)$ & $\displaystyle o\left(T^{\frac{p_0-2}{4(p_0-1)}-\bar{\varepsilon}}\right) $ \\
   \hline 
\end{tabular}
\caption{Growth rate $d$ such that CLT holds}
\label{table:table_dim_growth_true}
\end{center}
\end{table}

It is important to note that the growth rates given in Table \ref{table:table_dim_growth_true} are the general case. However, in most applications of MCMC, there are more structural properties present such that the growth rate of the dimension can be greatly improved. From the proofs of Theorems \ref{theorem_approx_discrete_one} and \ref{theorem_approx_discrete_multi} it follows that the dependence of the dimension of the feature space is equal to \[ \sqrt{d} \pi(\abs{f}^{p_0})^{1/{p_0}}.\] In general, we will have that $\pi(\abs{f}^{p_0})^{1/{p_0}}\leq \sqrt{d} \sup_{i}\pi(\abs{f_i}^{p_0})^{1/{p_0}}$, which gives us $\psi_d=d$.
However, in many applications of MCMC this term can be smaller. For example, in Bayesian statistics the prior is often used to induce sparsity in the posterior or to provide regularisation, which will reduce $\pi(\abs{f}^p)$ and, in the most favourable case, allow growth rates of $d=o(T)$ and $d=o(\sqrt{T})$ for the one-step and multi-step minorisation respectively. Similarly, the dimension dependence of the strong approximation can be improved for specific settings. Moreover, under Assumption \ref{assumption_moment_condition}.3, it can be shown that the rate of Theorem \ref{theorem_approx_discrete_one} improves to $\bigO_P\left(d \log^2(\tilde{\psi}_{N})\log^2(T)\right) $ with $\tilde{\psi}_{N}=\alpha^{-1}b/ (\alpha \lambda(1-\lambda))$ for the weak approximation case.\\

Finally, in order for a central limit theorem to hold in the high-dimensional setting, we see that the simulation time of our MCMC algorithm should scale with the dimension of the state-space and feature space as detailed in Table \ref{table:table_sim_growth_true} below, for any $\bar{\varepsilon}>0$.

\begin{table}[ht]
\begin{center}
\begin{tabular}{|c|c|c|}
    \hline
    \textbf{} & \textbf{one-step minorisation} & \textbf{multi-step minorisation} \\
    \hline
    \rule{0pt}{18pt}
        geometric drift & $\displaystyle \Omega\left(d^{\frac{2p}{p-2}+\bar{\varepsilon}}\  \bar{\psi}_N^{4/(p-2)+\bar{\varepsilon}}\right)$ & $\displaystyle \Omega\left(d^{\frac{4(p-1)}{p-2}+\bar{\varepsilon}} \ \bar{\psi}_N^{8(p-1)/(p^2-2p)+\bar{\varepsilon}}\right) $ \\
   \hline
       \rule{0pt}{18pt}
    polynomial drift & $\displaystyle \Omega\left(d^{\frac{2p_0}{p_0-2}+\bar{\varepsilon}}\ \tilde{\psi}_N^{4/(p_0-2)+\bar{\varepsilon}}\right)$ & $\displaystyle \Omega\left(d^{\frac{4(p_0-1)}{p_0-2}+\bar{\varepsilon}} \ \tilde{\psi}_N^{8(p_0-1)/(p_0^2-2p_0)+\bar{\varepsilon}}\right) $ \\
   \hline 
\end{tabular}
\caption{Simulation time $T$ such that CLT holds}
\label{table:table_sim_growth_true}
\end{center}
\end{table}
\vspace{-\baselineskip}
Hence we see that for large $p$ the simulation time of our sampling algorithm should scale as $\psi_d^2$ and $\psi_d^4$ in the one-step and multistep minorisation case respectively.  We note that the simulation requirements for consistent estimation of the asymptotic variance and obtaining precision guarantees, are given by the results in Section \ref{section:4}. 


\end{remark}

\begin{remark}
\label{discussion_GA}
 Theorem \ref{theorem_approx_discrete_one} gives the first Gaussian approximation results that attain the optimal Koml\'os--Major--Tusn\'ady approximation rate, up to a logarithmic factor, for MCMC samplers in the multivariate setting. We note that these results are applicable to the high-dimensional settings considered in \cite{qin2019convergence,yang2023complexity} and \cite{zhou2022dimension}, where an asymptotically stable one-step minorisation condition and family of exponential drift conditions are shown to hold. 
 Moreover, we note that Theorem \ref{theorem_approx_discrete_multi} can also be proven for continuous-time processes under a petite set condition and similar drift conditions.
Furthermore, we note that all obtained weak approximation results can also be formulated as high probability deviation bounds.\\
\forceindent  
In \cite{gouezel2010almost} a strong Gaussian approximation is obtained with rate $o_{a.s.}(T^\theta)$ for any $\theta > {\frac{1}{4}+\frac{1}{4(p-1)}}$, under assumptions on the dependence decay of the process through its spectral properties. These conditions on the characteristic function could be challenging to verify in MCMC settings. While the obtained approximation error $\Psi_T$ with respect to the sampling time is independent of dimension, contrary to previously obtained results mentioned in the paper, the dimensionality of both the state-space and the feature-space can still influence the overall approximation error. \cite{lu2022} gives a strong Gaussian approximation with rate $o_{a.s.}\left(T^{1/3+2/3(3p-2)}\right)$ for Hilbert space 
valued stochastic processes whose dependence is controlled through exponentially decaying $\beta$--coefficients. The results of \cite{lu2022} are therefore applicable for Markov chains satisfying a geometric drift condition. Applying their Gaussian approximation result to $\mathbb{R}^d$ results in the approximation rate of \cite{gouezel2010almost} with a multiplicative factor $d^8 \log (d)$  introduced to the approximation error. However, since their result requires exponential decay of $\beta$--mixing coefficients, they are not applicable to the polynomial drift condition case.\\
 %
 %
%
%
%
\end{remark}

\begin{remark}
Contrary to previously obtained Gaussian approximation results for MCMC, see \cite{ flegal_bm,banerjee2022,li2024multivariate,merlevede2015,pengel2024strong,multivariate_consistency}, in the high-dimensional setting, Gaussian approximation cannot be expected to hold for arbitrary initial distributions. It is well known that regularity conditions on the initial state are required in high-dimensional settings, since for example a cold initialisation can easily lead to mixing time bounds that are exponentially increasing in dimension, as demonstrated by \cite{bandeira2022free}.
 The harmonic function argument of \cite[ Proposition 17.1.6]{meyn_tweedie_2012} that is often used to generalise these types of results to an arbitrary initial distribution, is not applicable to the high-dimensional setting.  Despite this, stationarity is not required for any of our obtained results to hold. By the Comparison theorem, \cite[Theorem 14.2.2]{meyn_tweedie_2012}, it can be shown that our result can be formulated for any initial condition $x$, as detailed in Remark \ref{initial_remark_proof}. However, the drift at the initial value, i.e., $V(x)$ would enter the obtained bounds, see  \eqref{exp_initial} and \eqref{poly_initial} for the geometric and polynomial drift case respectively. Therefore some regularity conditions, which would be similar in nature to a warm start condition, would be required to ensure that the initialisation does not dominate our obtained dimension-dependent bounds.
\end{remark}

\begin{remark}
In our approach, we control the decay of the dependence of our process through both the drift and minorisation conditions.
In alignment with expectations, we see that if a polynomial drift condition holds the dimension of the state-space introduces a larger penalty in our approximation rate when compared to the case where a geometric drift condition holds. Note that for higher rates of polynomial drift, the approximation rate $T^{1/p_0}$ tends to the rate with a geometric drift. 

\forceindent 
We note that for our results to be applicable in high-dimensional settings, we require the drift conditions to be asymptotically stable. Additionally, we see that in contrast to the low-dimensional case where the minorisation volume is negligible, see for example \cite[Remark 2]{flegal_bm}, in the high-dimensional case it plays a vital role since an improper minorisation lower bound can cause to the Gaussian approximation error rate to increase exponentially in dimension. Hence approaches like the ones given in \cite{qin2019convergence,yang2023complexity} are critical to ensure asymptotically stable drift and minorisation conditions.

\forceindent  Furthermore, we see that the spectral condition number, which is defined as the quotient of the largest and the smallest eigenvalue, of the asymptotic covariance matrix appears in our strong approximation errors.  It is known that better dimension dependence in Gaussian approximations can be obtained for isotropic targets, see for example \cite{fathi2019stein}. From the specification of the asymptotic covariance matrix given in \eqref{covvie}, we see that if the distribution of the features under the target measure is isotropic, then the spectral condition number of the first term $\Var_\pi(f(X_0))$ will be smaller. Furthermore, if the auto-covariance of the process decays faster then the second and third terms of \eqref{covvie} will be smaller. By \cite[Theorem 1.1]{Rio} and \cite[Theorem F.3.3]{douc2018markov}, we see that we can control the decay of the auto-covariance through the drift conditions. Consequently, for the geometric drift case, we can guarantee the condition number of the asymptotic covariance matrix to be smaller than the polynomial case. 
\end{remark}

%
%
%
%
\section{High-dimensional MCMC Output Analysis}
\label{section:4}
 In order to assess the accuracy of our sampling method, we need to estimate the asymptotic variance appearing in the central limit theorem. In the high-dimensional setting this gives us additional requirements for the simulation time of our algorithms. Estimation of the asymptotic auto-covariance matrix $\Sigma_f$ plays a central role in MCMC output analysis, specifically for computing many convergence diagnostics, and implementing termination criteria.
Through our obtained Gaussian approximation results, we are able to extend the results on variance estimation and termination criteria for MCMC algorithms. We illustrate the applications of our obtained Gaussian approximation results to the batch means method considered in \cite{multivariate_output}. We use our results to adapt the tuning parameters of the considered variance estimation methods to take slower convergence rates into account due to polynomial drift conditions or high dimensionality. Finally, our quantitative convergence bounds for termination criteria allow us to analyse the influence of the ergodicity and dimensionality on the appropriate running time of our MCMC algorithms.

\subsection{Estimation of the Asymptotic Variance}
We first consider the multivariate batch estimator since it enjoys computational advantages over other variance estimation methods in MCMC settings. The batch means method divides the simulation output into $k_T$ batches of length $\ell_T$ such that $k_T=\floor{T/\ell_T}$. The batch means estimator is then given by
\begin{equation}
\label{multi_bm_est}
\hat{\Sigma}^{BM}_T=\frac{\ell_T}{k_T-1}\sum_{i=1}^{k_T}\left(\bar{Z}_i(\ell_T)-\frac{1}{k_T}\sum_{i=1}^{k_T}\bar{Z}_i(\ell_T)\right)\left(\bar{Z}_i(\ell_T)-\frac{1}{k_T}\sum_{i=1}^{k_T}\bar{Z}_i(\ell_T)\right)^T,
\end{equation}
\noindent
where $\bar{Z}_i(\ell_T)$ denotes the sample average of each obtained batch, i.e.,
\begin{equation}
\label{batch_def}
\bar{Z}_i(\ell_T):=\frac{1}{\ell_T}\sum_{s=(i-1)\ell_T}^{i\ell_T}\hspace{-0.2cm}f(X_s), \quad i=1,\dots,k_T.
\end{equation}

We impose the following conditions on the batch size.
\begin{assumption}
\label{assumption_truncation}
Let $\ell_T$ be an integer sequence such that $\ell_T \to \infty$ and $n/\ell_T \to \infty$ as $n \to \infty$ where $\ell_T$ and $n/\ell_T$ are non-decreasing. Moreover, assume that there exists a constant $c \geq 2$ such that $\sum_T (\ell_T/T)^c < \infty$, $(\ell_T/T) \log (T)=o(1)$, $\ell_T^{-1} \log T = o(1)$, and $T > 2\ell_T$.
\end{assumption}
Applying our Gaussian approximations to the results on the batch means estimator of \cite{multivariate_output} gives us the following theorem.
\begin{theorem}
\label{multi_bm_theorem}
Suppose that $f:\mathbb{R}^N\rightarrow\mathbb{R}^d$, with $\sup_{i  \in \{1,\cdots,d\}}\pi
(\abs{f_i}^{p+\varepsilon})< \infty$ for some $p\geq4$ and let $X$ satisfies a weak Gaussian approximation with approximation error $\bar{\psi}_N\psi_d\Psi_T \log(T)$ with $\psi_d=d^a$ for some $a>0$. Assume that Assumption \ref{assumption_truncation} holds, and that  
\begin{equation}
\label{as_multi_bm}\frac{\bar{\psi}_N d\psi_d\Psi_T\log(T)}{\ell_{T,d}^{1/2}} =o(1)\ \textrm{and }
    \frac{\bar{\psi}^2_N d\psi^2_d\Psi^2_T\log(T)}{T}=o(1),
\end{equation}

then we have that $\abs{\hat{\Sigma}^{BM}_T-\Sigma_f}\rightarrow 0$ with probability 1 as $T\rightarrow \infty$. Moreover, if we assume that 
  \begin{equation}
      \Psi_T=   \left\{
                \begin{array}{ll}
                      \Psi_T^{(1)}:= T^{1/p_0} \log(T),   \\
                   \Psi_T^{(2)}:= T^{1/4+1/4(p_0-1)} \log(T) ,             
                \end{array}
              \right.
  \end{equation}

 for some $p_0>2$ and we choose the simulation time 
 
     \begin{equation}
      T=   \left\{
                \begin{array}{lll}
                      \Omega \left(\left({\bar{\psi}_N {d}\psi_d }\right)^{\frac{2p_0}{(p_0-2)}(1+\bar{\delta})} \right), &   \textrm{under rate }\Psi_T^{(1)} \,  \\
                    \vspace{-0.2cm}  \\
                   \Omega \left( \left(\psi_N d^{1/4} \psi_d \right)^{\frac{p_0-1}{p_0-2}4(1+\bar{\delta})} \right), \,    &\textrm{under rate } \Psi_T^{(2)},                 
                \end{array}
              \right.
        \end{equation}
 
for any $\bar{\delta} > 1/(1+a)$ then the choice of batch size $\ell_T \asymp d^{-(p_0-2)/(2p_0(1+\bar{\delta}))}\floor{T^{\alpha}}$ with
  \begin{equation}
  \label{optimal_bs}
      \alpha=  \left\{
                \begin{array}{ll}
                       \frac{1}{2}+\frac{p_0-2}{2p_0(1+\bar{\delta})}+ \frac{1}{p_0}, &   \textrm{under rate }\Psi_T^{(1)} \,  \\
                    {\frac{3}{4}+ \frac{1}{4(p_0-1)}+\frac{(p_0-2)}{4(p_0-1)(1+\bar{\delta})}}, \,    &\textrm{under rate } \Psi_T^{(2)},                 
                \end{array}
              \right.
        \end{equation}

optimises the given convergence rate for $T \rightarrow \infty$.
\end{theorem}

\begin{remark}
While the result of Theorem \ref{multi_bm_theorem} is formulated in a high-dimensional setting, note that in an application where we can assume the influence of the dimension to be negligible, we obtain as an immediate consequence the following the choice of bath size from our strong approximation results.

\begin{table}[ht]
\begin{center}
\begin{tabular}{|c|c|c|}
    \hline
    \textbf{} & \textbf{one-step minorisation} & \textbf{multi-step minorisation} \\
    \hline
    \rule{0pt}{16pt}
        exponential drift & $ \displaystyle T^{\frac{1}{2}+ \frac{1}{p_0}}$ & $ \displaystyle T^{\frac{3}{4}+ \frac{1}{4(p_0-1)}}$\\
   \hline
   \rule{0pt}{16pt}
     polynomial drift & $ \displaystyle T^{\frac{1}{2}+ \frac{1}{p_0}}$ & $ \displaystyle T^{\frac{3}{4}+ \frac{1}{4(p-1)}}$ \\
   \hline

\end{tabular}
\caption{Batch size $\ell_T$ multivariate setting}
\label{table:batch_size_multi_dim}
\end{center}
\end{table}
\end{remark}
%
%
%
It has empirically been observed that for many practical problems, where slower convergence rates to stationarity are expected, larger batch sizes and truncation windows are required when applying the batch means and spectral variance methods for MCMC simulation output. Consistency of the batch means estimator requires that each batch gives an accurate representation of the dependence structure of the process. Naturally, in situations with slower convergence rates, a larger batch size will be required. An immediate consequence of slower mixing, is the slower decay of the autocovariance function. Hence also spectral variance estimators will require larger truncation points for consistent estimation of the asymptotic variance. But these corrections for either polynomial convergence rates to stationarity or the dimension of the problem have been done in heuristic ways. In Table \ref{table:batch_size_high_dim}, we give the appropriate batch sizes for the batch means estimator, which guarantee consistency in the polynomial drift as well as in the high-dimensional setting. 
\begin{table}[ht]
\begin{center}
\begin{tabular}{|c|c|c|}
    \hline
    \textbf{} & \textbf{one-step minorisation} & \textbf{multi-step minorisation} \\
    \hline
       \rule{0pt}{16pt}
        exponential drift & $\displaystyle T^{\frac{1}{2}+\frac{p-2}{2p(1+\bar{\delta})}+ \frac{1}{p}}$ & $ \displaystyle T^{\frac{3}{4}+ \frac{1}{4(p-1)}+\frac{(p-2)}{4(p-1)(1+\bar{\delta})}}$\\
   \hline
     \rule{0pt}{16pt}
     polynomial drift & $\displaystyle T^{\frac{1}{2}+\frac{p_0-2}{2p_0(1+\bar{\delta})}+ \frac{1}{p_0}}$ & $ \displaystyle T^{\frac{3}{4}+ \frac{1}{4(p_0-1)}+\frac{(p_0-2)}{4(p_0-1)(1+\bar{\delta})}}$ \\
   \hline
\end{tabular}
\caption{Batch size $\ell_T$ high-dimensional setting}
\label{table:batch_size_high_dim}
\end{center}
\end{table}

Note that our obtained results summarised in Table \ref{table:batch_size_high_dim} differ from the well-known results of \cite{chien1997large,goldsman1990}, and \cite{song1995optimal} who obtain an optimal batch size of order ${T}^{1/3}.$ This is due to the fact that the aforementioned results choose the batch size such that the mean squared error of the batch means estimator is minimised. For the analysis of termination criteria we require strong convergence of the corresponding variance estimator. Hence we have chosen the batch size that optimises the rate of strong convergence. We note that for the one-step minorisation case, Theorem \ref{multi_bm_theorem} also allows the MSE optimal batch size. However, this will result in a weaker bound on the termination time. In \cite{damerdji1995} it is shown that the Gaussian approximation results can also be used to obtain mean square convergence of the batch means estimator. However, this approach yields slower convergence rates than directly considering the MSE. The results of \cite{chien1997large,goldsman1990}, and \cite{song1995optimal} give a bias of order $\bigO(\ell_T^{-1})$ for the batch means estimator, whereas Theorem \ref{multi_bm_theorem} implies an upper bound for the bias of order $\bigO(\Psi_T\ell_T^{-1/2})$. Both approaches give the same bound on the variance of the batch means estimator. Our results rely on the fact that in applications of MCMC the run-length of the algorithm can always the chosen such that $T$ is of greater magnitude than the dimension of the problem. However, applying techniques from high-dimensional covariance matrix estimation might lead to less restrictive conditions on the simulation time. We leave it as a topic of future research to improve  the strong convergence rate of estimators of the MCMC asymptotic variance.\\

Furthermore, consistent estimation of the asymptotic variance is required for computing the effective sample size (ESS) of an MCMC algorithm. The ESS is a measure of the efficiency of an MCMC algorithm, incorporating
the difference in magnitude between the MCMC standard error and the in-
dependent Monte Carlo error. Following \cite{multivariate_output}, the ESS in the multivariate setting can be defined as
\begin{equation}
    \label{ess2}    \textrm{ESS}=T \left(\frac{\lvert\hat{\Gamma}_f\rvert}{\lvert\hat{\Sigma}_f\rvert}\right)^{1/d},
\end{equation}

\noindent where $\hat{\Sigma}_f$ denotes an estimator of $\Sigma_f$ and $\hat{\Gamma}_f$ denotes the sample covariance of the target, namely,
\begin{equation}
   \hat{\Gamma}_f= \frac{1}{T}\sum_{t=1}^T\left(f(X_t)-\hat{\pi}_T(f)\right)(f(X_t)-\hat{\pi}_T(f))^{\top}.
\end{equation}

From Theorem \ref{multi_bm_theorem} we
immediately obtain the following result for consistency of the ESS in the high-dimensional setting.

\begin{corollary}
\label{ESS_corr}
    Suppose the conditions of Theorem \ref{multi_bm_theorem} hold and that we estimate the ESS using the batch-means estimator. Then
    \begin{equation*}
\frac{\bar{\psi}_N d^{3/2}\psi_d({\sigma_d} / {\sigma_0})\Psi_T\log(T)}{\ell_{T,d}^{1/2}} =o(1)\ \textrm{and }
    \frac{\bar{\psi}^2_N d^{3/2}\psi^2_d({\sigma_d} / {\sigma_0})\Psi^2_T\log(T)}{T}=o(1),
\end{equation*}

then we have that 
    \[
    \abs{\left(\frac{\lvert\hat{\Gamma}_f\rvert}{\lvert\hat{\Sigma}^{BM}_T\rvert}\right)^{1/d}- \left(\frac{\lvert{\Gamma}_f\rvert}{\lvert{\Sigma}_f\rvert}\right)^{1/d}  } \rightarrow 0
    \]
     with probability 1 as $T\rightarrow \infty$.
\end{corollary}

We see that for consistent estimation of the ESS we require more stringent conditions on the convergence rate of $\hat{\Sigma}_T$.

\subsection{ MCMC Termination Criteria}
 Sequential termination rules are the standard practice for determining the appropriate running time of an MCMC algorithm. The Fixed Volume Stopping Rule (FVSR) allows termination of the simulation when the volume of a confidence region for the parameter of interest is below some predetermined tolerance level. Firstly, note that due to the obtained weak Gaussian approximation results, we can construct a confidence interval for $\pi(f)$, namely,
\begin{equation}
\label{CI}
C(T)=\left\{x \in \mathbb{R}^d: T(\hat{\pi}_T(f)-x)^\top \hat{\Sigma}^{-1}_T(\hat{\pi}_T(f)-x)< q_\alpha \right\},
\end{equation}
where $\hat{\pi}_T(f)$ denotes the empirical average of $f$ over the simulation output and $\hat{\Sigma}_T$ denotes some estimator of the asymptotic covariance matrix, which is evaluated using simulation output until time $T$, and  $q_\alpha$ denotes the $(1-\alpha)$--quantile of the $\chi^2$ distribution with $d$ degrees of freedom.  Given some user-specified tolerance level $\varepsilon$, the FVSR defines the time of termination $T(\varepsilon)$ for our simulation experiment as 
\begin{equation}
\label{FVSR}
    T(\varepsilon)=\inf\{t> 0: \Vollie(C(t))^{1/d}+ \Lambda(t)\leq \varepsilon \}.
\end{equation}
Here $\Vollie(\cdot)$ denotes the standard volume element and $\Lambda(t)$ is some positive sequence tending to zero. The role of $\Lambda(t)$ is to prevent early termination due to an inaccurate estimate of the covariance matrix or unreliability of the CLT from an insufficient sample size. A common choice is $\Lambda(t)=\mathbbm{1}_{\{t<T^*\}}+o(t^{-1/2})$, for some appropriate threshold $T^*$.\\

All termination criteria, see for example \cite{glynn, jones_fixed, multivariate_output} and \cite{gong2016}, make use of some sort of minimum simulation threshold in order to prevent early termination due to an inaccurate estimate of the covariance matrix or unreliability of the FCLT due to the insufficient sample size. It is often mentioned that this minimum simulation effort should take the complexity of the problem into account. However, the choice of this simulation threshold has always been done in heuristic ways. Our approach for determining the simulation threshold guarantees that the approximation rate between the estimated confidence ellipsoid and its limiting quadratic form is of a smaller asymptotic magnitude than the desired precision level $\varepsilon$. This results in both the error of the Gaussian approximation and the covariance matrix estimation procedure being of a smaller magnitude than the desired precision level. In the following theorem, we generalise \cite[Theorem 1]{glynn} by giving quantitative convergence bounds for the FVSR.

\begin{theorem}
\label{FVSR_optimal_rate}
Suppose that $X$ satisfies the following strong Gaussian approximation
\begin{equation*}
\abs{\sum_{t=1}^T f(X_t)-T\pi(f)-\Sigma_f^{1/2}W_T} = \bigO_{a.s.} \left( \bar{\psi}_N \psi_d\Psi_T\right),
\end{equation*}

with approximation error $\Psi_T=T^{1/{p_0}}\log (T)$ for some $p_0>4$ and $\psi_d=d^a$ for some $a>0$.
Let $T_1(\varepsilon)$ be given by
\begin{equation}
     T_1(\varepsilon)=\inf\{t> 0: \Vollie(C(t))^{1/d}+\varepsilon \Lambda(t) < \varepsilon \},
\end{equation}
with $C(t)$ the confidence ellipsoid given in \eqref{CI} and $\Lambda(t)=\mathbbm{1}_{\{t<T^* (\varepsilon,d,N)\}}+t^{-1}$, with
\begin{equation}
 \label{T_star_conditions_thrm}
 T^* (\varepsilon,d,N) = 
                       \left({\bar{\psi}_N \left(\frac{\tr(\Sigma_f)}{\sigma_0}\right)^{2}}d^{3}\psi_d \right)^{\frac{2p_0}{(p_0-2)}(1+\bar{\delta}_1)} \left(\frac{1}{\varepsilon}\right)^{\frac{4p_0}{(p_0-2)}(1+\bar{\delta}_2)}\vee e^{\frac{10p_0}{p_0-2}}
 \end{equation}

for any $\bar{\delta}_1 > 3/(3+a)$ and $\bar{\delta}_2 > 0 $. Let $\widehat{\Sigma}_T$ in \eqref{CI} denote the batch means estimator defined in \eqref{multi_bm_est}, with batch size $\ell_T$ set as
\begin{align}
\label{optimal_batch_size}
\ell_T=   
 \bar{\psi}_N \psi_d  T^{1/2+1/p_0}.
 \end{align}
Suppose that Assumptions \ref{assumption_moment_condition} and \ref{assumption_covariance} hold.
Then we have as $\varepsilon \downarrow 0$ the following:
\begin{enumerate}
    \item The asymptotic behaviour of the termination time $T_1(\varepsilon)$ is characterised by 
\begin{equation}
\frac{\varepsilon^2T_1(\varepsilon)}{{c_{\alpha,d}^{2/d}\det(\Sigma_f)}^{1/d}} = 1+ o_{a.s.}\left( \log^2(\bar{\psi}_N d^3 \psi_d )\bar{\psi}_N^{-\bar{\delta}_1/2}d^{-1/2} \varepsilon\right),
\end{equation}
where $c_{\alpha,d}$ denotes the product of $q_\alpha^{d/2}$ and the volume of a standard $d$-dimensional hypersphere.
\item Asymptotic validity of the resulting confidence set
\begin{equation}
\P_\pi\left(C(T_1(\varepsilon) ) \ni \pi(f)\right)  \xrightarrow[]{}1-\alpha.
\end{equation}
\end{enumerate}
\end{theorem}

By choosing an appropriate simulation threshold $T^*$, the results of Theorem \ref{FVSR_optimal_rate} can also be guaranteed hold for the approximation rate obtained for the multi-step minorisation case.

\begin{theorem} 
\label{FVSR_one_dep_sip}
Suppose that $X$ satisfies the following strong Gaussian approximation
\begin{equation*}
\abs{\sum_{t=1}^T f(X_t)-T\pi(f)-\Sigma_f^{1/2}W_T} = \bigO_{a.s.} \left( \bar{\psi}_N \psi_d\Psi_T\right),
\end{equation*}

with approximation error $\Psi_T=T^{1/{4p_0}+1/4(p_0-1)}\log (T)$ for some $p_0>2$ and $\psi_d=d^a$ for some $a>0$.  Let $T_1(\varepsilon)$ be defined in \eqref{FVSR} with $\Lambda(t)=\mathbbm{1}_{\{t<T^* (\varepsilon,d,N)\}}+t^{-1}$, where
\begin{equation}
 \label{T_star_conditions_thrm}
 T^* (\varepsilon,d,N) = 
                       \left({\bar{\psi}_N \left(\frac{\tr(\Sigma_f)}{\sigma_0}\right)^{2}}d^{3}\psi_d \right)^{\frac{4(p_0-1)}{(p_0-2)}(1+\bar{\delta}_1)} \left(\frac{1}{\varepsilon}\right)^{\frac{8(p_0-1)}{(p_0-2)}(1+\bar{\delta}_2)}\vee e^{\frac{16(p_0-1)}{(p_0-2)}}
 \end{equation}

for any $\bar{\delta}_1 > 3/(3+a)$ and $\bar{\delta}_2 > 0 $. Let $\widehat{\Sigma}_T$ in \eqref{CI} denote the batch means estimator defined in \eqref{multi_bm_est}, with batch size $\ell_T$ set as
\begin{align}
\label{optimal_batch_size}
\ell_T=   
 \bar{\psi}_N \psi_d T^{\frac{3}{4}+\frac{1}{4(p-1)} }\log^{\bar{\delta}_3}(T),
 \end{align}
 for any $\bar{\delta_3} >0$. Suppose that Assumptions \ref{assumption_moment_condition} and \ref{assumption_covariance} hold.
Then we have as $\varepsilon \downarrow 0$ that the conclusions of Theorem \ref{FVSR_optimal_rate} regarding the asymptotic behaviour of the termination time and the asymptotic validity of the resulting confidence set hold.
\end{theorem}

The quantitative convergence bounds obtained in Theorems \ref{FVSR_optimal_rate} and \ref{FVSR_one_dep_sip} offer guidelines for the implementation of the FVSR in a wide array of settings.

\begin{remark}
While the result of Theorem \ref{FVSR_optimal_rate} is formulated in a high-dimensional setting, note that in an application where we can assume the influence of the dimension to be negligible, we obtain as an immediate consequence the following appropriate minimum simulation thresholds. We note that these simulation thresholds are also required for the reliability of the ESS.
\end{remark}
%
%
%
\begin{table}[ht]
\begin{center}
\begin{tabular}{|c|c|c|}
    \hline
    \textbf{} & \textbf{one-step minorisation} & \textbf{multi-step minorisation} \\
    \hline
    \rule{0pt}{22pt}
        geometric drift & $ \displaystyle \left(\frac{1}{\varepsilon}\right)^{\frac{4p}{(p-2)}(1+\bar{\delta}_2)}\vee e^{\frac{10p}{p-2}}$  & $\displaystyle \left(\frac{1}{\varepsilon}\right)^{\frac{8(p-1)}{(p-2)}(1+\bar{\delta}_2)}\vee e^{\frac{16(p-1)}{(p-2)}}$ \\   [12pt]  
   \hline
   \rule{0pt}{22pt}
    polynomial drift & $ \displaystyle \left(\frac{1}{\varepsilon}\right)^{\frac{4p_0}{(p_0-2)}(1+\bar{\delta}_2)}\vee e^{\frac{10p_0}{p_0-2}}$ & $\displaystyle \left(\frac{1}{\varepsilon}\right)^{\frac{8(p_0-1)}{(p_0-2)}(1+\bar{\delta}_2)}\vee e^{\frac{16(p_0-1)}{(p_0-2)}}$
     \\[12pt]
   \hline
\end{tabular}
\caption{Dependence of minimum simulation threshold $T^*$ on precision $\varepsilon$}
\label{table:table_precision_termination}
\end{center}
\end{table}

 Our results also enable us to study the convergence complexity of the FVSR and guarantee its validity in high-dimensional settings. Moreover, it can be guaranteed that the termination time scales polynomially in dimension while ensuring a desired level of precision. In the high-dimensional setting, we need to impose an additional multiplicative factor on the simulation threshold $T^*$. These factors are detailed in Table \ref{table:min_seq_dim}.\\

\begin{table}[ht]
\begin{center}
\begin{tabular}{|c|c|c|}
    \hline
    \textbf{} & \textbf{one-step minorisation} & \textbf{multi-step minorisation} \\
    \hline
    \rule{0pt}{26pt}
        geometric drift & $ \displaystyle \left({\bar{\psi}_N \left(\frac{\tr(\Sigma_f)}{\sigma_0}\right)^{2}}d^{3}\psi_d \right)^{\frac{2p}{(p-2)}(1+\bar{\delta}_1)} $  & $\displaystyle \left({\bar{\psi}_N \left(\frac{\tr(\Sigma_f)}{\sigma_0}\right)^{2}}d^{3}\psi_d \right)^{\frac{4(p-1)}{(p-2)}(1+\bar{\delta}_1)} $\\[12pt]
   \hline
   \rule{0pt}{26pt}
    polynomial drift & $ \displaystyle \left({\bar{\psi}_N \left(\frac{\tr(\Sigma_f)}{\sigma_0}\right)^{2}}d^{3}\psi_d \right)^{\frac{2p_0}{(p_0-2)}(1+\bar{\delta}_1)} $ & $\displaystyle \left({\bar{\psi}_N \left(\frac{\tr(\Sigma_f)}{\sigma_0}\right)^{2}}d^{3}\psi_d \right)^{\frac{4(p_0-1)}{(p_0-2)}(1+\bar{\delta}_1)} $ \\[12pt]
   \hline

\end{tabular}
\caption{Dimension dependence of the minimum simulation threshold $T^*$}
\label{table:min_seq_dim}
\end{center}
\end{table}

A widely used method for determining the run-length of the simulation is to terminate the simulation when the ESS reaches a desired threshold. We note that stopping according to the effective sample size is equivalent to termination rules based on confidence sets, see for example \cite{multivariate_output}. Moreover, the results of \cite{glynn} show that only consistent estimation of the variance or the ESS is insufficient to justify their use
as termination rules. Hence terminating according to the ESS would require similar minimum simulation requirements as those stated Theorem \ref{FVSR_optimal_rate}, Tables \ref{table:table_precision_termination} and \ref{table:min_seq_dim}.


%
%
%
%
%
%

\section{Proofs of Main Results}

\label{section:5}
\subsection{Preliminary results on drift and moment conditions}
In this section, we discuss the preliminary results that are required for our Gaussian approximation results. More specifically, we show how they follow from the assumed drift and minorisation conditions. We briefly review the splitting procedure of Harris chains based on \citet{asmussen}, \citet[Chapter 17.3]{meyn_tweedie_2012},  and \citet{sigman_review} and we discuss the implications of the assumed drift and minorisation conditions.
Let $X$ be an irreducible, aperiodic, positive Harris recurrent Markov chain taking values in Polish state space. From \citet[Proposition 5.4.5]{meyn_tweedie_2012} we know that $X$ satisfies the following minorisation condition 
\begin{equation}
P^{m_0}(x,dy)\geq \alpha \mathbbm{1}_C(x) \nu (dy),
\label{minorization_skeleton}
\end{equation}
for some  $\alpha \in (0,1)$, $m_0 \in \mathbb{N}$, measurable set $C$ with $\pi(C)>0$, and probability measure $\nu$ that is equivalent to $\pi|_{C}$. Note that from (\ref{minorization_skeleton}) it follows that the transition kernel of the so-called $m_0$-skeleton chain, defined as $(X_{km_0})_{k \in \mathbb{N}} $, can be interpreted as a mixture of two transition kernels, namely

\begin{equation}
\label{skeleton_mixture}
P^{m_0}(x,dy)= s(x) \nu(dy) + (1-s(x)) R(x,dy),
\end{equation}
where $s(x)= \alpha \mathbbm{1}_C(x)$ and the so-called residual kernel $R(x,dy)$ is defined as 
\begin{equation}
\label{residualkernel}
    R(x,dy)=\frac{P^{m_0}(x,dy)- s(x) \nu(dy)}{1-s(x)}.
\end{equation}

Given that the skeleton chain has hit $C$, with probability $\alpha$ the chain will move independently of its past according to the small measure $\nu$ and with probability $(1-\alpha)$ it will move according to the residual kernel $R$. Since the $m_0$-skeleton chain is also positive Harris recurrent, it will hit $C$ infinitely often. By a Borel--Cantelli argument it follows that the chain will transition according to $\nu$ infinitely often. Let $X'$ denote the split chain of the $m_0$-skeleton of $X$, i.e. for $n\in \ZZ_+$ we define $X'_n:=(X_{nm_0},\delta_n),$ where $\delta_n$ is a Bernoulli random variable that describes the distribution of the next point of the skeleton chain. The split chain has state space $\mathscr{X}\times \{0,1\}$  and has transition kernel
\begin{equation}
\label{split_transition}
    P'((x,\delta), (dy,d\delta'))=  \left\{
                \begin{array}{ll}

                P(x,dy) s(y)^{\delta'} (1-s(y))^{1-\delta'} ,  &  x\notin C,   
                  \bigskip\\
                  
                  \nu(dy) s(y)^{\delta'} (1-s(y))^{1-\delta'}\ \ \ \    ,   & \  x\in C;\ \delta=1,
                  \bigskip\\
                 
                  R(x,dy) s(y)^{\delta'} (1-s(y))^{1-\delta'},   & \ x\in C;\  \delta=0.
                  
                \end{array}
              \right.
\end{equation}

Hence if $\delta_n=1$, the next point of the skeleton chain has law $\nu$ and otherwise its law is described by the residual kernel. Let $R_k$ denote the $k$-th time that the $m_0$-skeleton chain moves according to $\nu$. The randomised stopping times $(R_k)_k$ serve as regeneration epochs for the skeleton chain, whereas they will be semi-regeneration epochs for the process $X$. We say that a process is semi-regenerative if there exists (by enlarging the probability space if necessary) a sequence of independent and identically distributed random variables $(\rho_k)$ that define a renewal process $(R_n)$ with $R_n=\sum_{k=1}^n \rho_k$  such that for each $n\geq 0$ the post-$R_n$  process

$$\{ (X_{R_n+k})_{k \geq 0}, R_{n+1}, R_{n+2}, \cdots \}$$

is independent of $R_0,\cdots,R_n$ and its distribution does not depend on $n.$ Note that this implies that the process can be split into identically distributed cycles, where the lengths of the cycle are described by a renewal process. The classically regenerative definition would also impose the cycles to be independent. If the chain $X$ satisfies a one-step minorization condition, i.e., (\ref{minorization_skeleton}) holds with $m_0=1$, that the process inherits a classical regenerative structure, whereas for the general case where $m_0>1$  the process has a semi-regenerative structure with one-dependent cycles. In order to see this, we first show how the state space can be enlarged to support the semi-regeneration times of the process. Let $(\delta_n)$ again denote a sequence of Bernoulli random variables which will describe the distribution of the $m_0$-skeleton points. Let $
(\mathcal{F}_t^X)_t,(\mathcal{F}_t^\delta)_t$ denote the natural filtration of the process and the auxiliary Bernoulli variables respectively. Consider the joint law of the chain in blocks of size $m_0$; 
\begin{flalign}
    &\PP\left( \delta_n=1, X_{nm_0+1} \in dx_1, \cdots, X_{(n+1)m_0-1} \in dx_{m_0-1},,X_{(n+1)m_0} \in dy \mid \mathcal{F}^X_{nm_0}, \mathcal{F}^\delta_{n-1}; X_{nm_0} \right) \nonumber \\
    &=\PP\left( \delta_n=1, X_{nm_0+1} \in dx_1, \cdots, X_{(n+1)m_0-1} \in dx_{m_0-1},X_{(n+1)m_0} \in dy \mid  X_{nm_0}\right) \nonumber \\
      &=\alpha r(X_{nm_0},y)P(X_{nm_0},dx_1)\cdots P(x_{m_0-1},dy),
      \label{m_points_law}
\end{flalign}
where all equalities hold almost surely and $r$ denotes the Radon-Nykodym derivative
\begin{equation}
\label{regeneration_acceptance_probability}
r(x,y)=\mathbbm{1}_C(x) \frac{\nu(dy)}{P^{m_0}(x,dy)}.
\end{equation}
Note that we also have 
\begin{flalign}
&\mathbb{P}(\delta_n=1,X_{(n+1)m_0}\in dy \mid \mathcal{F}^X_{nm_0}, \mathcal{F}^\delta_{n-1}; X_{nm_0})
\nonumber\\
&=\int_{x_1, \cdots, x_{m_0-1}} \hspace{-0.6cm}
\PP\left( \delta_n=1, X_{nm_0+1} \in dx_1, \cdots, X_{(n+1)m_0-1} \in dx_{m_0-1},X_{(n+1)m_0} \in dy \mid  X_{nm_0} \right)
\nonumber\\
&=\int_{x_1, \cdots, x_{m_0-1}} \hspace{-0.6cm} \alpha r(X_{nm_0},y)P(X_{nm_0},dx_1)\cdots P(x_{m_0-1},dy)
\nonumber\\
&=\alpha \mathbbm{1}_C(X_{nm_0}) \frac{\nu(dy)}{P^{m_0}(X_{nm_0},dy)}P^{m_0}(X_{nm_0},dy)
\nonumber\\
&=\alpha \mathbbm{1}_C(X_{nm_0}) {\nu(dy)}, \nonumber
\end{flalign}
where the third equality follows from the Chapman--Kolmogorov equations.

It easily follows that we also have 
\begin{flalign}
&\mathbb{P}(\delta_n=1 \mid \mathcal{F}^X_{nm_0}, \mathcal{F}^\delta_{n-1}; X_{nm_0})= \alpha \mathbbm{1}_C(X_{nm_0})
\label{law_delta}\\
&\mathbb{P}(X_{(n+1)m_0 } \in dy \mid \mathcal{F}^X_{nm_0}, \mathcal{F}^\delta_{n-1}; X_{nm_0}, \delta_n=1)= \nu(dy)
\label{law_regeneration}
\end{flalign}

From (\ref{law_delta}) and (\ref{law_regeneration}) we see that given $\delta_n=1$ we have that 
$$\{X_k, \delta_i: k \leq nm_0, i\leq n \}\ \textrm{is independent of } \{X_k,\delta_i : k \geq (n+1)m_0, i \geq n+1\}.$$

Furthermore, we also have that the process $\{X_k,\delta_i : k \geq (n+1)m_0, i \geq n+1\}$ is equal in distribution to $\{(X_k,\delta_k) : k \geq 0\}$ with initial distribution 
$$\PP(X_0 \in dx_0 , \delta_0 \in d\delta)=  \nu(dx_0) \Bern(s(x))$$

Hence we see that the process $X$ can be embedded in a richer process, which admits a recurrent atom $A:=C \times \{1\}$ in the sense of the following proposition. 

\begin{proposition}
\label{semi_reg_discrete}
Let $(S_n,R_n)$ be a sequence of stopping times defined as $S_0=R_0:=0$ and 
$$S_{n+1}:=\inf \{ km_0 > R_n: (X_{km_0}, \delta_k) \in C \times \{1\} \} \ \ \textrm{and} \ \ R_{n+1}:= S_{n+1}+m_0.$$
Then $X_{R_n}$ is independent of $R_n$ and $\mathcal{F}_{R_{n-1}}$ for all $n\geq 1$, $(X_{R_n}, \delta_{R_n})_{n\geq 1}$ is an i.i.d sequence with $$(X_{R_n}, \delta_{R_n})\sim \nu(dx)\Bern(s(x))\ \ \textrm{for all}\ n\geq 1,$$
the process is semi-regenerative, the cycles lengths $\rho_k$ are independent and identically distributed, and the cycles $$\{X_k: R_{n-1} \leq k < R_n\}$$
are identically distributed and one-dependent for all $n\in \mathbb{N}$. If $m_0=1$, then the cycles are in fact independent.
\end{proposition}
\begin{proof}
\citet[Theorem 4.2]{sigman_review} and \citet[Theorem ]{asmussen}.
\end{proof}
The stopping times $\{S_n\}_n$ thus denote the hitting times of the recurrent atom $A$ and $\{R_n\}_n$ denote the implied regeneration epochs of the chain. As a direct consequence, of the semi-regenerative structure and the fact that $R_n$ forms a renewal process, we obtain the following characterisation of the stationary measure. Moreover, introduce the following stopping time for the $m_0$--skeleton:
\begin{equation}
\label{hitting_time_skeleton}
\bar{\tau}_{C\times \{1\}}= \inf \{k: (X_{km_0}, \delta_k) \in C \times \{1\}\}    
\end{equation}

It is well-known that the drift inequalities are closely related to the moments of hitting times of the process. For processes satisfying a geometric drift condition, the moment bounds for the hitting time $\bar{\tau}_{C\times \{1\}}$ given in \cite{baxendale2005renewal} are to the best of our knowledge the tightest bounds that are currently available.

\begin{lemma}[\protect{\cite[Proposition 4.4]{baxendale2005renewal}  }]
\label{moment_bound_exp} Let $(X_t)_{t \in \mathbb{N}}$ be an irreducible aperiodic Markov chain, assume that an $m_0$--step minorisation condition and drift condition \ref{exp_drift_discrete_skeleton} are satisfied. Then, for $1< r\leq \lambda^{-1}$,
\begin{align}
\label{exp_moment_bound_geom}
\mathbb{E}_x[ r^{ \bar{\tau}_{{C} \times \{1\} }}] \leq \frac{\alpha G(r,x)}{1-(1-\alpha)r^{a}}  \ \textrm{and} \ \ 
\E_\nu [ r^{ \bar{\tau}_{{C} \times \{1\} }}] \leq \frac{\pi(V) }{(1-(1-\alpha)r^{a})\alpha \pi(C) } ,
\end{align} 
where
$$a=1+\left(\log \frac{\lambda \upsilon_V+b-\alpha}{1-\alpha}\right)/(\log(\lambda^{-1})),$$
and
\begin{equation}
    G(r,x)\leq   \left\{
                \begin{array}{ll}
                  V(x),   & \mbox{for}\  x \in C\\
                 
                  r(\lambda \upsilon_C+b),   & \mbox{for}\  x \notin C.
                  
                \end{array}
              \right.
\end{equation}
Moreover, we have that
\begin{align}
\label{V_bound_geom}
    \pi(V) \leq \frac{b}{1-\lambda}\pi(C).
\end{align}
\end{lemma}

\begin{proof}
The claim \eqref{exp_moment_bound_geom} is a combination of 
\cite[Proposition 4.4]{baxendale2005renewal} and  Lemma \ref{pi_nu_bound}. The claim \eqref{V_bound_geom} is given in in \cite[Proposition 4.3 ]{meyn1994computable}.

\end{proof}

For chains satisfying a polynomial drift condition, we formulate the following result which gives us bounds on the moments of the semi-regeneration times. 

\begin{lemma}
\label{moment_bound_poly}
Let $(X_t)_{t \in \mathbb{N}}$ be an irreducible aperiodic Markov chain, assume that an $m_0$--step minorisation condition and drift condition
\ref{poly_drift_discrete_skeleton} are satisfied. Then, provided that $\eta>1/2$, we can assume that drift condition
\ref{poly_drift_discrete_skeleton} holds for some drift function $V$ with $\pi(V)\leq \pi(C)b/c$ such that 
\begin{align}
\E_x[\bar{\tau}_{C\times \{1\}}^{q}] \leq   {V}(x)+ (1-\alpha)^{-1} \left( \upsilon_C- {c} +{b}\right)\mathbbm{1}_C(x) 
\end{align} 
and 
\begin{align}
    \E_\nu[\bar{\tau}_{C\times \{1\}}^{q}]\leq  \frac{\pi(V)}{\alpha \pi(C)}+\frac{\upsilon_C-c+b}{1-\alpha},
\end{align}
with $q=\frac{\eta}{1-\eta}$. Moreover,
\begin{align}
\label{first_moment_bound_poly}
\mathbb{E}_x[  {\tau_C}] \leq \frac{1}{(1-\eta) c} \left(V^{1-\eta}(x)+(b^\eta+b_0)\mathbbm{1}_C(x)\right) \ \textrm{and} \ \ 
\mathbb{E}_\nu[  \tau_C] \leq  \frac{\pi(V)^{1-\eta}}{(1-\eta)\alpha c \pi(C)}+ \frac{ (b^\eta+b_0)}{(1-\eta)\alpha c}.    
\end{align} 
\end{lemma}
\begin{proof}
Firstly,  by the Comparison theorem, \cite[Theorem 14.2.2]{meyn_tweedie_2012}, we have that $\pi(V^\eta)\leq  {b \pi(C)}/c.$
 From  \cite[Lemma 3.5]{jarner2002polynomial} we see that $V^\eta$ is also a Lyapunov function satisfying drift condition 
 \[
PV^\eta \leq V^\eta - c \eta  V^{2\eta -1}+ (b^\eta+b_0)\mathbbm{1}_C,
\]
for some $b_0>0$. Hence provided that $\eta>1/2$, the first assertion follows. In order to use the hitting time bounds implied by a polynomial drift condition, we must first show that the split-chain of the $m_0$--skeleton also satisfies a polynomial drift condition. Define a Lyapunov function on the extended state space as follows: $ {V}(x,0)= {V}(x,1)= {V}(x)$.  Since the transition kernel of the split chain of the $m_0$--skeleton is given by \eqref{split_transition}, we have that for $x \notin C$, the Lyapunov condition is already satisfied. For $x\in C$ and $\delta=1$ it immediately follows that

\begin{align*}
    \hat{P}^{m_0} {V}((x,1))&=\int_C  {V}(y) \nu(dy)\\
    & \leq    {V}(x) -  {c}  {V}^{ {\eta}}(x) +  {c}\upsilon_C^{ {\eta}}+ \upsilon_C
\end{align*}
For $x\in C$ and $\delta=0$ we have that
    \begin{align*}
     \hat{P}^{m_0} {V}((x,0))&=\int_E  {V}(y) R(x,dy)\\
    & \leq (1-\alpha)^{-1} \int_E  {V}(y) P^{m_0}(x,dy)\\
        & \leq (1-\alpha)^{-1} \left(   {V}(x)-  {c} {V}^{\eta}(x)+ {b}\right)\\
         & \leq  {V}(x) -  {c} {V}(x)^\eta+  {c}\upsilon_C^{ {\eta}} + (1-\alpha)^{-1} \left( \upsilon_C-  {c} + {b}\right)
\end{align*}
    By \cite[Proposition 2.2]{douc2004practical}
    we have that 
\[
 \E_x[\bar{\tau}_C^{q}]\leq  {V}(x)+ (1-\alpha)^{-1} \left( \upsilon_C- {c} +{b}\right)\mathbbm{1}_C(x)
\]
By Lemma \ref{pi_nu_bound} it follows that
\begin{align*}
 \E_\nu[\bar{\tau}_C^{q}]&=\frac{1}{\alpha \pi(C)} \E_\pi[\bar{\tau}_C^{q}]\\
 &=\frac{1}{\alpha \pi(C)} \left(  \pi(V)+  (1-\alpha)^{-1} \left( \upsilon_C-  {c} + {b}\right)\pi(C)\right)
    \end{align*}

 In order to show \eqref{first_moment_bound_poly}, we note that by  \cite[Lemma 3.5]{jarner2002polynomial}] $V^{1-\eta}$ is also a Lyapunov function satisfying drift condition 
 \[
PV^{1-\eta} \leq V^{1-\eta} - c(1-\eta)   + (b^\eta+b_0)\mathbbm{1}_C.
\]
By the Comparison theorem, \cite[Theorem 14.2.2]{meyn_tweedie_2012}, we have that the first claim follows. The second part again follows from Lemma \ref{pi_nu_bound} and Jensen's inequality. 
\end{proof}

Note that the bounds obtained in Lemma  \ref{moment_bound_poly} are quite general and therefore for specific situations tighter bounds could be obtained, see for example 
\cite{andrieu2015quantitative} obtain quantitative bounds for subgeometric markov chains under the assumptions of a lower bound on the Lyapunov function outside of $C$ and a minorisation condition for all skeleton chains of the process. They also assume a lower bound on V.
\cite{andrieu2015quantitative} show that if we assume a lower bound on the Lyapunov function: $\inf_{x \notin C}V(x)\geq b(1-\varepsilon)^{-1}$, for some $\varepsilon>0$ and that a one-step minorisation condition holds for the original chain and all its skeletons, then the bounds presented in Lemma \ref{moment_bound_poly} can be improved. \\

 Note that we can obtain a drift condition for the $m_0$--skeleton by iterating the one-step drift condition, namely,

    \[
    P^{m_0}V(x) \leq \lambda^{m_0}V(x)+b\sum_{i=0}^{m_0-1} P^i \mathbbm{1}_C(x) \leq \lambda V(x)+b m_0\mathbbm{1}_{C(m_0)}(x),
    \]
    where $C(m_0)$ is a small set for the skeleton chain and $C \subseteq C(m_0)$. Under additional regularity conditions, we can obtain a drift condition for the skeleton chain towards the set $C$.

 
\label{exp_condition_superset}
\label{poly_condition_superset}

\begin{proposition}
\label{exp_condition_superset}Let $(X_t)_{t \in \mathbb{N}}$ be an irreducible aperiodic Markov chain, assume that an $m_0$--step minorisation condition holds and that drift condition \ref{exp_drift_discrete_skeleton} holds for some $C_0 \supsetneq C$    such that   either there exists some $\bar{\alpha}>0$ such that 
\begin{align}
    \label{contains_condition_geom}
P^{m_0}(x,C)\geq \bar{\alpha}\mathbbm{1}_{C_0 \backslash C}(x) \ \ \textrm{or} \ 
 \inf_{x \in  C_0 \backslash C }V(x)> b(1-\lambda)^{-1}.
 \end{align}
 Then there exists a function $\hat{V}:E\rightarrow\R^+$ and   $\hat{\lambda} \in (0,1)$ and $\hat{b}>0$   such that
\[
P^{m_0}\hat{V}(x)\leq \hat{\lambda}\hat{V}(x) +\hat{b}\mathbbm{1}_C(x).
\] 

\end{proposition}
\begin{proof}
Let $z_1:= \inf_{x \in  C_0 \backslash C }V(x)> b(1-\lambda)^{-1}$.   Then for $x \in C_0 \backslash C$ we have that
\begin{align*}
 PV(x)&\leq  V(x)-cV(x)^\eta+b\\
 &\leq \hat{\lambda}V(x),
\end{align*}
provided that we choose $\hat{\lambda}=\frac{b+z_1\lambda}{z_1}$. Note that since $\lambda< \hat{\lambda}<1$, it follows immediately that for $x\in C$ we have $PV\leq \lambda V+b \leq \hat{\lambda}V+b $. Hence the desired claim immediately follows. Under the assumption $P^{m_0}(x,C)\geq \bar{\alpha}\mathbbm{1}_{C_0 \backslash C}(x)$, the claim follows completely analogously to the proof of \cite[Theorem 6.1]{meyn1994computable}. 
\end{proof}

\begin{proposition}
\label{poly_condition_superset}Let $(X_t)_{t \in \mathbb{N}}$ be an irreducible aperiodic Markov chain, assume that an $m_0$--step minorisation condition holds and that drift condition \ref{poly_drift_discrete_skeleton} holds for some $C_0 \supsetneq C$    such that   either there exists some $\bar{\alpha}>0$ such that 
\begin{align}
    \label{contains_condition_poly}
P^{m_0}(x,C)\geq \bar{\alpha}\mathbbm{1}_{C_0 \backslash C}(x) \ \ \textrm{or} \ 
 \inf_{x \in  C_0 \backslash C }V(x)> (1+b/c)^{1/\eta}.
 \end{align}
 Then there exists a function $\hat{V}:E\rightarrow\R^+$ and   $\hat{\eta} \in (0,\eta)$ and $\hat{c}>0$   such that
\[
P^{m_0}\hat{V}(x)\leq   \hat{V}(x) -\hat{c}\hat{V}^{\hat{\eta}}(x) +\hat{b}\mathbbm{1}_C(x).
\] 

\end{proposition}
\begin{proof}
Under the assumption $P^{m_0}(x,C)\geq \bar{\alpha}\mathbbm{1}_{C_0 \backslash C}(x)$, the claim follows completely analogously to the proof of Proposition \ref{poly_drift_discrete_subset}. For the second case, we will take $\hat{c}=c$ and $\hat{V}=V$. Let  $z_1:= \inf_{x \in  C_0 \backslash C }V(x)> (1+b/c)^{1/\eta}$ and $\hat{\eta}=\ln(z_1^\eta-b/c)/ \ln(z_1).$ Then for $x \in C_0 \backslash C$ we have that
\begin{align*}
 PV(x)&\leq V(x)- cV^\eta(x)+b\\
 &\leq V(x) -cV^{\hat{\eta}(x)},
\end{align*}
since we have that $\hat{\eta}<\eta$ and therefore
\[
\frac{c}{V^{\eta-\hat{\eta}}(x)}+\frac{b}{V^\eta(x)} \leq \frac{c}{z_1^{\eta-\hat{\eta}}}+ \frac{b}{z_1^\eta} \leq c,
\]
which is equivalent to $cV^{\hat{\eta}}(x)+b \leq cV^\eta(x)$.
\end{proof}
Hence, Propositions \ref{exp_condition_superset} and \ref{poly_condition_superset} give conditions such that the one-step drift condition implies the desired drift condition for the skeleton chain. 
Furthermore,  we note that in order for a drift condition towards $C$ to hold, it is sufficient to show that it holds on some appropriate subset of $C$. This might be a useful property for $T$-chains, where every closed set is petite and hence a small set. 

\begin{proposition} [\protect{\cite[Theorem 6.1]{meyn1994computable} }]
\label{exp_drift_discrete_subset}Let $(X_t)_{t \in \mathbb{N}}$ be an irreducible aperiodic Markov chain, assume that an $m_0$--step minorisation condition holds and that drift condition \ref{exp_drift_discrete_skeleton} holds for some $C_0 \subsetneq C$ with $\pi(C_0)>0$. Then there exists a function $\hat{V}:E\rightarrow\R^+$ with $V\leq \hat{V}\leq V+b/\alpha \nu(C_0)$ such that
\[
P^{m_0}\hat{V}(x)\leq \hat{\lambda}\hat{V}(x) +\hat{b}\mathbbm{1}_C(x),
\] 
with 
\[
\hat{\lambda}= \frac{\lambda \alpha \nu(C_0)+b}{\alpha \nu(C_0)+ b } \ \textrm{and}\ \ \hat{b}=b+b/ \alpha \nu(C_0).
\]
    
\end{proposition}
    
 We can easily prove an analogous lemma for chains satisfying a polynomial drift condition.

\begin{proposition}
\label{poly_drift_discrete_subset} Let $(X_t)_{t \in \mathbb{N}}$ be an irreducible aperiodic Markov chain, assume that an $m_0$--step minorisation condition holds and drift condition \ref{poly_drift_discrete_skeleton} 
is satisfied for some $C_0 \subseteq C$ with $\pi(C_0)>0$,  then there exist a function $\hat{V}:E\rightarrow\R^+$ with $V\leq \hat{V}\leq V+(1 \wedge b)/\alpha \nu(C)$ such that
\[
P^{m_0}\hat{V}(x)\leq \hat{V}(x)-\hat{c}V(x)^{\hat{\eta}}+\hat{b}\mathbbm{1}_C(x),
\] 
with 
\[
\hat{c} \in \left(\frac{c}{(1 \wedge b)/\alpha \nu(C)^\eta+1},\frac{c}{2}\right),
\]
\[
\hat{b}=b+\frac{(1 \wedge b)}{\alpha \nu(C)}(1-\alpha \nu(C)) \  \textrm{and} \   \hat{\eta}=\frac{ \ln(\frac{c-\hat{c}}{\hat{c}})}{ \ln((1 \wedge b)/\alpha \nu(C))}
\]
\end{proposition}

\begin{proof}
    For some $B$ to be determined at a later point in the proof, let
    \[
    \hat{V}(x)=   \left\{
                \begin{array}{ll}
                  V(x),   & \  x \in C\\
                 
                  V(x)+B,   & \  x \notin C.
                 
                \end{array}
              \right.
    \]
    Now we see that
    \begin{align*}
    P^{m_0}\hat{V}(x)&= \int_E \hat{V}(y) P^{m_0}(x,dy) = \int_{C} V(y)P^{m_0}(x,dy)+ \int_{{C}^c} (V(y)+B)P^{m_0}(x,dy)\\
    &=P^{m_0}V(x)+BP^{m_0}(x,C^c)\end{align*}

For $x\notin C$ we have that
\begin{align*}
P^{m_0}\hat{V}(x) &\leq V(x)-cV(x)^\eta+B 
\end{align*}
Hence if $B>1$, $\hat{c} \in \left(\frac{c}{B^\eta+1},\frac{c}{2}\right),$ and  $\hat{\eta}\in \left(0,  \ln(\frac{c-\hat{c}}{\hat{c}})/ \ln(B)\right],$
we see that $\hat{\eta}<\eta$ and thus
\[
V^{\hat{\eta}-\eta}+\frac{B^{\hat{\eta}}}{V^{\eta}} \leq 1+ B^{\hat{\eta}}\leq \frac{c}{\hat{c}}.
\]
Therefore it also follows that 
\[
 (V+B)^{\hat{\eta}}\leq V^{\hat{\eta}}+B^{\hat{\eta}}\leq \frac{c}{\hat{c}}V^\eta.
\]
Note that  $1+B^{\hat{\eta}}\leq c/\hat{c}$ gives $\hat{\eta}$ and the restriction $0<\hat{\eta}< \eta $ gives the restriction on $\hat{c}$. It follows that
\[
   P^{m_0}\hat{V}(x)  \leq \hat{V}(x)-\hat{c}\hat{V}(x)^{\hat{\eta}},
\]
as desired. For $x\in C^c \cap C'$, we have that $P^{m_0}(x,C)\geq \alpha \nu(C)$ and hence
\begin{align*}
    P^{m_0}\hat{V}(x)&=P^{m_0}V(x)+B(1-P^{m_0}(x,C))\\
    &\leq V(x)-cV(x)^\eta+b+B-\alpha \nu(C) B\\
    & \leq  \hat{V}(x) -cV(x)^\eta\\
    & \leq \hat{V}(x)-\hat{c}\hat{V}(x)^{\hat{\eta}},
\end{align*}
given that we choose $B=(1 \wedge b)/\alpha \nu(C)$. For $x \in C$ we see that
\begin{align*}
    P^{m_0}\hat{V}(x)&\leq V(x)-cV(x)^\eta+b+B(1-\alpha \nu(C))\\
    &\leq \hat{V}(x)-\hat{c}\hat{V}(x)^{\hat{\eta}}+b+B(1-\alpha \nu(C))
\end{align*}
\end{proof}

In \cite{jones_fixed} and \cite{remarks_fixed} it is shown that geometric ergodicity and moment conditions with respect to the stationary measure are sufficient to guarantee moment conditions of functionals over regenerative cycles. These results easily carry over to the continuous-time and multivariate setting as seen in  \citet{pengel2024strong} and \citet{banerjee2022} respectively. We generalise these results using the explicit bounds given in Lemma \ref{moment_bound_exp} and \ref{moment_bound_poly}. These results are of independent interest, see for example \cite{bertail2018new}, who assume that explicit bounds for moments of hitting times and blocks are given. Lemma \ref{moments_regen_exp} and \ref{moments_regen_poly} could for example be used to obtain dimension-dependent Bernstein inequalities, see \cite{bertail2018new}. Note that we have formulated the moment bounds in Lemmas \ref{moments_regen_exp} and \ref{moments_regen_poly}as inequalities up to a universal constant, however, from the proof it is immediate that all these constants can be given explicitly. However, since they play no role in the considered asymptotics in either dimension or simulation time, we for the sake of clarity and brevity shall omit them.

\begin{lemma}
\label{moments_regen_exp}
Let $(X_t)_{t \in \mathbb{N}}$ be an irreducible aperiodic Markov chain, assume that an $m_0$-step minorisation condition holds and that drift condition \ref{exp_drift_discrete_skeleton} is satisfied. 
Then for any $t$ with $\abs{t}\leq \ln(1/\lambda) /m_0$ we have that \begin{equation}
\label{regen_moment_cond_exp}
\E_\nu [e^{t R_1}]  \leq \frac{b}{\alpha \lambda (1-\lambda)}
\end{equation}
 Moreover, under Assumption A\ref{assumption_moment_condition}.1,  we have that 
    \begin{equation}
    \label{moment_coordinate_exp_p}
\sup_{i\in \{1,\dots,d\}}\EE_\nu \left[\left(\sum_{t=0}^{R_1} \abs{f_i(X_t)}\right)^{p} \right]  \   \lsim
 \alpha^{-1} \left( \frac{b}{\alpha \lambda (1-\lambda)}\right)^{\varepsilon/p}  \sup_{i\in \{1,\dots,d\}} \pi(\abs{f_i}^{p+\varepsilon}),          \end{equation}
and
     \begin{equation}
     \label{moment_norm_exp_p}
\EE_\nu \left[\abs{\sum_{t=0}^{R_1}f(X_t)}^{p} \right] \lsim \alpha^{-1}  d^{p/2}\left(\E_\nu e^{tR_1}\right)^{\varepsilon/p}\sup_{i\in \{1,\dots,d\}} \pi(\abs{f_i}^{p+\varepsilon}), 
    \end{equation}
 If we can additionally assume that Assumption A\ref{assumption_moment_condition}.3 holds, we also have that 
   \begin{equation}
    \label{moment_coordinate_exp_exp}
\sup_{i\in \{1,\dots,d\}}\EE_\nu \left[\exp\left(\sum_{t=0}^{R_1} \abs{f_i(X_t)}\right) \right]  \   \lsim    \E_\nu [e^{tR_1}] .        \end{equation}

\end{lemma}

\begin{lemma}
\label{moments_regen_poly}
Let $(X_t)_{t \in \mathbb{N}}$ be an irreducible aperiodic Markov chain, assume that an $m_0$--step minorisation condition holds and that drift condition  \ref{poly_drift_discrete_skeleton} is satisfied for some Lyapunov function. 
Then  \begin{equation}
\label{regen_moment_cond_poly}
\EE_\nu[R_1^{q}] \lsim 2^{q-1}m_0^q \left(1+\frac{b}{c \alpha}+\frac{\upsilon_c-c+b}{1-\alpha}\right).
\end{equation}
 
where $
    q(\eta)=\frac{\eta}{1-\eta}
    $. Moreover,  for all $f: E\rightarrow \mathbb{R}^d$ such that Assumption A\ref{assumption_moment_condition}.1 holds,  we have that 
    \begin{equation}
    \label{moment_coordinate_poly}
\sup_{i\in \{1,\dots,d\}}\EE_\nu \left[\left(\sum_{t=0}^{R_1} \abs{f_i(X_t)}\right)^{p_0} \right]  \   \lsim  \alpha^{-1} \left(\E_\nu R_1^q\right)^{\varepsilon/p_0} \sup_{i\in \{1,\dots,d\}} \pi(\abs{f_i}^{p_0+\varepsilon}),      
    \end{equation}
   
where  

        \begin{equation}
        \label{p_0definition}
    p_0=   \left\{
                \begin{array}{ll}
                  \frac{pq(\eta)}{p+q(\eta)+\varepsilon},   & \  \textrm{if } \frac{p}{2p-1} < \eta \leq p(p+\varepsilon)/(p(p+\varepsilon)+\varepsilon) ,\\
                 
                  p,   & \  \textrm{if } \eta > p(p+\varepsilon)/(p(p+\varepsilon)+\varepsilon) ,\\

                  q(\eta)-\bar{\epsilon},   & \  \textrm{if } \eta> 1/2 \  \textrm{  and A\ref{assumption_moment_condition}.2 holds},

                \end{array}
              \right.
    \end{equation} 
    for any fixed $\bar{\epsilon} \in \left(0,\min \{\frac{1}{2},\frac{2\eta-1}{1-\eta}\}\right)$.

\end{lemma}

\subsubsection{Proof of Lemma \ref{moments_regen_exp}}

\begin{proof}
Firstly, we show that $R_1$ admits an exponential moment in some neighborhood of zero. Note that from Proposition \ref{semi_reg_discrete} we see that $R_1=m_0+m_0 \bar{\tau}_{C \times \{1\}}$. Hence for any $t$ with $\abs{t}< \frac{\ln(1/\lambda)}{m_0}$ we have that
\begin{align*}
\E_\nu [e^{t R_1}] &= \E_\nu [e^{t (m_0+m_0 \bar{\tau}_{C \times \{1\}}))}] \\
&\leq \E_\nu [e^{ \frac{\ln(1/\lambda)}{m_0} (m_0+m_0 \bar{\tau}_{C \times \{1\}}))}] \leq \frac{1}{\lambda}\E_\nu \left[ \left(\frac{1}{\lambda}\right)^{\bar{\tau}_{C \times \{1\}}} \right]
  \end{align*}
By an application of Lemma \ref{moment_bound_exp} we obtain that
\[
\E_\nu [e^{t R_1}]  \leq \frac{b}{\alpha \lambda (1-\lambda)}
\]

Note that $$\sum_{t=0}^{R_1}{f(X_t)} = \left(\sum_{t=0}^{R_1}{f_1(X_t)}, \dots, \sum_{t=0}^{R_1}{f_d(X_t)}\right)^T,$$
with $f_i: E\rightarrow \R$ for $i=1,\dots,d.$  By Lemma \ref{pi_nu_bound}  and a coordinate-wise application of \citet[Lemma 2]{remarks_fixed}, we see for all $\i=1,\dots,d$ we have that
\begin{align*}
\EE_\nu\left[\left ({\sum_{t=0}^{R_1} \abs{f_i(X_t)}} \right)^{p} \right] & \leq (\alpha \pi(C))^{-1}\EE_\pi\left[\left ({\sum_{t=0}^{R_1} \abs{f_i(X_t)}} \right)^{p} \right]  \\
& \leq (\alpha \pi(C))^{-1} \pi(\abs{f_i}^{p+\varepsilon})^{\frac{p}{p+\varepsilon}} \left( \sum_{k=1}^\infty \P_\pi \left( R_1>k \right)^{\frac{\varepsilon}{p(p+\varepsilon)}}\right)^{p} \\
& \hspace{-4.3cm}\textrm{By an application of Markov's inequality we have that}\\
& \leq (\alpha \pi(C))^{-1} \pi(\abs{f_i}^{p+\varepsilon})^{\frac{p}{p_0+\varepsilon}} \left( \sum_{k=1}^\infty e^{-\frac{t\varepsilon}{p(p+\varepsilon)}k}\right)^{p} \left(\E_\pi e^{tR_1} \right)^{\varepsilon/({p} +\varepsilon)} \\
%
& \leq (\alpha \pi(C))^{-1}   \pi(\abs{f_i}^{p+\varepsilon})^{\frac{p}{p+\varepsilon}} \left(e^{t\varepsilon/p(p+\varepsilon)}-1 \right)^{-p} \left( \frac{b}{\alpha \lambda (1-\lambda)}\right)^{\varepsilon/p} .
\end{align*}

The second claim follows directly from the equivalence of norms, 
 \begin{align*}
 \E_\nu\left[ \abs{\sum_{t=0}^{R_1}{f(X_t)}}^p\right]&= \E_\nu\left[\left(\sum_{i=1}^d \left(\sum_{t=0}^{R_1}\abs{f_i(X_t)}\right)^2\right)^{p/2}\right]\\
 &\le  \E_\nu\left[\left(d^{1/2-1/p}\left(\sum_{i=1}^d \abs{\sum_{t=0}^{R_1}\abs{f_i(X_t)}}^p\right)^{1/p}\right)^p\right]\\
 &=d^{p/2-1}\E_\nu\left[\sum_{i=1}^d \abs{\sum_{t=0}^{R_1}\abs{f_i(X_t)}}^p\right] =d^{p/2-1} \left[\sum_{i=1}^d \E_\nu \abs{\sum_{t=0}^{R_1}\abs{f_i(X_t)}}^p\right] \\
  &\le d^{p/2} \sup_{i\in \{1,\dots,d\}} \EE_\nu\abs{\sum_{t=0}^{R_1}\abs{f_i(X_t)}}^{p},
 \end{align*}
which gives \eqref{moment_norm}.

 \end{proof}

\subsubsection{Proof of Lemma \ref{moments_regen_poly}}

\begin{proof}
Firstly, we show that $R_1$ has an exponential moment.
From Proposition \ref{semi_reg_discrete} we see that 
\[
\E_\nu R_1^q = \E_\nu [(m_0+m_0 \bar{\tau}_{C \times \{1\}})^q] \leq 2^{q-1} (m_0^q+m_0^q \E_\nu\bar{\tau}_{C \times \{1\}}^q),
\]

where the first  equality follows since $S_1=m_0\bar{\tau}_{C\times \{1\}}.$
The first claim \eqref{regen_moment_cond_poly} now follows for drift conditions \ref{exp_drift_discrete_skeleton} and \ref{poly_drift_discrete_skeleton}  by an application of Lemma \ref{moment_bound_exp} and Lemma \ref{moment_bound_poly}  respectively. 
Note that $$\sum_{t=0}^{R_1}{f(X_t)} = \left(\sum_{t=0}^{R_1}{f_1(X_t)}, \dots, \sum_{t=0}^{R_1}{f_d(X_t)}\right)^T,$$
with $f_i: E\rightarrow \R$ for $i=1,\dots,d.$ By  Lemma \ref{pi_nu_bound}  we see that for any $p_0 \leq p$ and for all $i=1,\dots,d$ we have that
\begin{align*}
\EE_\nu\left[\left ({\sum_{t=0}^{R_1} \abs{f_i(X_t)}} \right)^{p_0} \right] & \leq (\alpha \pi(C))^{-1}\EE_\pi\left[\left ({\sum_{t=0}^{R_1} \abs{f_i(X_t)}} \right)^{p_0} \right] 
\end{align*}

Following the argument of \citet[Lemma 2]{remarks_fixed} for every coordinate, we see that
\begin{align*}
\EE_\pi\left[\left ({\sum_{t=0}^{R_1} \abs{f_i(X_t)}} \right)^{p_0} \right]& \leq  \pi(\abs{f_i}^{p_0+\varepsilon})^{\frac{p_0}{p+\varepsilon}} \left( \sum_{k=1}^\infty \P_\pi \left( R_1>k \right)^{(1-\frac{p_0}{p+\varepsilon})}\right)^{p_0} \\
& \leq  \pi(\abs{f_i}^{p_0+\varepsilon})^{\frac{p_0}{p_0+\varepsilon}} \left( \sum_{k=1}^\infty k^{-\frac{q(\eta)}{p_0}\left(1-\frac{p_0}{p+\varepsilon}\right)}\right)^{p_0} \left(\E_\nu R_1^{q(\eta)}\right)^{\left(1-\frac{p_0}{p+\varepsilon}\right) }.
\end{align*}

Note that in each of the following cases the series $\sum_{k=1}^\infty k^{-\frac{q(\eta)}{p_0}\left(1-\frac{p_0}{p+\varepsilon}\right)}$ converges. Firstly, consider $\frac{p}{2p-1} < \eta \leq p(p+\varepsilon)/(p(p+\varepsilon)+\varepsilon) $ we have that $q(\eta)\leq p(p+\varepsilon)/\varepsilon$ and hence for a given value of $p$, the largest value $p_0$ such that the series converges is  $p_0=\frac{pq(\eta)}{p+q(\eta)+\varepsilon}$.
Note that impose $\eta > \frac{p}{2p-1} $ in order to have $p_0>1$.
We also see that if $\eta > p(p+\varepsilon)/(p(p+\varepsilon)+\varepsilon)$, then $q(\eta)>p(p+\varepsilon)/\varepsilon$ and we can take $p_0=p$. Finally, we see that under Assumption \ref{assumption_moment_condition}.2, we can take $p>q(q-\bar{\epsilon})/\bar{\epsilon}$, since $f$ has moments of all orders with respect to $\pi$, and hence for any $\bar{\epsilon}>0$ such that $q(\eta)-\bar{\epsilon}>1$ we can take $p_0=q(\eta)-\bar{\epsilon}$. The remaining claims follow completely analogously to Lemma \ref{moments_regen_exp}.

 \end{proof} 

\begin{lemma}(Marcinkiewicz–-Zygmund inequality; \protect{\cite[Lemma 1.4.13.]{de2012decoupling}})
\label{MZ_inequality}
 Suppose $z_1,\dots,z_d$ are one-dimensional independent zero-mean random variables with finite $p$-th moment, then there exist constants $c_p$ and $C_p$ such that
 $$
 c_p\E\left[\left(\sum_i z_i^2\right)^{p/2}\right]\le \E\left[\abs{\sum_i z_i}^p\right]\le C_p\E\left[\left(\sum_i z_i^2\right)^{p/2}\right]
 $$
 \end{lemma}
By the Marcinkiewicz–-Zygmund inequality, we can describe the growth of the moments of sums in terms of the bounds given in Lemmata \ref{moments_regen_exp} and \ref{moments_regen_poly}.
 \begin{lemma}
 \label{order_norm_power}Let $\{\xi_{k}\}_{k \in \mathbb{N}}$ be $d$-dimensional one-dependent zero-mean random vectors with $$\sup_{k \in \mathbb{N}} \ \sup_{i \in \{1,\cdots,d\}} \E[ \abs{\xi_{k,i}}^p] \leq C_{\xi},$$   then we have that 
 \[
 \E\left[\left|\sum_{k=0}^{\lfloor m/2\rfloor} \xi_{2k}\right|^p\right]\lsim C_{\xi}\   d^{p/2}m^{p/2}.
 \]
 \end{lemma}
 \begin{proof}
 Let $Z:=(z_1,\dots,z_d)^T:=\sum_{k=0}^{\lfloor m/2\rfloor} \xi_{2k}$ which is a $d$-dimensional random vector. Note that
 \[
 \abs{Z}:=\sqrt{\sum_{k=1}^d z_i^2}
 \]
By the equivalence of norms we see that
 \begin{align*}
 \E\left[ \abs{Z}^p\right]&=\E\left[\left(\sum_{i=1}^d z_i^2\right)^{p/2}\right]=\E\left[\left(\left(\sum_{i=1}^d z_i^2\right)^{1/2}\right)^{p}\right]\\
 &\le  \E\left[\left(d^{1/2-1/p}\left(\sum_{i=1}^d |z_i|^p\right)^{1/p}\right)^p\right]=d^{p/2-1}\E\left[\sum_{i=1}^d |z_i|^p\right]\le d^{p/2}\sup_{\{ i \in 1,\cdots,d\}}\E[|z_i|^p]
 \end{align*}
 Now note that $z_i=\sum_{k=0}^{\lfloor m/2\rfloor} \xi_{2k,i}$ is a sum of one-dimensional independent random variables. From an application of the Marcinkiewicz–-Zygmund inequality, see Lemma \ref{MZ_inequality}, 
 and the equivalence of norms
 we obtain
 \begin{align*}
  \E[\abs{z_i}^p]&=\E\left[\abs{\sum_{k=0}^{\lfloor m/2\rfloor} \xi_{2k,i}}^p\right]\le C_p \E\left[\left(\sum_{k=0}^{\lfloor m/2\rfloor} \xi_{2k,i}^2\right)^{p/2}\right]\\
 &= C_p \E\left[\left(\left(\sum_{k=0}^{\lfloor m/2\rfloor} \xi_{2k,i}^2\right)^{1/2}\right)^p\right]\le C_p \E\left[\left(m^{1/2-1/p}\left(\sum_{k=0}^{\lfloor m/2\rfloor} \abs{\xi_{2k,i}}^p\right)^{1/p}\right)^p\right]\\
 &=C_p m^{p/2-1}\E\left[\sum_k \abs{\xi_{2k,i}}^p\right]\le C_p m^{p/2}\sup_{k}\E[\abs{\xi_{2k,i}}^p].
 \end{align*}
 Overall, we have that
 \[
 \sup_{i \in \{1,\cdots,d\}}\E[\abs{z_i}^p] \lsim m^{p/2}C_{\xi},
 \]
 and hence the claim follows.
 \end{proof}

\subsection{Auxiliary results}

Note that an immediate consequence of Minorisation condition \ref{minorisation} is that the small measure $\nu$ is absolutely continuous with respect to the stationary measure $\pi$. Furthermore, from the minorisation condition it also follows that we can bound the expectations with respect to the measures as formulated in \cite[Lemma 1]{hobert}.
\begin{lemma}[\protect{\cite[Lemma 1]{hobert}}]
\label{pi_nu_bound} Let $(X_t)_{t \in T}$ be a positive Harris recurrent Markov process with invariant distribution $\pi$. Then for any  $\pi$-integrable function $g:E^{T} \rightarrow \mathbb{R}^d$ we have the following inequality holds
\begin{equation}
\EE_\pi \abs{g} \geq \alpha \pi(C) \ \EE_\nu \abs{g}, 
\end{equation}
 where $\alpha$ and $C$ are defined in $(\ref{minorization_skeleton})$.
\end{lemma}

Note that Lemma \ref{pi_nu_bound} was stated by \cite[Lemma 1]{hobert} in a one-dimensional setting, and a one-step minorisation condition. However, the claim can mutatis mutandis be stated for the general multidimensional $m_0$--step local minorisation case.

\subsection{ Preliminary results on Gaussian approximation}

Firstly, we will consider the following weak Gaussian approximation for bounded independent random vectors.

\begin{lemma}[\protect{\cite[Theorem 1]{eldan2020clt}}]
 \label{lemma_ind_weakGA} 
Suppose $x_1,\dots,x_n\in {\R}^d$ are independent identically distributed mean-zero bounded random vectors, i.e.,  $|x_{i}|\le \tau,$ for $i=1,\dots,n$ for some $\tau > 0$. Let $\sigma_d$ denote the maximal eigenvalue of $\Cov(\frac{1}{\sqrt{n}}\sum_i x_i)$. Then one can construct independent random vectors $(x_i^c)_i$ and $(y_i)_i$ in a richer probability space such that $x_i^c\stackrel{D}{=}x_i$ and $Y_n\sim \mathcal{N}_d(0,\Cov(\sum x_i))$ such that for all $\theta_0> 0$ we have that,
\[
\lim_{n\rightarrow \infty} \mathbb{P}\left(  \abs{\sum_{i=1}^n x_i - Y_n} \leq  \tau d^{1/2}(32+\log(n))\right)=1,    
\]

\end{lemma}

Note that this lemma is an immediate consequence of \cite[Theorem 1]{eldan2020clt} and the fact that convergence in the Wasserstein metric implies convergence of moments, see for example 
\cite[Theorem 6.9]{villani2009optimal}. Furthermore, note that we also have the following weak coupling inequality, which is a corollary of either \cite[Theorem 1.1]{zaitsev1987gaussian}  or \cite[Theorem 1.1]{zaitsev1987estimates}.

	\begin{lemma}[\protect{\cite[Fact 2.2]{einmahl1997gaussian}}]
 \label{lemma_ind_GA} 
		Suppose $x_1,\dots,x_n\in {\R}^d$ are independent random vectors that $\E[x_i]=0$ and $x_i=(x_{i,1},\dots,x_{i,d})^T$, in which $|x_i|$ are bounded such that $|x_{i}|\le \tau, i=1,\dots,n$. Let $X_n=\sum_{k=1}^n x_k$, then one can construct $X_n^c$ and $Y_n^c$ in a richer probability space such that $x_n^c\stackrel{D}{=} x_n$ and for all $\epsilon\ge 0$,
		\[
		\mathbb{P}\left(|X_n^c-Y_n^c|>\epsilon\right)\le C_2 d^{2}\exp\left(-\frac{\epsilon}{C_3d^{2}\tau}\right),
		\]
		where $Y_n^c$ is sum of $n$ i.i.d Gaussian vectors with the same mean and covariance matrix as $X_n^c$ and $C_2, C_3$ are positive dimension-independent constants.
	\end{lemma}

Note that \cite[Theorem 1.1]{zaitsev1987gaussian} and \cite[Theorem 1.1]{zaitsev1987estimates} give us a dimension dependence of $d^{5/2}$ and $d^2$ respectively in Lemma \ref{lemma_ind_GA}. As shown in \cite[Fact 2.2]{einmahl1997gaussian}, both results apply to bounded vectors. 
Furthermore, we will also consider the following strong Gaussian approximation for bounded independent random vectors. 
\begin{lemma}[\protect{\cite[Corollary 3]{zaitsev2013}}]
 \label{lemma_ind_strongGA} 
Suppose $x_1,\dots,x_n\in {\R}^d$ are independent mean-zero bounded random vectors, i.e.,  $|x_{i}|\le \tau,$ for $i=1,\dots,n$ for some $\tau \geq 1$. Let $\sigma_d$ denote the maximal eigenvalue of $\Cov(\frac{1}{\sqrt{n}}\sum_i x_i)$. Then one can construct independent random vectors $(x_i^c)_i$ and $(y_i)_i$ in a richer probability space such that $x_i^c\stackrel{D}{=}x_i$ and $y_i\sim \mathcal{N}_d(0,\Cov(x_i))$ such that for all $\theta_0> 0$ we have that,
\[
\mathbb{P}\left(\limsup_{n\rightarrow \infty}  \frac{1}{\log(n)}\abs{\sum_{i=1}^n x^c_i - \sum_{i=1}^n y_i} \leq  K_{\theta_0} \tau d^{23/4+\theta_0}\sqrt{\sigma_d}\log^*(d)  \right)=1,    
\]
where $K_{\theta_0}$ depends only on $\theta_0$ and $\log^*(d):=\max(1, \log(d)).$
\end{lemma}

 \begin{lemma} [\protect{\citet[Lemma 2.4]{merlevede2015}}] Let $B$ be a standard Brownian motion and $L$ be a Poisson process with intensity $\lambda$, independent of $B.$ Then there exists a standard Brownian motion $W$ that is also independent of $L$ such that 
\label{mer_lemma}
$$\abs{B(n)-\frac{1}{\sqrt{\lambda}} W(L(n))}=\bigO_{a.s.}(\log(n))$$
\end{lemma}
\begin{proof}
The claim immediately follows from \citet[Lemma 2.4]{merlevede2015} and a Borel--Cantelli argument.
\end{proof}

 \begin{lemma} [\protect{Ando--van Hemmen's inequality; \cite[Corollary 4.2]{van1980inequality}}]
 \label{Ando_inequality}
 Let $\Sigma_1,\Sigma_2$ be two positive definite matrices with the smallest eigenvalue bounded below by $\sigma_0>0$. Then for every $0<r \leq 1$ we have that
   \begin{equation*}  
   \displaystyle
    \abs{\Sigma_1^{r}-\Sigma_2^{r}} \leq \left(\frac{1}{\sigma_0}\right)^{1-r} \abs{\Sigma_1-\Sigma_2}.        \end{equation*}
 \end{lemma}
\bigskip

We will show that the strong approximation for bounded vectors given in Lemma \ref{lemma_ind_strongGA} can be extended to independent vectors with only $p$ finite moments. Furthermore, we obtain as an immediate consequence also the corresponding weak Gaussian approximation. The weak Gaussian approximation result for independent random vectors is comparable to the result of \cite{mies2023sequential}, which they obtain through martingale embeddings.

\begin{proposition}
\label{GA_iid}
Let $\{\xi_k\}$ be a sequence of independent and identically distributed mean zero random vectors in $\R^d$ such that  $\sup_{k \in \mathbb{N}}\sup_{i\in \{1,\cdots,d\}}\E[|\xi_{k,i}|^p]\leq C_{\xi}$, for some finite $C_{\xi}$, and for some $p>2$  and such that the smallest eigenvalue of $\Cov(\frac{1}{\sqrt{n}}\sum_{k=1}^n \xi_k)$ is larger than some constant $\sigma_0>0$.  Then
$$\left|\sum_{k=1}^n \xi_k- \sum_{k=1}^nY_k\right|=   \left\{
            \begin{array}{lll}                                          
                 \displaystyle \bigO_{P} \left(     d C_{\xi}^{1/p}  n^{1/p}\log(n)\right)  \\
                                   \textcolor{white}{.}   \\ 
              \displaystyle   \bigO_{a.s.} \left(  \left(   d^{25/4+\theta_0}\log^* (d)C_{\xi}^{1/p}\vee   d^{3/2} C_{\xi}^{2/p} \right) \left(\frac{\sigma_d}{\sigma_0}\right)^{1/2}n^{1/p}\log(n)\right)
                 
                \end{array},   
              \right.   
$$

where $\sum_{k=1}^nY_k$ has a Gaussian distribution with the same mean and covariance matrix as $\sum_{k=1}^n \xi_k.$     
\end{proposition}
\begin{proof}
We will first prove the almost sure bound. Let $\tau_n:=n^{1/p}, \delta=1/2$ and define the truncated sequence $\tilde{\xi}_k$ as follows
        \[
        \tilde{\xi}_k:= \xi_k \mathbbm{1}_{\{|\xi_k|\le d^{\delta}C_{\xi}^{1/p}\tau_n\}}-\E[\xi_k \mathbbm{1}_{\{|\xi_k|\le d^{\delta}C_{\xi}^{1/p}\tau_n\}}].
        \]
        Clearly, we have $|\tilde{\xi}_k|\le 2d^{1/2}\tau_n$ and $\E[\tilde{\xi}_k]=0$. 
 Then by triangle inequality, we have
        \begin{align*}
            \left|\sum_{j=1}^{n} \xi_j-\sum_{j=1}^{n} \tilde{\xi}_j\right|
        &\le  \left|\sum_{j=1}^{n} \xi_j\mathbbm{1}_{\{|\xi_j|>d^{\delta}C_{\xi}^{1/p}\tau_n)\}}\right| +\left|\sum_{j=1}^{n}\E[\xi_j\mathbbm{1}_{\{|\xi_j|>d^{\delta}C_{\xi}^{1/p}\tau_n\}}]\right|.
        \end{align*}

 Since $\sup_k\sup_i\E[|\xi_{k,i}|^p]]\leq C_{\xi}$ it follows from the equivalence of vector norms  that
            \begin{equation}
     \label{moment_norm}
\sup_k \EE [ \abs{\xi_k}^p]  \lsim d^{p/2} \sup_k\sup_i\E[|\xi_{k,i}|^p]  \lsim d^{p/2}C_{\xi}.   
    \end{equation}

 For $r\in [1,p)$ we have that
 \begin{align}
 \label{exact_xi_tail_bound}
\E \abs{\xi_j}^r\mathbbm{1}_{\abs{\xi_j}>d^{\delta}C_{\xi}^{1/p}\tau_n}&=r\int_{d^{\delta}C_{\xi}^{1/p}\tau_n}^\infty \hspace{-0.2 cm}s^{r-1}\P(\abs{\xi_j}>s)ds+d^{\delta r}C_{\xi}^{r/p}\tau^r_n \P(\abs{\xi_j}>d^{\delta}C_{\xi}^{1/p}\tau_n) \nonumber\\
&\le r \int_{d^{\delta}C_{\xi}^{1/p}\tau_n}^\infty \frac{\E\abs{\xi_j}^p}{s^{p-r+1}}ds+ d^{\delta r}C_{\xi}^{r/p}\tau^r_n  \frac{\E\abs{\xi_j}^p}{d^{p\delta }\tau_n^pC_{\xi}}\nonumber\\ 
&= \frac{r}{(p-r)} \frac{\E \abs{\xi_j}^p}{(d^{\delta}C_{\xi}^{1/p}\tau_n)^{(p-r)}} + d^{\delta r}C_{\xi}^{r/p}\tau^r_n  \frac{\E\abs{\xi_j}^p}{d^{p\delta }\tau_n^pC_{\xi}} \nonumber\\
&\lsim  d^{\delta r}\tau_n^{r-p} C_{\xi}^{r/p}.
 \end{align} 
Taking $r=1$ in (\ref{exact_xi_tail_bound})  it follows that $\left(\abs{\xi_j}\mathbbm{1}_{\{\abs{\xi_j}>d^{\delta}C_{\xi}^{1/p}\tau_n\}}\right)_{j \in \mathbbm{N}}$ is an i.i.d.~sequence with mean $\mu_{n,d}> 0$ such that $\mu_{n,d} \lsim d^{\delta} \tau_n^{1-p}C_{\xi}^{r/p}.$ We claim that 
\begin{align}
\label{slln_claim1}
\sum_{j=1}^{n} \abs{\xi_j}\mathbbm{1}_{\{|\xi_j|>d^{\delta}C_{\xi}^{1/p}\tau_n)\}}&=\bigO_{a.s.}(d^{\delta}C_{\xi}^{1/p} n\tau^{1-p}_n)   
\end{align}
Note that since $d^{-\delta} C_{\xi}^{-1/p}{\mu_{n,d}}$ tends to zero for any $\tau_n$ that tends to infinity, the claim does not follow by a standard law of large numbers for triangular arrays. We will prove \eqref{slln_claim1} by following a truncation and sub-sequence argument. Firstly, note that for any $\theta>1/p$ we have that 
\begin{equation}
\label{claim_decomp}
\sum_{j=1}^{n} \abs{\xi_j} \mathbbm{1}_{\{|\xi_j|>d^{\delta}C_{\xi}^{1/p}\tau_n)\}}=\sum_{j=1}^{n} \abs{\xi_j}\mathbbm{1}_{\{ d^{\delta}C_{\xi}^{1/p}\tau_n \leq|\xi_j|\leq d^{\delta}C_{\xi}^{1/p}n^\theta)\}}+\sum_{j=1}^{n} \abs{\xi_j}\mathbbm{1}_{\{|\xi_j|>d^{\delta}C_{\xi}^{1/p}n ^\theta)\}}.
\end{equation}

For the second term in \eqref{claim_decomp}, we see that if we choose
\begin{equation}
\label{theta_condition}
\theta> \frac{2p-1}{p(p-1)},    
\end{equation}
then for every $\varepsilon>0$ we have that
\begin{align*}
\sum_{n=1}^\infty \P \left( \abs{\sum_{j=1}^{n} \abs{\xi_j}\mathbbm{1}_{\{|\xi_j|>d^{\delta}C_{\xi}^{1/p}n^\theta\} } } > \varepsilon d^{\delta}C_{\xi}^{1/p} n^{1/p} \right) &\leq \sum_{n=1}^\infty  \frac{ n\E [\abs{\xi_j}\mathbbm{1}_{\{|\xi_j|>d^{\delta}C_{\xi}^{1/p}n^\theta \}}]} {\varepsilon d^{\delta}C_{\xi}^{1/p} n^{1/p}}\\
 &\lsim \frac{1}{\varepsilon} \sum_{n=1}^\infty {n^{\theta(1-p)+1-1/p}} < \infty,
\end{align*}

since we have mutatis mutandis to \eqref{exact_xi_tail_bound} that $ \E [\abs{\xi_j}\mathbbm{1}_{\{|\xi_j|>d^{\delta}C_{\xi}^{1/p}n^\theta\}}]\lsim d^{\delta}C_{\xi}^{1/p}n^{\theta(1-p)}$ for $p>2$. Hence by the Borel--Cantelli Lemma we have that 
\begin{equation}
\label{decomp_term2}
    \sum_{j=1}^{n} \abs{\xi_j}\mathbbm{1}_{\{|\xi_j|>d^{\delta}C_{\xi}^{1/p}n^\theta)\}}=\bigO_{a.s.}(d^{\delta}C_{\xi}^{1/p}n^{1/p}).
\end{equation}

Note that the first term in the RHS of \eqref{claim_decomp} is bounded above by 
\begin{align}
\label{claim_decomp2}
  &\abs{ \sum_{j=1}^{n} \left( \abs{\xi_j}\mathbbm{1}_{\{ d^{\delta}C_{\xi}^{1/p}\tau_n \leq|\xi_j|\leq d^{\delta}C_{\xi}^{1/p}n^\theta)\}}-\E \left[\abs{\xi_j}\mathbbm{1}_{\{ d^{\delta}C_{\xi}^{1/p}\tau_n \leq|\xi_j|\leq d^{\delta}C_{\xi}^{1/p}n^\theta)\}} \right] \right) } \nonumber \\
  &+  \sum_{j=1}^{n} \E \left[\abs{\xi_j}\mathbbm{1}_{\{ d^{\delta}C_{\xi}^{1/p}\tau_n \leq|\xi_j|\leq d^{\delta}C_{\xi}^{1/p}n^\theta)\}} \right]
\end{align}

From \eqref{exact_xi_tail_bound} it also follows that 
\begin{align}
\label{bounded_twosided_truncated}
    \sum_{j=1}^{n} \E \left[\abs{\xi_j}\mathbbm{1}_{\{ d^{\delta}C_{\xi}^{1/p}\tau_n \leq|\xi_j|\leq d^{\delta}C_{\xi}^{1/p}n^\theta)\}} \right]&=\bigO(d^{\delta}C_{\xi}^{1/p}n\tau_n^{1-p})
\end{align}

Now we will show that 
\begin{equation}
\label{truncated_slln}
    \abs{ \sum_{j=1}^{n} \left( \abs{\xi_j}\mathbbm{1}_{\{ d^{\delta}C_{\xi}^{1/p}\tau_n \leq|\xi_j|\leq d^{\delta}C_{\xi}^{1/p}n^\theta)\}}-\E \left[\abs{\xi_j}\mathbbm{1}_{\{ d^{\delta}C_{\xi}^{1/p}\tau_n \leq|\xi_j|\leq d^{\delta}C_{\xi}^{1/p}n^\theta)\}} \right] \right) }=o_{a.s.}(d^{\delta}C_{\xi}^{1/p}n^{1/p}).
\end{equation}
Which is equivalent to showing
\begin{equation}
\label{truncated_slln2}
    \abs{ \frac{1}{n^{1/p}C_{\xi}^{1/p}  d^{\delta}}\sum_{j=1}^{n} \left( \abs{\xi_j}\mathbbm{1}_{\{ d^{\delta}C_{\xi}^{1/p}\tau_n \leq|\xi_j|\leq d^{\delta}C_{\xi}^{1/p}n^\theta)\}}-\E \left[\abs{\xi_j}\mathbbm{1}_{\{ d^{\delta}C_{\xi}^{1/p}\tau_n \leq|\xi_j|\leq d^{\delta}C_{\xi}^{1/p}n^\theta)\}} \right] \right) }\xrightarrow[]{a.s.}0 \ as \  n\rightarrow \infty.
\end{equation}

We will first show that \eqref{truncated_slln2} holds for a subsequence. Fix any $\lambda>1$ and define $n(\ell)=\floor{\lambda^\ell}$. Then, we claim that as $\ell \rightarrow \infty$, we have
\begin{equation}
\label{truncated_slln_subseq}
    \abs{ \frac{1}{{n(\ell)^{1/p}}C_{\xi}^{1/p}  d^{\delta}}\sum_{j=1}^{{n(\ell)}} \left( \abs{\xi_j}\mathbbm{1}_{\{ d^{\delta}C_{\xi}^{1/p}\tau_{n(\ell)} \leq|\xi_j|\leq d^{\delta}C_{\xi}^{1/p}{n(\ell)}^\theta)\}}-\E \left[\abs{\xi_j}\mathbbm{1}_{\{ d^{\delta}C_{\xi}^{1/p}\tau_{n(\ell)} \leq|\xi_j|\leq d^{\delta}C_{\xi}^{1/p}{n(\ell)}^\theta)\}} \right] \right) }\xrightarrow[]{a.s.}0.
\end{equation}

For any $\varepsilon>0$ we have that
\begin{align*}
    &\sum_{\ell=1}^\infty \P \left( \abs{ \frac{1}{{n(\ell)^{1/p}} C_{\xi}^{1/p} d^{\delta}}\sum_{j=1}^{{n(\ell)}}  \left( \abs{\xi_j}\mathbbm{1}_{\{ d^{\delta}C_{\xi}^{1/p}\tau_{n(\ell)} \leq|\xi_j|\leq d^{\delta}C_{\xi}^{1/p}{n(\ell)}^\theta)\}}-\E \left[\abs{\xi_j}\mathbbm{1}_{\{ d^{\delta}C_{\xi}^{1/p}\tau_{n(\ell)} \leq|\xi_j|\leq d^{\delta}C_{\xi}^{1/p}{n(\ell)}^\theta)\}} \right] \right) } > \varepsilon \right)  \\
     &\leq\frac{1}{\varepsilon^2}\sum_{\ell=1}^\infty      \frac{1}{{n(\ell)^{2/p}}C_{\xi}^{2/p}  d^{2\delta } }\sum_{j=1}^{{n(\ell)}} \Var \left(  \abs{\xi_j}\mathbbm{1}_{\{ d^{\delta}C_{\xi}^{1/p}\tau_{n(\ell)} \leq|\xi_j|\leq d^{\delta}C_{\xi}^{1/p}{n(\ell)}^\theta)\}}  \right)     \\
      &\leq\frac{1}{\varepsilon^2}\sum_{\ell=1}^\infty      \frac{n(\ell)^{1-2/p}}{ d^{2 \delta} C_{\xi}^{2/p} }  \E \left[  \abs{\xi_1}^2\mathbbm{1}_{\{ d^{\delta}C_{\xi}^{1/p}\tau_{n(\ell)} \leq|\xi_1|\leq d^{\delta}C_{\xi}^{1/p}{n(\ell)}^\theta)\}}  \right]. 
\end{align*}

By Tonelli's theorem, we have that
\begin{align*}
    \sum_{\ell=1}^\infty     \frac{n(\ell)^{1-2/p}}{  d^{2 \delta} C_{\xi}^{2/p}}  \E \left[  \abs{\xi_1}^2\mathbbm{1}_{\{ d^{\delta}C_{\xi}^{1/p}\tau_{n(\ell)} \leq|\xi_1|\leq d^{\delta}C_{\xi}^{1/p}{n(\ell)}^\theta)\}}  \right] = \E \left[  \abs{\xi_1}^2  \sum_{\ell=1}^\infty     \frac{n(\ell)^{1-2/p}}{   d^{2 \delta } C_{\xi}^{2/p}}   \mathbbm{1}_{\{ d^{\delta}C_{\xi}^{1/p}\tau_{n(\ell)} \leq|\xi_1|\leq d^{\delta}C_{\xi}^{1/p}{n(\ell)}^\theta)\}}   \right] 
\end{align*}

Since we have that
\begin{align*}    
    \sum_{k=1}^\infty     {n(\ell)^{1-2/p}}  \mathbbm{1}_{\{ d^{\delta} C_{\xi}^{1/p}{n(\ell)}^{1/p}\leq|\xi_1|\leq d^{\delta}C_{\xi}^{1/p}{n(\ell)}^\theta)\}}
   &= 
    \sum_{k=1}^\infty     {n(\ell)^{1-2/p}}  \mathbbm{1}_{\{ d^{ - {\delta}{ /\theta}} C_{\xi}^{-1/p\theta} \abs{\xi_1}^{1/\theta}   \leq n(\ell) \leq\ d^{- \delta p} C_{\xi}^{-1}\abs{\xi_1}^{p} \} }\\ 
       &\leq 
    \sum_{k=1}^\infty     {\lambda^{(1-2/p)\ell}}  \mathbbm{1}_{\{  d^{ -\delta/\theta} C_{\xi}^{-1/p\theta}\abs{\xi_1}^{\frac{1}{\theta}}   \leq \lambda^\ell \leq\ d^{- \delta p} C_{\xi}^{-1} \abs{\xi_1}^{p} +1 \} }.
    \end{align*}
    This gives
    \begin{align*}
        &\sum_{\ell=1}^\infty     {\lambda^{(1-2/p)\ell}}  \mathbbm{1}_{\{  d^{ - {\delta}/ {\theta} } C_{\xi}^{{-1}/{p\theta}}\abs{\xi_1}^{\frac{1}{\theta}}   \leq \lambda^\ell \leq\ d^{- \delta p} C_{\xi}^{-1}\abs{\xi_1}^{p} +1 \} }\\   &\leq 
    \sum_{\ell=1}^\infty     {\lambda^{(1-2/p)\ell}}  \mathbbm{1}_{\{ \floor{ \ln (  d^{ - \frac{\delta}{\theta} }C_{\xi}^{{-1}/{p\theta}}\abs{\xi_1}^{\frac{1}{\theta}}) / \ln(\lambda) }  \leq \ell \leq\  \ceil{\ln ( d^{- \frac{p}{2}} C_{\xi}^{-1} \abs{\xi_1}^{p} +1 ) / \ln(\lambda) }\}}\\ 
      &= 
    \sum_{\ell= \floor{\ln (  d^{ -\delta/\theta} C_{\xi}^{{-1}/{p\theta}}\abs{\xi_1}^{1/\theta}) / \ln(\lambda) }}^{\ceil{\ln ( d^{- \delta p} C_{\xi}^{-1}\abs{\xi_1}^{p} +1 ) / \ln(\lambda)  }}     {e^{(1-2/p)\ell)\ln(\lambda)}}\\
       &\leq   
    \int_{  \ln (  d^{ -\delta/\theta}C_{\xi}^{{-1}/{p\theta}} \abs{\xi_1}^{1/\theta}) / \ln(\lambda) }^{{\ln ( d^{- \delta p} C_{\xi}^{-1} \abs{\xi_1}^{p} +1 ) / \ln(\lambda) +2  }}     {e^{(1-2/p)x)\ln(\lambda)}dx}=  \frac{1}{\ln(\lambda)}  
    \int_{  \ln (  d^{ -\delta/\theta} C_{\xi}^{{-1}/{p\theta}} \abs{\xi_1}^{1/\theta})  }^{{\ln ( d^{- \delta p} C_{\xi}^{{-1} } \abs{\xi_1}^{p}    +1)+2\ln(\lambda)  }}     {e^{(1-2/p)x}dx}\\
    &\leq  C_1 
    (d^{- \delta p} C_{\xi}^{-1}\abs{\xi_1}^{p}+1)^{1-2/p} -  C_2( d^{ -\delta/\theta} C_{\xi}^{{-1}/{p\theta}}\abs{\xi_1}^{1/\theta})^{(1-2/p)}\\  
     &\leq  2C_1 
    d^{\delta(2-p) } C_{\xi}^{2/p-1} \abs{\xi_1}^{p-2 } -  C_2d^{ -(1-2/p) \delta /\theta} C_{\xi}^{-(p-2)/p^2\theta} \abs{\xi_1}^{(1-2/p)/\theta},
\end{align*}
with $C_1=\frac{p\lambda^{2(1-2/p)}}{(p-2)\ln(\lambda)}$ and $C_2=(\ln \lambda)^{-1}p/(p-2) $. Here the last inequality follows since on the event $\{\abs{\xi_1}\geq d^{\delta}C_{\xi}^{1/p}\tau_{n(\ell)}\}$, we have that $\abs{\xi_1}^{p} C_{\xi}^{-1} d^{-\delta p}\geq 1$ and hence $( d^{-\delta p} C_{\xi}^{-1}  \abs{\xi_1}^{p} +1)^{\alpha}\leq 2 d^{-\alpha \delta p}C_{\xi}^{-\alpha}  \abs{\xi_1}^{\alpha p} $ for any $\alpha>0$. 
From \eqref{exact_xi_tail_bound} it follows that for $p>2$ and $\theta $ satisfying \eqref{theta_condition} we have that
\begin{align*}
    \sum_{\ell=1}^\infty     \frac{n(\ell)^{1-2/p}}{  C_{\xi}^{2/p} d^{2\delta} }  \E \left[  \abs{\xi_1}^2\mathbbm{1}_{\{ d^{\delta}C_{\xi}^{1/p}\tau_{n(\ell)} \leq|\xi_1|\leq d^{\delta}C_{\xi}^{1/p}{n(\ell)}^\theta)\}}  \right] &\leq    2C_1 
    d^{- \delta p  } C_{\xi}^{-1}\E \abs{\xi_1}^{p }\\
    &\hspace{-1.7cm}+  C_2d^{ -(1-2/p)\delta/\theta-2\delta} C_{\xi}^{-(p-2)/p^2\theta-2/p} \E \abs{\xi_1}^{2+(1-2/p)/\theta} < \bar{C}, 
\end{align*}
where $\bar{C}$ is some universal constant, since by \eqref{moment_norm} the first term can be bounded by a universal constant and since for $\theta $ satisfying \eqref{theta_condition} and $p>2$ we have that $1<2+(1-2/p)/\theta<p$, and thus by Jensen's inequality we can also bound the second term with a universal constant. Consequently, the claim \eqref{truncated_slln_subseq} follows from the Borel--Cantelli lemma. Moreover, from \eqref{bounded_twosided_truncated} we also have that
\begin{align*}
        \abs{ \sum_{j=1}^{n(\ell)} \left( \abs{\xi_j}\mathbbm{1}_{\{ d^{\delta}C_{\xi}^{1/p}\tau_{n(\ell)} \leq|\xi_j|\leq d^{\delta}C_{\xi}^{1/p}{n(\ell)}^\theta)\}} \right) }&=o_{a.s.}(d^{\delta}C_{\xi}^{1/p}n(\ell)^{1/p})+ \sum_{j=1}^{n(\ell)}\E \left[\abs{\xi_j}\mathbbm{1}_{\{ d^{\delta}C_{\xi}^{1/p}\tau_{n(\ell)} \leq|\xi_j|\leq d^{\delta}C_{\xi}^{1/p}{n(\ell)}^\theta)\}} 
        \right]\\
        &=o_{a.s.}(d^{\delta}C_{\xi}^{1/p}n(\ell)^{1/p})+  
        \bigO_{a.s.}(d^{\delta}C_{\xi}^{1/p}n(\ell)\tau_{n(\ell)}^{1-p})\\
        &=         \bigO_{a.s.}(d^{\delta}C_{\xi}^{1/p}n(\ell)^{1/p}) = \bigO_{a.s.}(d^{\delta}C_{\xi}^{1/p}\lambda^{\ell/p}).
\end{align*}
Since for any $n$ there exists an $\ell$ such that $\lambda^\ell< n \leq \lambda^{\ell+1},$ we have that 
$$n^{-1/p} \sum_{j=1}^{n}  \abs{\xi_j}\mathbbm{1}_{\{ d^{\delta}C_{\xi}^{1/p}n^{1/p} \leq|\xi_j|\leq d^{\delta}C_{\xi}^{1/p}n^\theta)\}}\leq \lambda^{-\ell/p}\sum_{j=1}^{\lambda^{\ell+1}} \abs{\xi_j}\mathbbm{1}_{\{ d^{\delta}C_{\xi}^{1/p}\lambda^{(\ell+1)/p} \leq|\xi_j|\leq d^{\delta}C_{\xi}^{1/p}\lambda^{(\ell+1)\theta}\}}$$

\begin{align*}
\limsup_{n \rightarrow \infty} n^{-1/p} \sum_{j=1}^{n}  \abs{\xi_j}\mathbbm{1}_{\{ d^{\delta}C_{\xi}^{1/p}n^{1/p} \leq|\xi_j|\leq d^{\delta}C_{\xi}^{1/p}n^\theta)\}} &\leq \frac{\lambda^{1/p}}{\lambda^{(\ell+1)/p}}\sum_{j=1}^{\lambda^{\ell+1}} \abs{\xi_j}\mathbbm{1}_{\{ d^{\delta}C_{\xi}^{1/p}\lambda^{(\ell+1)/p} \leq|\xi_j|\leq d^{\delta}C_{\xi}^{1/p}\lambda^{\ell\theta}\}}\\
&\lsim_{a.s.}d^{\delta}C_{\xi}^{1/p}.
\end{align*}

Thus we have proven the claim \eqref{truncated_slln} and consequently also \eqref{slln_claim1} immediately follows.
Putting all the obtained results together, we have that 
        \begin{align}
        \label{truncate_error1}
         \left|\sum_{j=1}^{n} \xi_j-\sum_{j=1}^{n} \tilde{\xi}_j\right| &= \bigO_{a.s.}(d^{\delta} C_{\xi}^{1/p}n\tau^{1-p}_n)  \nonumber \\
         &=\bigO_{a.s.}(d^{1/2}C_{\xi}^{1/p} n^{1/p}).
         \end{align}

Next, since $\E[\tilde{\xi}_j]=0$ and $\abs{\tilde{\xi}_j}\le  2d^{1/2}C_{\xi}^{1/p}\tau_n$, by Lemma \ref{lemma_ind_strongGA} there exists a sequence of Gaussian random variables $\tilde{y}_1,\dots,\tilde{y}_n$ with $\tilde{y}_j\sim \mathcal{N}_d(0,\Cov(\tilde{\xi}_i))$ such that for all $\theta_0>0$ we have that

\begin{align}
 \label{truncate_error2}
    \left| \sum_{j=1}^{n} \tilde{\xi}_j - \sum_{j=1}^n \tilde{y}_j \right|&= \bigO_{a.s.}( (d^{23/4+\delta+\theta_0}\log^* d) \sqrt{\sigma_d} C^{1/p}_{\xi} \tau_n {\log n})) \nonumber  \\
    &= \bigO_{a.s.}( (d^{25/4 +\theta_0 }\log^* d) \sqrt{\sigma_d}  C^{1/p}_{\xi}  n^{1/p} {\log n})),    
\end{align}
for any $\theta_0>0$.
Hence by (\ref{truncate_error1}) and (\ref{truncate_error2}) it follows that 
\begin{align}
\label{truncate_error3}
\left|\sum_{k=1}^n \xi_k  - \tilde{Y}_n\right| &\le \left|\sum_{k=1}^n \xi_k - \sum_{j=1}^{n} \tilde{\xi}_j \right| + \left| \sum_{j=1}^{n} \tilde{\xi}_j - \tilde{Y}_n \right| \nonumber\\ &=\bigO_{a.s.}(d^{1/2}C^{1/p}_{\xi,d}n^{1/p}) +\bigO_{a.s.}( (d^{25/4 +\theta_0}\log^* d) \sqrt{\sigma_d} C^{1/p}_{\xi}  n^{1/p} {\log n})),
\end{align}

where $\tilde{Y}_n:= \sum_{j=1}^n \tilde{y}_j$ has the same covariance matrix as $\sum_{j=1}^{n} \tilde{\xi}_j$. Note that we can write $\tilde{y}_j$ in terms of a standard multidimensional Gaussian as \[\tilde{y}_j =\left(\Cov(\tilde{\xi}_1)\right)^{1/2}z_j \ \ \textrm{for } j=1,\dots,n,\]
where $z_1,\dots,z_n\sim  \mathcal{N}_d(0,I_d)$ are all independent and identically distributed. Similarly, using the same Gaussian sequence $z_1,\dots,z_n$, we can define for $j=1,\dots,n$ the following
 \[
 y_j:=\left(\Cov({\xi}_1)\right)^{1/2}z_j.
 \]

Now it follows that
\begin{align*}
\abs{\sum_{j=1}^n\tilde{y}_j-\sum_{j=1}^n{y}_j}&= \abs{ \Cov\left( \xi_1\right)^{1/2}-\Cov( \tilde{\xi}_1)^{1/2} } \abs{\sum_{j=1}^n z_j} \\
 & \le \frac{1}{\sqrt{\sigma_0}}\abs{ \Cov\left( \xi_1\right)-\Cov( \tilde{\xi}_1) } \abs{\sum_{j=1}^n z_j}\\
&\le \frac{ \abs{\sum_{j=1}^n z_j}}{\sqrt{\sigma_0}} \left[2\sqrt{\E\left[\left| \xi_1\right|^2\right]}\sqrt{\E\left[\left| \xi_1-\tilde{\xi}_1\right|^2\right]}+\E\left[\left| \xi_1-\tilde{\xi}_1\right|^2\right]\right],    
\end{align*}
where the first inequality follows by Lemma \ref{Ando_inequality}; Ando--{van Hemmen}'s inequality. \\
By a component-wise application of the law of iterated logarithm we have that \[\abs{\sum_{j=1}^n z_j}=\bigO_{a.s.}\left( \sqrt{d n \log \log n}\right).\]

Note that $\mathbb{E}[\abs{\xi_1}^2]\leq \tr(\Sigma)\leq d\sigma_d.$ Furthermore, from \eqref{exact_xi_tail_bound} we see that

\[\E\left[\left| \xi_1-\tilde{\xi}_1\right|^2 \right]\le \E \left[\abs{\xi_1}^2\mathbbm{1}_{\{|\xi_1|>d^{\delta}C_{\xi}^{1/p}\tau_n\}}\right]\lsim d^{2\delta} C_{\xi}^{2/p} \tau_n^{2-p}.\]

We see that
\begin{align}
\label{truncate_y_variance}
    \abs{\sum_{j=1}^n\tilde{y}_j-\sum_{j=1}^n{y}_j}&=\bigO_{a.s.}\left( {d^{(1+\delta )\wedge (2\delta+1/2)}}\left( \frac{\sigma_d}{\sigma_0} \right)^{1/2}   C_{\xi}^{2/p} \tau_n^{2-p} \sqrt{n \log \log n}\right) \nonumber \\
    &=\bigO_{a.s.}\left( {d^{3/2}}\left( \frac{\sigma_d}{\sigma_0} \right)^{1/2}C_{\xi}^{2/p} n^{2/p-1} \sqrt{ \log \log n}\right),
\end{align}
where the last equality follows from the choice $\tau_n=n^{1/p}$ and $\delta=1/2$. From (\ref{truncate_y_variance})  and (\ref{gaus_approx_error})  it follows that
\begin{align*}
\left|\sum_{j=1}^n \xi_j  - \sum_{j=1}^n{y}_j\right|&\le \left|\sum_{j=1}^n \xi_j  - \sum_{j=1}^n\tilde{y}_j\right| +\left|\sum_{j=1}^n{y}_j-Y_n\right|\\
&=\bigO_{a.s.} \left(  \left(   d^{25/4+\theta_0}\log^* (d)C_{\xi}^{1/p}\vee   d^{3/2} C_{\xi}^{2/p} \right) \left(\frac{\sigma_d}{\sigma_0}\right)^{1/2}n^{1/p}\log(n)\right).
   \end{align*}

This finishes the proof for the strong Gaussian approximation case.\\

For the weak Gaussian approximation case the proof follows completely analogously up to some minor differences. Firstly, we consider an application of Lemma \ref{lemma_ind_weakGA} instead of Lemma \ref{lemma_ind_strongGA} for approximating the truncated sequence $(\tilde{\xi_k})_{k \in \mathbb{N}}$. This gives us
\begin{align}
 \label{weak_truncate_error2}
    \left| \sum_{j=1}^{n} \tilde{\xi}_j - \tilde{Y}_n\right|&= \bigO_{P}( d^{\delta+1/2} C^{1/p}_{\xi} \tau_n {\log n}),  
 \end{align}

where $\tilde{Y}_n$  has a mean-zero Gaussian distribution with covariance matrix $\Cov(\sum_{j=1}^{n} \tilde{\xi}_j)$.

By the same argument used to obtain (\ref{truncate_error3}), we see that

\begin{align}
\label{weak_truncate_error3}
\left|\sum_{k=1}^n \xi_k  - \tilde{Y}_n\right| &\le \left|\sum_{k=1}^n \xi_k - \sum_{j=1}^{n} \tilde{\xi}_j \right| + \left| \sum_{j=1}^{n} \tilde{\xi}_j - \tilde{Y}_n \right| \nonumber\\ &=\bigO_{P}(d^{\delta} C_{\xi}^{1/p}n\tau^{1-p}_n)+\bigO_{P}( d^{\delta+1/2} C^{1/p}_{\xi} \tau_n {\log n}).
\end{align}

Furthermore, the difference between the Gaussian approximation for the truncated sums and the Gaussian approximation for the true sums can now be obtained with the following argument.
$Z\sim\mathcal{N}(0, I_d)$ as follows:
\[
\tilde{Y}_n:=\tilde{\Sigma}_n^{1/2}Z \ \textrm{with} \ \tilde{\Sigma}_n=\Cov\left(\sum_{j=1}^{n} \tilde{\xi}_j\right).
\]
Similarly, using the same Gaussian $Z$, we define 
\[
Y_n:=\Sigma_n^{1/2}Z \ \textrm{with} \ {\Sigma}_n=\Cov\left(\sum_{k=1}^n \xi_k\right)
\]

Firstly, note that we have
\[
\E[|\Sigma_n^{1/2}Z-\tilde{\Sigma}_nZ|^2]=\E[\tr((\Sigma_n^{1/2}Z-\tilde{\Sigma}_nZ)(\Sigma_n^{1/2}Z-\tilde{\Sigma}_nZ)^T)]=\E[\tr((\Sigma_n^{1/2}-\tilde{\Sigma}_n^{1/2})ZZ^T(\Sigma_n^{1/2}-\tilde{\Sigma}_n^{1/2}))]
\]
By linearity of the trace it immediately follows that
\[
\E\left[\abs{\Sigma_n^{1/2}Z-\tilde{\Sigma}_nZ}^2\right]=\tr((\Sigma_n^{1/2}-\tilde{\Sigma}_n^{1/2})\E[ZZ^T](\Sigma_n^{1/2}-\tilde{\Sigma}_n^{1/2})),
\]
 where $\E[ZZ^T]=I_d$. Therefore, by an application of \cite[Lemma 1]{eldan2020clt} we have that
\[
\tr((\sqrt{\Sigma_n}-\sqrt{\tilde{\Sigma}_n})^2)\le \tr((\Sigma_n-\tilde{\Sigma}_n)^2\Sigma_n^{-1}).
\]
By the cyclic invariance of the trace we have that
\[
\tr((\Sigma_n-\tilde{\Sigma}_n)^2\Sigma_n^{-1}) = \tr(\Sigma_n^{-1/2}(\Sigma_n-\tilde{\Sigma}_n)^2\Sigma_n^{-1/2})=\abs{(\Sigma_n-\tilde{\Sigma}_n)\Sigma_n^{-1/2}}^2.
\]
Since $|\Sigma_n\tilde{\Sigma}_n|\le |\tilde{\Sigma}_n||\Sigma_n|_{*}$ it follows that
\begin{align}
\label{diff_cov_1}
|(\Sigma_n-\tilde{\Sigma}_n)\Sigma_n^{-1/2}|^2 \le |\Sigma_n-\tilde{\Sigma}_n|^2 |\Sigma_n^{-1/2}|_{*}^2\le \frac{1}{n\sigma_0}|\Sigma_n-\tilde{\Sigma}_n|^2.
\end{align}

\begin{align*}
\abs{\Sigma_n-\tilde{\Sigma}_n}\le \left[\sqrt{\E\left[\left|\sum_{k=1}^n \xi_k\right|^2\right]}\sqrt{\E\left[\left|\sum_{k=1}^n \xi_k-\sum_{k=1}^{n} \tilde{\xi}_k\right|^2\right]}+\E\left[\left|\sum_{k=1}^n \xi_k-\sum_{k=1}^{n} \tilde{\xi}_k\right|^2\right]\right]
\end{align*}

 By following the argument of Lemma \ref{order_norm_power} and an application of Jensen's inequality it follows that 

\[
\sqrt{\E\left[\left|\sum_{k=1}^n \xi_k\right|^2\right]}\lsim d^{1/2}C_{\xi}^{1/p}n^{1/2}.
\]

Furthermore, by \eqref{exact_xi_tail_bound} we also have that
\begin{align*}
\E\left[\left|\sum_{k=1}^n \xi_k-\sum_{k=1}^{n} \tilde{\xi}_k\right|^2\right]&\le \sum_{k=1}^{n}\E[\abs{\xi_k}^2\mathbbm{1}_{\{|\xi_k|>d^{\delta}C_{\xi}^{1/p}\tau_n\}}]\\
&\lsim d^{2\delta} C_{\xi}^{2/p} n\tau_n^{2-p}
\end{align*}

This gives 
\begin{align}
\label{diff_cov_2}
&|\Sigma_n-\tilde{\Sigma}_n| = \bigO(d^{1/2}C_{\xi}^{1/p}n^{1/2})\sqrt{\bigO(d^{2\delta} C_{\xi}^{2/p} n\tau_n^{2-p})}+\bigO(d^{2\delta} C_{\xi}^{2/p} n\tau_n^{2-p})
\end{align}

Now we choose $\tau_n=n^{1/p}$ and $\delta=1/2$ we see that

\[
\E\left[\left|\sum_{k=1}^n( \xi_k- \tilde{\xi}_k)\right|^2\right]= \bigO(dC_{\xi}^{2/p}n^{2/p}),
\]

and consequently from \eqref{diff_cov_1} and \eqref{diff_cov_2} we see that
\begin{align}
\label{weak_gaus_approx_error}
    \E\left[\abs{\Sigma_n^{1/2}Z-\tilde{\Sigma}_nZ}\right]=\bigO_{P}(d C_{\xi}^{1/p} n^{1/p}) + \bigO_{P}(d C_{\xi}^{2/p}n^{2/p-1/2})=\bigO_{P}(dC_{\xi}^{1/p}n^{1/p}),
\end{align}
where the last equality holds since $p>2.$ Note that by Markov's inequality it immediately follows that
\begin{align}
\label{weak_gaus_approx_error4}
\abs{Y_n-\tilde{Y}_n}=\bigO_{P}(dC_{\xi}^{1/p}n^{1/p})
\end{align}

From (\ref{weak_truncate_error3})  and (\ref{weak_gaus_approx_error4})  it follows that
\[
\left|\sum_{k=1}^n \xi_k  -Y_n\right|\le \left|\sum_{k=1}^n \xi_k  - \tilde{Y}_n\right| +|Y_n-\tilde{Y}_n| =\bigO_{P} \left(     dC_{\xi}^{1/p}  n^{1/p}\log(n)\right).
\]
Finally, writing $Y_n$ as $\sum_{k=1}^n y_k$ completes the proof. Note that we can always, on a possibly extended probability space, do this since all conditional distributions of $(y_1,\dots,y_n,Y_n)$ are known.\\
\end{proof}

\begin{theorem}
\label{GA_one_dep}
Let $\{\xi_k\}$ be a one-dependent identically distributed zero-mean sequence such that  finite $p$-th moment for each element $ \sup_k\sup_i\E[|\xi_{k,i}|^p]\leq C_{\xi}$ for some $p>2$  and that the smallest eigenvalue of $\Sigma_\xi:=\Var(\xi)+ \Cov(\xi_1,\xi_2)+ \Cov(\xi_1,\xi_2)^T$ is larger than some constant $\sigma_0>0$  and $\sup_{i,j} \abs{\Sigma_{\xi_{ij}}}<\infty$ then there exists a Brownian motion $B$ such that
$$\left|\sum_{k=1}^n \xi_k- \Sigma_\xi^{1/2}W(n)\right|=    \left\{
            \begin{array}{lll}                                          
                 \displaystyle \bigO_{P} \left(    d C_{\xi}^{2/p}  n^{\frac{1}{4}+\frac{1}{4(p-1)}}\log(n)\right)  \\
                                   \textcolor{white}{.}   \\ 
              \displaystyle   \bigO_{a.s.} \left(  \left(   d^{25/4+\theta_0}\log^*(d)C_{\xi}^{2/p} \right) \left(\frac{\sigma_d}{\sigma_0}\right)^{1/2}n^{\frac{1}{4}+\frac{1}{4(p-1)}}\log(n)\right)
                 
                \end{array}.
              \right.   $$

\end{theorem}
\begin{proof}

Let $m$ be an integer such that $m>1$ and $m<n$. We define
\[
x_j:=\sum_{i=1+(j-1)m}^{jm-1}\xi_i,\quad j=1,\dots, \lfloor n/m\rfloor.
\]
We further define
\[
 x_{\lfloor n/m\rfloor+1}:=\mathbbm{1}_{\{m\lfloor n/m\rfloor <n\}}\left(\sum_{i=1+m\lfloor n/m\rfloor}^n\xi_i\right).
 \]
 Clearly, if $m\lfloor n/m\rfloor =n$ we have $x_{\lfloor n/m\rfloor+1}=0$. Note that $\{x_j\}$ is an independent sequence.

Then it is easy to verify that
\[
\sum_{k=1}^n \xi_k = \sum_{j=1}^{\lfloor n/m\rfloor} x_j + x_{\lfloor n/m\rfloor +1} + \sum_{j=1}^{\lfloor n/m \rfloor}\xi_{jm}.
\]
By triangle inequality
\[
\left|\sum_{k=1}^n \xi_k - \sum_{j=1}^{\lfloor n/m\rfloor} x_j \right| \le \mathbbm{1}_{\{m\lfloor n/m\rfloor <n\}}\left|\sum_{i=1+m\lfloor n/m\rfloor}^n\xi_i\right| + \left|\sum_{j=1}^{\lfloor n/m \rfloor}\xi_{jm}\right|
\]

Note that $(\xi_{jm})_{j \in \mathbb{N}}$ are a zero mean independent sequence with  $\sup_k\sup_i\E[|\xi_{k,i}|^p]< C_{\xi}$, hence by Proposition \ref{GA_iid} there exists a Brownian motion $B$  such that
$$\left|\sum_{j=1}^{\floor{n/m}} \xi_{jm}- \Sigma_1^{1/2} B(\floor{n/m})\right|= \bigO_{a.s.}\left(d^{\theta} \log^*(d) C_{\xi}^{2/p}\left(\frac{\sigma_d}{\sigma_0}\right)^{1/2}\floor{n/m})^{1/p}\log ( \floor{n/m})\right),$$
where $\theta=25/4+\theta_0$ for some $\theta_0>0$ in the strong approximation case, $\Sigma_1$ is the covariance matrix of $\xi_{m}$. By applying the law of iterated logarithm to every coordinate of $B$, it follows that 
 \begin{align*}
\abs{\sum_{j=1}^{\floor{n/m}} \xi_{jm}} &\le \abs{\Sigma_1^{1/2} B(\floor{n/m}))} +\bigO_{a.s.}(d^{\theta} C_{\xi}^{2/p}(\log d) (\floor{n/m})^{1/p}(\log \floor{n/m}))\\
&\le \abs{\Sigma_1^{1/2}} \left(\sum_{i=1}^d B^2_{i}(\floor{n/m})) \right)^{1/2}+\bigO_{a.s.}(d^{\theta} C_{\xi}^{2/p}(\log d) (\floor{n/m})^{1/p}(\log \floor{n/m}))    \\
&= \bigO_{a.s.}(d^{1/2} \sqrt{\tr(\Sigma_1)\floor{n/m}} \log \log \floor{n/m}) +\bigO_{a.s.}(d^{\theta} C_{\xi}^{2/p}(\log d) (\floor{n/m})^{1/p}(\log \floor{n/m})).
 \end{align*} 
Note that since $\sqrt{\tr(\Var(\xi))}=\sqrt{\E\abs{\xi}^2} \leq \left(\E\abs{\xi}^p \right)^{1/p}\lsim d^{1/2}C^{1/p}_\xi$, we have that
\begin{align}
\label{rem_asymp}
    \left|\sum_{j=1}^{\lfloor n/m \rfloor}\xi_{jm}\right| = \bigO_{a.s.} (d^{\theta}\log^* d C_{\xi}^{2/p} (n/m)^{1/2} \log(n/m)).
\end{align} 

Similarly, we can use the same argument for the term $\mathbbm{1}_{\{m\lfloor n/m\rfloor <n\}}\left|\sum_{i=1+m\lfloor n/m\rfloor}^n\xi_i\right|$. The only additional step is to split the sequence of $\xi_i$ by its odd and even indices and use triangle inequality to divide the sum into two sums of independent sequences, each having $\bigO(n/m)$ terms. Then by applying the same argument, we obtain
\[
\left|\mathbbm{1}_{\{m\lfloor n/m\rfloor <n\}}\left(\sum_{i=1+m\lfloor n/m\rfloor}^n\xi_i\right)\right|= \bigO_{a.s.} (d^{\theta} C_{\xi}^{2/p}(\log d)(n/m)^{1/2} \log(n/m)).
\]
Therefore, we have shown
\[
\left|\sum_{k=1}^n \xi_k - \sum_{j=1}^{\lfloor n/m\rfloor} x_j \right| =\bigO_{a.s.} (d^{\theta} C_{\xi}^{2/p}(\log d)(n/m)^{1/2} \log(n/m)).
\]
Similarly, by an application of the weak Gaussian approximation, we can show 

\begin{align}
\label{WE1}
\left|\sum_{k=1}^n \xi_k - \sum_{j=1}^{\lfloor n/m\rfloor} x_j \right| =\bigO_{P} (d^{1/2} C_{\xi}^{1/p}(n/m)^{1/2} \log(n/m)).
\end{align}

Let $\tau_n=n^\gamma \log(n)$ and $m:=m_n=\floor{n^\alpha}$, for some $\alpha,\gamma>0$, that will be determined later in the proof. Define
        \[
        \tilde{x}_j:= x_j \mathbbm{1}_{\{|x_j|\le d^{1/2}C_{\xi}^{1/p}\tau_n \}}-\E[x_j \mathbbm{1}_{\{|x_j|\le d^{1/2}C_{\xi}^{1/p}\tau_n\}}].
        \]
        Clearly, we have $|\tilde{x}_j|\le 2d^{1/2}C_{\xi}^{2/p}\tau_n$ and $\E[\tilde{x}_j]=0$. 
 Then by triangle inequality, we have
        \begin{align*}
            \left|\sum_{j=1}^{\lfloor n/m\rfloor} x_j-\sum_{j=1}^{\lfloor n/m\rfloor} \tilde{x}_j\right|
        &\le  \left|\sum_{j=1}^{\lfloor n/m\rfloor} x_j\mathbbm{1}_{\{|x_j|>d^{1/2}C_{\xi}^{1/p}\tau_n)\}}\right| +\left|\sum_{j=1}^{\lfloor n/m\rfloor}\E[x_j\mathbbm{1}_{\{|x_j|>d^{1/2}C_{\xi}^{1/p}\tau_n\}}]\right|.
        \end{align*}

       Since we can decompose $x_j$ into sums of sub-sequences of even and odd indices, this implies that by Lemma \ref{order_norm_power} we have for all $j$ that $\E[|x_j|^p] \lsim d^{p/2}C_{\xi}m^{p/2}$.
       Then we have
        
         \begin{align*}
        \left|\sum_{j=1}^{\lfloor n/m\rfloor}\E[x_j\mathbbm{1}_{\{|x_j|>d^{1/2}C_{\xi}^{1/p}\tau_n\}}]\right|&\le d^{1/2}C_{\xi}^{1/p}\tau_n  \sum_{j=1}^{\lfloor n/m\rfloor}\E\left[\frac{|x_j|}{d^{1/2}C_{\xi}^{1/p}\tau_n}\mathbbm{1}_{\{|x_j|>d^{1/2}C_{\xi}^{1/p}\tau_n\}}\right]\\
        &\le d^{1/2}C_{\xi}^{1/p}\tau_n  \sum_{j=1}^{\lfloor n/m\rfloor}\E\left[\frac{|x_j|^p}{d^{p/2}C_{\xi}\tau_n^p}\mathbbm{1}_{\{|x_j|>d^{1/2}C_{\xi}^{1/p}\tau\}}\right]\\
        &\lsim d^{1/2}C_{\xi}^{1/p}nm^{p/2-1}\tau_n^{1-p}.
        \end{align*}

Note that $(\abs{x_j}\mathbbm{1}_{\abs{x_j}>d^{1/2}C_{\xi}^{1/p}\tau_n})_{j \in \mathbbm{N}}$ is an i.i.d sequence with mean $\mu_{n,d}$ with $0<\mu_{n,d}\lsim d^{1/2}C_{\xi}^{1/p}m^{p/2}\tau_n^{1-p}$. Let $s_n=\floor{n / m}$, we claim that 
\begin{equation}
\label{block_slln_claim1}
\sum_{j=1}^{s_n} \abs{x_j}\mathbbm{1}_{\{|x_j|>d^{1/2}C_{\xi}^{1/p}\tau_n)\}}=\bigO_{a.s.}(d^{1/2}C_{\xi}^{1/p} n m^{p/2-1}n^{\gamma(1-p)}).
\end{equation}

This is equivalent to showing that 

\begin{equation}
\label{blocked_truncated_slln2}
    \abs{ \frac{n^{\gamma(p-1)} m^{1-p/2}}{n  d^{1/2}C_{\xi}^{1/p}}\sum_{j=1}^{s_n}  \abs{x_j}\mathbbm{1}_{\{  |x_j|\geq d^{1/2}C_{\xi}^{1/p}\tau_n)\}}}\xrightarrow[]{a.s.}0 \ as \  n\rightarrow \infty.
\end{equation}

We will first show that \eqref{blocked_truncated_slln2} holds for a subsequence. Fix any $\lambda>1$ and define $n(\ell)=\floor{\lambda^\ell}$. Then, we claim that

\begin{equation}
\label{blocked_truncated_slln_sub}
     \abs{ \frac{n(\ell)^{\gamma(p-1)} m_{n(\ell)}^{1-p/2}}{n(\ell) C_{\xi}^{1/p} d^{1/2}}\sum_{j=1}^{\lfloor s_{n(\ell)}\rfloor}  \abs{x_j}\mathbbm{1}_{\{  |x_j|\geq d^{1/2}C_{\xi}^{1/p}\tau_{n(\ell)})\}}}\xrightarrow[]{a.s.}0 \ as \  \ell \rightarrow \infty.
\end{equation}
For every $\varepsilon>0$ we have that
\begin{align*}
\sum_{\ell=1}^\infty \P \left( \abs{\sum_{j=1}^{\lfloor s_{n(\ell)}\rfloor} \abs{x_j}\mathbbm{1}_{\{|x_j|>d^{1/2}C_{\xi}^{1/p}\tau_{n(\ell)}\} } } > \varepsilon d^{1/2} C_{\xi}^{1/p} {n(\ell)} m_{n(\ell)}^{p/2-1}{n(\ell)}^{\gamma(1-p)} \right) &\leq \sum_{\ell=1}^\infty  \frac{ \E [\abs{x_1}\mathbbm{1}_{\{|x_1|>d^{1/2}C_{\xi}^{1/p}\tau_{n(\ell)} \}}]} {\varepsilon d^{1/2}C_{\xi}^{1/p}  m_{n(\ell)}^{p/2}n(\ell)^{\gamma(1-p)}}\\
 &\lsim \frac{1}{\varepsilon} \sum_{\ell=1}^\infty \frac{\tau_{n(\ell)}^{(1-p)}}{{n(\ell)}^{\gamma(1-p)}} \\
  &\lsim \frac{1}{\varepsilon} \sum_{\ell=1}^\infty \frac{1}{\ell^{(p-1)}}<\infty,
\end{align*}

since $ \E [\abs{x_1}\mathbbm{1}_{\{|x_1|>d^{1/2}\tau_n\}}]\lsim d^{1/2}m^{p/2}\tau_n^{(1-p)}$ and $p>2$. 
Consequently, the claim \eqref{blocked_truncated_slln_sub} follows from the Borel--Cantelli lemma. We will choose $\alpha$ and $\gamma$ such that 
\begin{equation}
\label{positive_condition}
1+\alpha(p/2-1)+\gamma(1-p)>0.
\end{equation}Since for any $n$ there exists a $\ell$ such that $\lambda^\ell< n \leq \lambda^{\ell+1},$ we have that 
\begin{align*}
&n^{-1+\alpha(1-p/2)+\gamma(p-1)} \sum_{j=1}^{\floor{n^{1-\alpha}}}  \abs{x_j}\mathbbm{1}_{\{ |x_j|\geq d^{1/2}C_{\xi}^{1/p}\tau_n)\}}\\
&\leq \lambda^{\ell(-1+\alpha(1-p/2)+\gamma(p-1)}\hspace{-0.2cm}\sum_{j=1}^{\floor{\lambda^{(\ell+1)(1-\alpha)}}} \hspace{-0.3cm}\abs{x_j}\mathbbm{1}_{\{ |x_j|\geq d^{1/2}C_{\xi}^{1/p}\lambda^{\gamma(\ell+1)}\ell\ln(\lambda)\}}    
\end{align*}

Therefore it immediately follows that
\begin{align*}
&\limsup_{n \rightarrow \infty} n^{-1+\alpha(1-p/2)+\gamma(p-1)} \sum_{j=1}^{n}  \abs{x_j}\mathbbm{1}_{\{ |x_j|\geq d^{1/2}C_{\xi}^{1/p}\tau_n)\}}\\ &\leq \frac{\lambda^{-1+\alpha(1-p/2)+\gamma(p-1)}}{\lambda^{(\ell+1)(-1+\alpha(1-p/2)+\gamma(p-1))}}\sum_{j=1}^{\lambda^{\ell+1}} \abs{x_j}\mathbbm{1}_{\{  |x_j|\geq d^{1/2}C_{\xi}^{1/p}\lambda^{\ell\gamma}\}}\lsim_{a.s.}d^{1/2}C_{\xi}^{1/p}.
\end{align*}

Thus we have proven the claim  \eqref{block_slln_claim1}.        Consequently, we have
        \[
         \left|\sum_{j=1}^{s_n} x_j-\sum_{j=1}^{s_n} \tilde{x}_j\right|=\bigO_{a.s.}(d^{1/2} C_{\xi}^{1/p}\lambda^{\ell\gamma}nm^{p/2-1}\tau_n^{1-p})
        \]
        Overall, we have shown that if $\tau_n$ and $m$ are chosen such that $m^{p/2}\tau_n^{1-p}\to 0$ then
        \[
\left|\sum_{j=1}^n \xi_j - \sum_{j=1}^{s_n} \tilde{x}_j \right| =\bigO_{a.s.} (d^{\theta} C_{\xi}^{2/p}(\log d)(n/m)^{1/2} \log(n/m))+\bigO_{a.s.}(d^{1/2}C_{\xi}^{1/p}nm^{p/2-1}\tau_n^{1-p}).
\]
Similarly, for the weak approximation we have that 
\begin{align}
\label{WE2}
\left|\sum_{j=1}^n \xi_j - \sum_{j=1}^{s_n} \tilde{x}_j \right| =\bigO_{P} (d^{1/2} C_{\xi}^{1/p}(n/m)^{1/2} \log(n/m))+\bigO_{P}(d^{1/2}C_{\xi}^{1/p}nm^{p/2-1}\tau_n^{1-p}).
\end{align}
Next, since $\E[\tilde{x}_j]=0$ and $\abs{\tilde{x}_j}\le  2d^{1/2}C_{\xi}^{1/p}\tau_n$, by Lemma \ref{lemma_ind_GA}  there exists a sequence of independent Gaussians $(\tilde{y}_j)_j$ such that $\tilde{y}_j\sim \mathcal{N}_d(0,\Cov(\tilde{x}_j)) $ for $j=1,\dots,{s_n}$ and
\begin{align*}
    \left| \sum_{j=1}^{s_n} \tilde{x}_j - \sum_{j=1}^{s_n} \tilde{y}_j \right| &=\bigO_{a.s.}(d^{\theta}\log^*(d) C_{\xi}^{1/p}\tau_n {\log n}).
\end{align*}
Furthermore, since a Brownian motion at integer times coincides with a sum of i.i.d. Gaussian
random variables,  there exists a standard Brownian motion $B$ such that
\[ \left| \sum_{j=1}^{s_n} \tilde{x}_j - s_n^{-1/2}\Cov \left(\sum_{j=1}^{s_n} \tilde{x}_j\right)^{1/2} B({s_n}) \right|=\bigO_{a.s.}(d^{\theta}\log^*(d) C_{\xi}^{1/p}\tau_n {\log n}),
\]
where we also used the fact that $(\tilde{x}_j)$ is an i.i.d. sequence. Note that by scale invariance of Brownian motion, $W:=(\sqrt{m} \ B(n/m))$ is also a Brownian motion. Now let $\tilde{Y}_n$ and $Y_n$ be defined as follows:
 \[
\tilde{Y}_n:=\frac{1}{\sqrt{n}}\Cov \left(\sum_{j=1}^{s_n} \tilde{x}_j\right)^{1/2} W(n) .
\]
Using the same Brownian motion we also define
\[
Y_n:=\frac{1}{\sqrt{n}}\Cov\left(\sum_{k=1}^n \xi_k\right)^{1/2}W(n),
\]

 Thus we have that 
\begin{align}
\label{one_dep_block_approx1}
\left|\sum_{j=1}^n \xi_j  - \tilde{Y}_n \right| 
 &\le \left|\sum_{j=1}^n \xi_j - \sum_{j=1}^{s_n} \tilde{x}_j \right| + \left| \sum_{j=1}^{s_n} \tilde{x}_j - \tilde{Y}_n \right|\\ &= \bigO_{a.s.} (d^{\theta} C_{\xi}^{2/p}(\log d)\left(\frac{n}{m}\right)^{1/2} \log(n/m))\label{one_dep_block_approx11}\\
&+\bigO_{a.s.}(d^{1/2}C_{\xi}^{1/p}nm^{p/2-1}\tau_n^{1-p})\\ &+\bigO_{a.s.}(d^{\theta}\log d   C_{\xi}^{1/p}\tau_n {\log n}) \label{one_dep_block_approx2}.
\end{align}


Therefore we have that
\begin{align*}
|Y_n-\tilde{Y}_n|&=  \abs{ \left(\Cov\left(\sum_{k=1}^n \xi_k\right)^{1/2}-\Cov\left(\sum_{j=1}^{s_n} \tilde{x}_j\right)^{1/2}\right) \frac{1}{\sqrt{n}}W(n)} \\\
\le & \left|\Cov\left(\sum_{k=1}^n \xi_k\right)^{1/2}-\Cov\left(\sum_{j=1}^{s_n} \tilde{x}_j\right)^{1/2}\right| \frac{1}{\sqrt{n}}\abs{W(n)}\\
&\le \frac{1}{(\sigma_0 n)^{1/2}}\left|\Cov\left(\sum_{k=1}^n \xi_k\right)-\Cov\left(\sum_{j=1}^{s_n} \tilde{x}_j\right)\right|\frac{1}{\sqrt{n}}\abs{W(n)},
\end{align*}
where the second inequality follows  by Ando--{van Hemmen}'s inequality, Lemma \ref{Ando_inequality}, using the assumption that the smallest eigenvalue of $\frac{1}{\sqrt{n}}\sum_{k=1}^n \xi_k$ is bounded below by $\sigma_0$. Following the proof of Proposition \ref{GA_iid}, we see that
\begin{flalign*}
\left|\Cov\left(\sum_{k=1}^n \xi_k\right)-\Cov\left(\sum_{j=1}^{s_n} \tilde{x}_j\right)\right| \le & \frac{1}{(\sigma_0 n)^{1/2}} 2\sqrt{\E\left[\left|\sum_{k=1}^n \xi_k\right|^2\right]}\sqrt{\E\left[\left|\sum_{k=1}^n \xi_k-\sum_{j=1}^{\lfloor n/m\rfloor} \tilde{x}_j\right|^2\right]}\\
& +\frac{1}{(\sigma_0 n)^{1/2}} \E\left[\left|\sum_{k=1}^n \xi_k-\sum_{j=1}^{\lfloor n/m\rfloor} \tilde{x}_j\right|^2\right].
\end{flalign*}

Note that by a coordinate-wise application of the law of iterated logarithm we obtain
\begin{align}
\label{lil_one}
\abs{W(n)}=\left(\sum_{i=1}^dW^2_{i}(n)\right)^{1/2}=\bigO_{a.s.}(\sqrt{dn}\log \log(n)).
\end{align}

By Lemma \ref{order_norm_power} and Jensen's inequality we obtain
\[
\sqrt{\E\left[\left|\sum_{k=1}^n \xi_k\right|^2\right]} \lsim {d^{1/2}C_{\xi}^{1/p}n^{1/2}}.
\]

Furthermore, we have
\begin{align*}
\E\left[\left|\sum_{j=1}^{\lfloor n/m\rfloor} x_j-\sum_{j=1}^{\lfloor n/m\rfloor} \tilde{x}_j\right|^2\right]
        &=\E\left[\left|\sum_{j=1}^{\lfloor n/m\rfloor} \left(x_j\mathbbm{1}_{\{|x_j|>d^{1/2}C_{\xi}^{1/p}\tau_n)\}}- \E[x_j\mathbbm{1}_{\{|x_j|>d^{1/2}C_{\xi}^{1/p}\tau_n\}}]\right)\right|^2\right]\\
        &=\sum_{j=1}^{\lfloor n/m\rfloor}\E\left[ \left|x_j\mathbbm{1}_{\{|x_j|>d^{1/2}C_{\xi}^{1/p}\tau_n)\}}- \E[x_j\mathbbm{1}_{\{|x_j|>d^{1/2}C_{\xi}^{1/p}\tau_n\}}]\right|^2\right]\\
        &\le \sum_{j=1}^{\lfloor n/m\rfloor}\E\left[ \left|x_j\right|^2\mathbbm{1}_{\{|x_j|>d^{1/2}C_{\xi}^{1/p}\tau_n)\}}\right].
\end{align*}
        
Hence it follows that 

\begin{align*}
\E\left[\left|\sum_{k=1}^n \xi_k-\sum_{j=1}^{\lfloor n/m\rfloor} \tilde{x}_j\right|^2\right]&\le 2\sum_{j=1}^{\lfloor n/m\rfloor} \E[|\xi_{jm}|^2] + 2\sum_{j=1}^{\lfloor n/m\rfloor}\E[\abs{x_j}^2\mathbbm{1}_{\{|x_j|>d^{1/2}C_{\xi}^{1/p}\tau_n\}}]\\
&\le 2\sum_{j=1}^{\lfloor n/m\rfloor} \E[|\xi_{jm}|^2] + 2C_{\xi}^{2/p}d\tau_n^2\sum_{j=1}^{\lfloor n/m\rfloor}\E\left[\frac{\abs{x_j}^p}{d^{p/2}C_{\xi} \tau_n^p}\mathbbm{1}_{\{|x_j|>d^{1/2}C_{\xi}^{1/p}\tau_n\}}\right]\\
&=\bigO(d C_{\xi}^{2/p}n m^{-1})+\bigO(dC_{\xi}^{2/p}nm^{p/2-1}\tau_n^{2-p}).
\end{align*}

This gives 
\begin{align}
&|Y_n-\tilde{Y}_n| \le \frac{\abs{W(n)}}{n(\sigma_0)^{1/2}} \left[2\sqrt{\E\left[\left|\sum_{k=1}^n \xi_k\right|^2\right]}\sqrt{\E\left[\left|\sum_{k=1}^n \xi_k-\sum_{j=1}^{\lfloor n/m\rfloor} \tilde{x}_j\right|^2\right]}+\E\left[\left|\sum_{k=1}^n \xi_k-\sum_{j=1}^{\lfloor n/m\rfloor} \tilde{x}_j\right|^2\right]\right] \nonumber \\
& = \frac{\bigO_{a.s.}(\sqrt{d} \log(n))}{(\sigma_0 n)^{1/2}}\Big[\bigO(d^{1/2}C_{\xi}^{1/p}n^{1/2})\sqrt{\bigO(d C_{\xi}^{2/p}n m^{-1})+\bigO(dC_{\xi}^{2/p}nm^{p/2-1}\tau_n^{2-p})}+\bigO(d C_{\xi}^{2/p}n m^{-1}) \Big] \nonumber\\
&+\frac{\bigO_{a.s.}(\sqrt{d} \log(n))}{(\sigma_0 n)^{1/2}}\Big[\bigO(dC_{\xi}^{2/p}nm^{p/2-1}\tau_n^{2-p})\Big]\nonumber \\
& =\bigO_{a.s.}(d^{3/2}C_{\xi}^{2/p}n^{1/2}m^{-1/2}\log(n))+\bigO_{a.s.}(d^{3/2}C_{\xi}^{2/p}n^{1/2}m^{p/4-1/2}\tau_n^{\frac{2-p}{2}}\log(n))\\
&+ \bigO_{a.s.}(d^{3/2}C_{\xi}^{2/p}n^{1/2}m^{p/2-1}\tau_n^{2-p}\log(n)) \nonumber\\& =\bigO_{a.s.}(d^{3/2}C_{\xi}^{2/p}n^{1/2}m^{-1/2}\log(n))+\bigO_{a.s.}(d^{3/2}C_{\xi}^{2/p}n^{1/2}m^{p/4-1/2}\tau_n^{\frac{2-p}{2}}\log(n)). \nonumber
\end{align}
Here the first inequality follows from \eqref{lil_one} and the last equality follows since $\tau_n$ and $m$ should be chosen such that $m^{p/2-1}\tau_n^{2-p}$ tends to zero and hence $m^{p/2-1}\tau_n^{2-p} \ll m^{p/4-1/2}\tau_n^{1-p/2}$. \\

Now choosing the block size $m$ and the truncation level $\tau_n$ as follows
\[
m=\lfloor n^{\alpha}\rfloor, \quad \alpha:=\frac{p-2}{2(p-1)},\quad  \tau_n:=n^{\alpha/2-\alpha/p+1/p},
\]

we have
\[
\E\left[\left|\sum_{k=1}^n \xi_k-\sum_{j=1}^{\lfloor n/m\rfloor} \tilde{x}_j\right|^2\right]=\bigO(d C_{\xi}^{2/p}n^{\frac{p}{2(p-1)}})
\]
and consequently
\begin{align}
    \label{gaus_approx_one_dep}
    |Y_n-\tilde{Y}_n|&=\bigO_{a.s.}(d^{3/2}C_{\xi}^{2/p}n^{1/2}m^{-1/2}\log(n))+\bigO_{a.s.}(d^{3/2}C_{\xi}^{2/p}n^{1/2}m^{p/4-1/2}\tau_n^{\frac{2-p}{2}}\log(n)) \nonumber \\
    &= \bigO_{a.s.}\left(d^{3/2}C_{\xi}^{2/p}n^{\frac{1}{4}+\frac{1}{4(p-1)}}\log(n)\right)+\bigO_{a.s.}\left(d^{3/2}C_{\xi}^{2/p}n^{\frac{1}{2(p-1)}})\log(n)\right) \nonumber\\
     &= \bigO_{a.s.}\left(d^{3/2}C_{\xi}^{2/p}\ n^{\frac{1}{4}+\frac{1}{4(p-1)}}\log(n)\right).
\end{align}
Note that since $p>2$ also we have that for the specified choice of $\tau_n$ and $m$ that \eqref{positive_condition} is satisfied. Furthermore, from \eqref{one_dep_block_approx11}-\eqref{one_dep_block_approx2} we see that 
\begin{align}
\label{one_dep_block_approx_2}
\left|\sum_{k=1}^n \xi_k  - \tilde{Y}_n\right| = \bigO_{a.s.}( C_{\xi}^{1/p} d^{3/2}  n^{\frac{1}{4}+\frac{1}{4(p-1)}}\log n)
\end{align}

Note that since $\{\xi_k\}$
is a one dependent identically distributed sequence, we have that
$\Cov\left( \sum_{k=1}^n \xi_k\right)= n\Var(\xi_1)+ (n-1)\left(\Cov(\xi_1,\xi_2)+ \Cov(\xi_1,\xi_2)^T)\right)$.  
Consequently,

\[
Y_n= \left(\Var(\xi_1)+  \frac{n-1}{n}\left(\Cov(\xi_1,\xi_2)+ \Cov(\xi_1,\xi_2)^T\right)\right)^{1/2}  W(n).
\]

Furthermore, we have that the asymptotic covariance matrix is given by $\Sigma_\xi=\Var(\xi_1)+ \Cov(\xi_1,\xi_2)+ \Cov(\xi_1,\xi_2)^T.$ This gives us
\begin{align}
    \label{asymp_variance_diff}
    \abs{ Y_n -\Sigma^{1/2}_\xi W(n)} &\leq \frac{1}{n(\sigma_0 )^{1/2}} \abs{\Cov(\xi_1,\xi_2)+ \Cov(\xi_1,\xi_2)^T } \abs{W(n)} \nonumber\\
    & = \frac{\bigO_{a.s.}(\sqrt{d}\log(n))}{(\sigma_0 n)^{1/2}}  \abs{\Cov(\xi_1,\xi_2)}, \end{align}
where the first inequality follows by Ando--{van Hemmen}'s inequality, Lemma \ref{Ando_inequality}. We note that since \eqref{asymp_variance_diff} is of smaller asymptotic magnitude that the other approximation terms.

Overall, from (\ref{gaus_approx_one_dep}), (\ref{one_dep_block_approx_2}), \eqref{asymp_variance_diff}  we have
\begin{align*}
\left|\sum_{k=1}^n \xi_k  -\Sigma^{1/2}_\xi W(n)\right| 
&= \bigO_{a.s.} \left(  \left(   d^{25/4+\theta_0}C_{\xi}^{2/p} \right) \left(\frac{\sigma_d}{\sigma_0}\right)^{1/2}n^{\frac{1}{4}+\frac{1}{4(p-1)}}\log(n)\right).
\end{align*}

This finishes the proof for the strong Gaussian approximation case. The weak Gaussian approximation case follows from the same argument with minor differences. Following the proof of Proposition \ref{GA_iid}, we will use \cite[Lemma 1]{eldan2020clt} to control the difference between the Gaussian approximation of the truncated and true blocks. Now note that

 \begin{align*}
\tilde{Y}_n:&= \frac{1}{\sqrt{n}}\Cov \left(\sum_{j=1}^{s_n} \tilde{x}_j\right)^{1/2} W(n)  \overset{d}{=} \Cov \left(\sum_{j=1}^{s_n} \tilde{x}_j\right)^{1/2}Z =: \tilde{Y}_n'.
\end{align*}
with $Z$ denoting a standard $d$--dimensional Gaussian random variable.
Using the same Gaussian random variable $Z$ we  define
\[
{Y}_n:=\frac{1}{\sqrt{n}}\Cov\left(\sum_{k=1}^n \xi_k\right)^{1/2}W(n)\overset{d}{=}\Cov\left(\sum_{k=1}^n \xi_k\right)^{1/2}Z=:Y_n'.
\]

In \eqref{diff_cov_1}, we have shown that
\[
|Y'_n-\tilde{Y}'_n|= \bigO_P \left(\frac{1}{\sqrt{n}}\abs{ {\Sigma}_n-\tilde{\Sigma}_n} \right)
\]
with
\[ \Sigma_n=\Cov\left(\sum_{k=1}^n \xi_k\right) \ \textrm{and} \ \tilde{\Sigma}_n=\Cov\left(\sum_{j=1}^{\lfloor n/m\rfloor} \tilde{x}_j\right).
\]
With the same choice of block size and truncation level, this gives us

\begin{align*}
|Y_n-\tilde{Y}_n|
= \bigO_{P}\left(|Y'_n-\tilde{Y}'_n|\right)=\bigO_{P}\left(dC_{\xi}^{2/p}\ n^{\frac{1}{4}+\frac{1}{4(p-1)}}\right).
\end{align*}
Note that the calculations are exactly the same as in the strong approximation case; \eqref{gaus_approx_one_dep}, without the $\sqrt{d}\log\log n$ factor from the Brownian motion. Similarly, we obtain
\begin{align*}
    \abs{Y_n -\Sigma^{1/2}_\xi W(n)}= \bigO_P \left(\frac{1}{\sqrt{n}}\abs{ {\Sigma}_n-n \Sigma_\xi} \right)= \bigO_P \left(\frac{1}{\sqrt{n}}\abs{  \Sigma_\xi} \right)
\end{align*}
Combining these results with \eqref{WE1} and \eqref{WE2} we see that
\begin{align}
\label{weak_one_dep_block_approx_2}
\left|\sum_{k=1}^n \xi_k  - \Sigma^{1/2}_\xi W(n)\right| = \bigO_{P}( dC_{\xi}^{2/p}   n^{\frac{1}{4}+\frac{1}{4(p-1)}}\log n)
\end{align}
\end{proof}

\begin{remark}
Note that from the proof of Theorem \ref{GA_one_dep} we can also obtain
$$\left|\sum_{k=1}^n \xi_k- \Sigma_\xi^{1/2}W(n)\right|=    \left\{
            \begin{array}{lll}                                          
                 \displaystyle \bigO_{P} \left(  \left(\sqrt{d}\ (\E \abs{\xi_1}^p)^{1/p} \vee (\E \abs{\xi_1}^p)^{2/p} \right)  n^{\frac{1}{4}+\frac{1}{4(p-1)}}\log(n)\right)  \\
                                   \textcolor{white}{.}   \\ 
              \displaystyle   \bigO_{a.s.} \left(  \left(   d^{25/4+\theta_0}\log^*(d)C_{\xi}^{1/p} \vee   d^{3/2} C_{\xi}^{2/p} \right) \left(\frac{\sigma_d}{\sigma_0}\right)^{1/2}n^{\frac{1}{4}+\frac{1}{4(p-1)}}\log(n)\right)
                 
                \end{array}.
              \right.   $$

Note that since $\E \abs{\xi_1^p} \leq d^{p/2} \sup_{i \in \{1,\cdots,d\}} \E\abs{\xi_{1i}}^p$, we obtain the same weak approximation error.
\end{remark}

\subsection{Proofs of Section \ref{section:3}}
\subsubsection{Proof of Theorem \ref{theorem_approx_discrete_multi}}
\begin{proof}
Since the minorisation condition holds for some $m_0>1$, we have by Proposition \ref{semi_reg_discrete} there exists a sequence of randomized stopping times $\{R_k\}$ such that we can define 
 $$ \xi_k := \sum_{t=R_{k-1}}^{R_k-1}(f(X_t)-\pi(f)),\ \ k \geq 1,$$ 
\noindent
such that $(\xi_k)_{k \in \mathbbm{N}}$ 
is a stationary one-dependent sequence under $\mathbb{P}_\nu.$  Furthermore, by \cite[Theorem 3.2]{asmussen},  we have   that 
  %
$$\EE_\nu \xi_1=\EE_\nu \sum_{t=0}^{R_1} \{f(X_t)-\pi(f) \}\ =\mu_\varrho \cdot \pi(f-\pi(f))=0. $$
 Moreover, under Assumption A\ref{assumption_moment_condition}.1,  we have that 
    \begin{equation}
    \label{exp_block_moment_proof}
\sup_{i\in \{1,\dots,d\}}\EE_\nu \left[\left(\sum_{t=0}^{R_1} \abs{f_i(X_t)}\right)^{p_0} \right]  \   \lsim
 \alpha^{-1} \left( \mathbb{E}_\nu [h(R_1]\right)^{\varepsilon/p_0}\sup_{i\in \{1,\dots,d\}} \pi(\abs{f_i}^{p_0+\varepsilon}),         \end{equation}
where $h(x)=x^{q}$ with $q=\eta/(1-\eta)$ and $p_0$ is defined in \eqref{p_0definition} in the case of a polynomial drift condition and $h(x)=e^{tx}$ with $\abs{t}< -\ln(\lambda)m_0$ and $p_0=p$ in the case of an exponential drift condition. By Theorem \ref{GA_one_dep}, we can redefine $(\xi_k)_k$ on a new probability space on which we can also construct a standard $d$-dimensional Brownian motion $W$ such that
\begin{equation}
\label{S1}
\abs{\sum_{k=1}^n \xi_k-n\mathbb{E}_\nu \xi_1-W(n)}=    \left\{
            \begin{array}{lll}                              %
                 \displaystyle         \bigO_{P}\left( d \tilde{C}_{N}^{2/p} \left(\frac{\sigma_d}{\sigma_0}\right) T^{\frac{1}{4}+\frac{1}{4(p_0-1)}} \right),   \\
                                   \textcolor{white}{.}   \\ 
              \displaystyle         \bigO_{a.s.}\left( d^{25/4+\theta_0}\log^* (d)\tilde{C}_{N}^{2/p} \left(\frac{\sigma_d}{\sigma_0}\right)^{1/2} T^{\frac{1}{4}+\frac{1}{4(p_0-1)}} \right),  
                 
                \end{array},
              \right.   
         \end{equation}

where $$\tilde{C}_{N}=\alpha^{-1} \left( \mathbb{E}_\nu [h(R_1]\right)^{\varepsilon/p_0} \sup_{i\in \{1,\dots,d\}} \pi(\abs{f_i}^{p_0+\varepsilon}).$$

In the remainder of the proof, we will consider the strong Gaussian approximation case, since the weak approximation case follows from similar and slightly easier arguments. Note, that since we assume that $X_0\sim \pi$,  we need to discard the first cycle. Similarly, in order to apply \eqref{S1}, we will need to discard the last cycle. Let $\eta(T)$  denote the number of regenerations of the $m_0$--skeleton chain up to time $T$, namely 
 $$\eta(T)=\max\{k: R_k \leq T \}.$$

It immediately follows that 
\begin{equation}
\label{remainder2}
\abs{\sum_{t=0}^{T}(f(X_t)-\pi(f)) \ -\xi_0 -\sum_{k=1}^{\eta(T)}\xi_k}=  \abs{\xi_0}+ \abs{\sum_{t=R_{\eta(T)}}^T(f(X_t)-\pi(f) ) },
\end{equation}
where \[
\xi_0:= \sum_{t=0}^{R_0}(f(X_t)-\pi(f)).
\]
By a Borel--Cantelli argument, we will show that
\begin{equation}
\label{BorelCantONE}
 \abs{\sum_{t=R_{\eta(T)}}^T(f(X_t)-\pi(f) ) }=\bigO_{a.s.} \left(\alpha^{-1/p}d^{1/2}n^{1/p_0}h(R_1)^{\varepsilon/p_0^2}\right)
\end{equation}

and  
\begin{equation}
\label{BorelCantTWO}
\abs{\xi_0}=\bigO_{a.s.}\left( d^{1/2}n^{1/p_0}h(R_1)^{\varepsilon/p_0^2}\right).
\end{equation}
In order to show \eqref{BorelCantONE}, note that
\begin{equation}
\label{triv_bound}
\abs{\sum_{t=R_{\eta(T)}}^T(f(X_t)-\pi(f) ) } \le \abs{\sum_{t=R_{\eta(T)}}^{R_{\eta(T)+1}}g(X_t)},  \end{equation}

where $g(x)=(\abs{f_1(x)-\pi(f_1)}, \dots, \abs{f_d(x)-\pi(f_d)})^T.$ Now let $\varepsilon>0$ be given and introduce the event $$A_{n}=\left\{\abs{\sum_{t=R_{n}}^{R_{n+1}}g(X_t)}> \alpha^{-1/p}d^{1/2}n^{1/p_0}h(R_1)^{\varepsilon/p_0^2}\right\}.$$
By Markov's inequality, it follows that the introduced sequence of events satisfies 
\begin{align*}
\sum_{n=1}^\infty \mathbb{P}_\nu\left(A_{n}\right)&\leq \sum_{n=1}^\infty \mathbb{P}_\nu\left(\abs{\sum_{t=R_{n}}^{R_{n+1}}g(X_t)}> \alpha^{-1/p}d^{1/2}n^{1/p_0}h(R_1)^{\varepsilon/p_0^2}\right)\\
&=\sum_{n=1}^\infty \mathbb{P}_\nu\left(\abs{\sum_{t=0}^{R_{1}}g(X_t)}^{p_0}> \alpha^{-1}d^{p_0/2}n h(R_1)^{\varepsilon/p_0}\right)\\
&\leq \mathbb{E}_\nu \left[\abs{\sum_{t=0}^{R_{n1}}g(X_t)}^p\right] (\alpha^{-1}d^{p_0/2}n h(R_1)^{\varepsilon/p_0})^{-1} <C,
\end{align*}
where $C$ is some universal constant and the one for the last inequality follows by \eqref{moment_norm_exp_p} of Lemma \ref{moments_regen_exp} and \ref{moments_regen_poly} in the case of a geometric and polynomial drift condition respectively. By the Borel--Cantelli lemma it follows that $\mathbb{P}_{\nu}(\limsup A_{n})=0$.
Consequently, we have that $\mathbb{P}_{\nu}(\liminf A_{n}^c)=1$. 
Hence it follows that 
\begin{align}
\label{block_asymp}
    \abs{\sum_{t=R_{n}}^{R_{n+1}}g(X_t)}= \bigO_{a.s.} (\alpha^{-1/p}d^{1/2}n^{1/p_0}h(R_1)^{\varepsilon/p_0^2}).
\end{align} 
    Moreover, since $\eta(T)$ is almost surely increasing and $\eta(T)=\bigO_{a.s.} (T)$, it follows that
$$\abs{\sum_{t=R_{n}}^{R_{n+1}}g(X_t)} = \bigO_{a.s.} (\alpha^{-1/p}d^{1/2}{\eta(T)}^{1/p_0}h(R_1)^{\varepsilon/p_0^2})= \bigO_{a.s.} (\alpha^{-1/p}d^{1/2}n^{1/p_0}h(R_1)^{\varepsilon/p_0^2}).$$
Hence the claim formulated in (\ref{BorelCantONE}) directly follows. The claim \eqref{BorelCantTWO} follows completely analogously to \eqref{block_asymp}. 
Since we have that $\mathbb{E}_\nu R_1^{p_0}< \infty$, by \cite[Theorem 2.4]{weighted_approx} we can construct a Brownian motion $\tilde{W}$ such that
\begin{equation}
    \abs{\eta(T)-\frac{T}{\mu_{\varrho}}-\frac{\sigma_{\varrho}}{\mu_{\varrho}^{3/2}}\tilde{W}_T}=o_{a.s.}(T^{1/{p_0}}),
\end{equation}
\noindent
By the law of iterated logarithm for Brownian motion we obtain 
\begin{equation}
\label{standardarguments2}
\eta(T)=\frac{T}{\varrho}+\bigO_{a.s.} (\sqrt{T \log \log T})\quad \textrm{a.s.   } 
\end{equation}
\noindent
Furthermore, by (\ref{S1}), there exists an almost surely finite random variable $C$ such that for almost all sample paths $\omega$ we have that $\textrm{for all } n\geq N_0\equiv N_0(\omega)$ we have that 
\begin{equation}
\label{split1}
\left.\frac{1}{\left(\left(   d^{25/4+\theta_0}\log^* (d)\tilde{C}_{N}^{2/p} \left(\frac{\sigma_d}{\sigma_0}\right)^{1/2}\right)   n^{\frac{1}{4}+\frac{1}{4(p_0-1)}}\log n\right)}\abs{\sum_{k=1}^n \xi_k(\omega)-W(n)}\right. < C(\omega)
\end{equation}
 Since $\eta(T)$ is almost surely increasing and tends to infinity, we have that for almost every sample path $\omega$ there exists a  $T_0\equiv T_0(\omega)$ such that $\eta(T)(\omega)\geq N_0$ for all $T\geq T_0$. Hence we obtain from (\ref{split1}) that
 \begin{equation}
 \label{split2}
 \left. \limsup_{T\rightarrow \infty} \frac{\left\lvert \sum_{k=1}^{\eta(T)} \xi_k - \Sigma_\xi^{1/2} W(\eta(T))\right\lvert }{\left(  d^{25/4+\theta_0}\log^* (d)\tilde{C}_{N}^{2/p} \left(\frac{\sigma_d}{\sigma_0}\right)^{1/2} \right)   \eta(T)^{\frac{1}{4}+\frac{1}{4(p_0-1)}}\log \eta(T)}
\right. < C \quad \textrm{a.s. }.
\end{equation}
We see that (\ref{split2}) can be reformulated as
\begin{align}
\label{stoch_invariance}
\left\lvert \sum_{k=1}^{\eta(T)} \xi_k -(\Sigma_\xi)^{1/2} W(\eta(T))\right\lvert&=\bigO_{a.s.}\left(d^{25/4+\theta_0}\log^* (d)\tilde{C}_{N}^{2/p} \left(\frac{\sigma_d}{\sigma_0}\right)^{1/2}  \eta(T)^{\frac{1}{4}+\frac{1}{4(p_0-1)}}\log \eta(T)\right) \nonumber \\
&= \bigO_{a.s.}\left(d^{25/4+\theta_0}\log^* (d)\tilde{C}_{N}^{2/p} \left(\frac{\sigma_d}{\sigma_0}\right)^{1/2}   T^{\frac{1}{4}+\frac{1}{4(p_0-1)}}\log T\right)
\end{align}
Here the second equality follows by  (\ref{standardarguments2}). 
By a coordinate-wise application of  \cite[Theorem 1.2.1]{sip_boek}) it follows that 
\begin{align}
\label{browniandifference}
\big\lvert W({\eta(T)})-W(T/\mu_\varrho) \big\lvert &\leq  \left(\sum_{i=1}^d \lvert W_i({\eta(T)})-W_i(T/\mu_\varrho) \big\lvert^2 \right)^{1/2} \nonumber \\
&=\bigO_{a.s.}( \sqrt{d}  T^{1/4}\log T).
\end{align}

From \eqref{remainder2} and combining results (\ref{BorelCantONE}), \eqref{BorelCantTWO},  (\ref{stoch_invariance}), and (\ref{browniandifference}) the asserted theorem follows. The weak approximation case follows completely analogously.
\end{proof}

\subsubsection{Proof of Theorem \ref{theorem_approx_discrete_one}}
\begin{proof}
We will first assume that $X_0\sim \nu$ and consider the case where the geometric drift condition holds. Since the minorisation condition holds with $m_0=1$, we have by Proposition \ref{semi_reg_discrete} there exists a sequence of randomized stopping times $\{R_k\}$ such that we can define
 $$ \xi_k := \sum_{t=R_{k-1}}^{R_k-1}(f(X_t)-\pi(f)),\ \ k \geq 1,$$ 
\noindent
such that $(\xi_k)_{k \in \mathbbm{N}}$ 
is a mean-zero independent and identically distributed sequence under $\mathbb{P}_\nu.$ 

By Lemma \ref{moments_regen_exp} we have that for
any $t$ with $\abs{t}\leq \ln(1/\lambda) /m_0$ we have that 
\begin{equation}
\label{exp_regen_moment_proof}
\E_\nu [e^{t R_1}]  \lsim \frac{b}{\alpha \lambda (1-\lambda)}
\end{equation}
and 
    \begin{equation}
    \label{exp_block_moment_proof2}
\sup_{i\in \{1,\dots,d\}}\EE_\nu \left[\left(\sum_{t=0}^{R_1} \abs{f_i(X_t)}\right)^{p_0} \right]  \   \lsim
 \alpha^{-1} \left( \mathbb{E}_\nu [h(R_1]\right)^{\varepsilon/p_0} \sup_{i\in \{1,\dots,d\}} \pi(\abs{f_i}^{p+\varepsilon}),         \end{equation}

Note that from \eqref{exp_regen_moment_proof} we also have that
\begin{equation}
\label{p_moment_time}
\E_\nu [  R_1^{p}]  \leq \left(\frac{p}{\ln(1/\lambda)e}\right)^p\left(\frac{b}{\alpha \lambda (1-\lambda)}\right). 
\end{equation}

\noindent
Define $(\varrho_k)_{k \in \mathbb{N}}$ as $\varrho_k=R_k-R_{k-1}$. Note that this is an i.i.d sequence and let $\mu_\varrho$ and $\sigma_\varrho^2$ denote the respective mean and variance. The sequence of $(d+1)$-random vectors $(\xi_k,\varrho_k)$ are independent dependent and identically distributed.  Introduce the sequence $(\tilde{\xi}_k)_{k \in \mathbb{N}}$ as $\tilde{\xi}_k=\xi_k-\beta (\varrho_k- \mu_\varrho)$. Note that the sequence of random vectors $(\tilde{\xi}_k,\varrho_k)$ are also independent and identically distributed.  If we choose $\beta=\Covnu(\xi_1,\varrho_1)/\sigma^2_\varrho$, then it immediately follows that $\rho_k$ and every component of $\tilde{\xi}_k$ are uncorrelated. Let $\tilde{\Sigma}$ denote the limiting covariance of $\frac{1}{n}\sum_k (\tilde{\xi}_k,\varrho_k).$ We see that $\tilde{\Sigma}$ is given by 
\begin{equation}
\label{v1}\Varnu
\left(\begin{array}{ c} 
   \tilde{\xi}_1 \\ 
      \varrho_1  
    \end{array}  
\right)= \Varnu\left( \begin{array}{c} 
      \xi_1-\beta (\varrho_1- \mu_\varrho)  \\ 
      \varrho_1  
    \end{array}  
\right)= \left( 
    \begin{array}{c c} 
      {V}_\xi & \underline{0} \\ 
       
      \underline{0}^T & \sigma_\varrho^2 
    \end{array}\right), 
\end{equation}

\noindent
where by definition of $\beta$ we see that the off-diagonal entries of block matrix (\ref{v1}) are $\underline{0}$, which denotes a $d$-dimensional vector of zeros, and ${V}_\xi$ is given by
\begin{equation}
\Varnu(\xi_1-\beta(\varrho_1-\mu_\varrho))=\Varnu (\xi_1)+ \beta \beta^T \sigma^2_\varrho-2 \beta \beta^T \sigma^2_\varrho =\Varnu (\xi_1)- \beta \beta^T \sigma^2_\varrho.
\end{equation}
\noindent
From \eqref{exp_block_moment_proof2},
\eqref{p_moment_time}, and a coordinate-wise application of Cauchy-Schwarz, we see that the conditions of Proposition \ref{GA_iid} are satisfied with 
\begin{align*}
\sup_{i\in \{1,\dots,d\}}\E_\nu \tilde{\xi}^p_{1i}&\leq\left(1 +  \left(\frac{p}{\ln(1/\lambda)e}\right)^p\left(\frac{b}{\alpha \lambda (1-\lambda)}\right) \right) \alpha^{-1}\left(\frac{b}{\alpha(1-\lambda)} \right)^{\varepsilon/p} \sup_{i\in \{1,\dots,d\}} \pi(\abs{f_i}^{p+\varepsilon})\\ &\leq 2 \alpha^{-1}\left(\frac{b}{\alpha(1-\lambda)} \right)^{1+\varepsilon/p} \left(\frac{p}{\ln(1/\lambda)e}\right)^p\sup_{i\in \{1,\dots,d\}}\pi(\abs{f_i}^{p+\varepsilon}) =:C_{\xi}
\end{align*}

Applying the multivariate Gaussian  approximation given in Proposition \ref{GA_iid}, we have that

$$\abs{\sum_{k=1}^n\left( 
    \begin{array}{c } 
      \tilde{\xi}_k \\ 
      \varrho_k 
    \end{array}\right)-\sum_{k=1}^n \left( 
    \begin{array}{c } 
      Y^1_k \\ 
      Y^2_k 
    \end{array}\right)}=\bigO_{a.s.}\left(   \tilde{C}_{\xi,d} \Psi_n \right),$$
where $$\Psi_n= n^{1/p}\log n,$$
and
$$\tilde{C}_{\xi,d}:=\left(d^{25/4+\theta_0}\log^* (d)C_{\xi}^{1/p}\vee   d^{3/2} C_{\xi}^{2/p} \right) \left(\frac{\sigma_d}{\sigma_0}\right)^{1/2},$$
and $\sum_{k=1}^n\begin{psmallmatrix} Y^1_k\\ Y^2_k\end{psmallmatrix}$ has a Gaussian distribution with the same mean and covariance matrix as $\sum_{k=1}^n\begin{psmallmatrix} \tilde{\xi}_k\\\varrho_k\end{psmallmatrix}$. Given the block structure of  $\tilde{\Sigma}$ given in (\ref{v1}), we see that $Y_2^k$ is independent of all components of $Y^1_k.$ By the Skorohod embedding theorem, two independent Brownian motions $B_1$  and $B_2$ can be constructed, where $B_1$ is a $d$-dimensional Brownian motion and $B_2$ is one-dimensional, such that they coincide with the Gaussian sequences at all integer time points. Therefore we have that

\begin{equation}
\abs{\sum_{k=1}^n\xi_k-\beta(\sum_{k=1}^n \varrho_k- \mu_\varrho)- {V}_\xi^{\frac{1}{2}}B_1}=\bigO_{a.s.}\left(  \tilde{C}_{\xi,d}\Psi_n  \right)     \label{s1}
\end{equation}

and 
\begin{equation}
\abs{R_n-n\mu_\varrho-\sigma_\varrho B_2(n)}=\bigO_{a.s.}\left( \tilde{C}_{\xi,d} \Psi_n \right).
\label{s2}
\end{equation}

\noindent
 By \citet[Theorem 1(ii)]{kmt_1}, a Poisson process $L$ with intensity $\lambda=\mu_\varrho^2/\sigma^2_\varrho$ can be constructed from the one-dimensional Brownian motion $B_2$ such that
\begin{align}
    \abs{L(n)-\frac{\mu_\varrho}{\gamma}n-\frac{\sigma_\rho}{\gamma}B_2(n)}=\bigO_{a.s.}(\log n), \label{PP_sip}
\end{align}
\noindent
where $\gamma=\sigma^2_\varrho/\mu_\varrho$ and $L$ is constructed increment-wise from $B_2$ in a determinstic way and is therefore also independent of $B_1$. From (\ref{s2}) and (\ref{PP_sip}) it follows that \begin{align}
\abs{R_n-\gamma L(n)}=\bigO_{a.s.}(\tilde{C}_{\xi,d} \Psi_n).
\label{claim1}
\end{align} We claim that it therefore follows that
\begin{align}
    \abs{\sum_{k=1}^n \xi_k- \sum_{k=1}^{\gamma L(n)} \xi_n}=\abs{\sum_{k=0}^{R_n} (f(X_k)-\pi(f)) - \sum_{k=0}^{\gamma L(n)} (f(X_k)-\pi(f)) }=\bigO_{a.s.}( \tilde{C}_{\xi,d} \Psi_n) \label{claim2}
\end{align}
We see that
\begin{align}
    \abs{\sum_{k=0}^{R_n} (f(X_k)-\pi(f)) - \sum_{k=0}^{\gamma L(n)} (f(X_k)-\pi(f)) }&= \abs{\sum_{b_n}^{c_n} (f(X_k)-\pi(f))}
\end{align}
where $b_n:=\min\{R_n, \gamma L(n)\}$ and $c_n:=\max\{R_n, \gamma L(n)\}$. Therefore we can introduce the positive sequence $\kappa_n$ as follows  $$\kappa_n:=c_n-b_n=\abs{R_n-\gamma L(n)}.$$ From (\ref{claim1}) it follows that $\kappa_n=\bigO_{a.s.}(C_{\xi,N,d} \Psi_n),$ hence it follows that there exists some universal almost surely finite $C$ such that 
for almost every $\omega$ it holds that   there exists an $N_1:=N_1(\omega)$ such that for all $n\geq N_1$ we have that $\kappa_n< C(\omega) C_{\xi,N,d} \Psi_n$ and hence $c_n= b_n+\kappa_n \leq b_n+\bar{\kappa}_n$, with $\bar{\kappa}_n= \ceil{C(\omega) C_{\xi,N,d} \Psi_n}$.
Note that the stopping times $(R_k)_{k\geq 0}$ are regeneration epochs of the process, and hence the corresponding cycles $\mathcal{C}_k:=(X_s: R_k\leq s<R_{k+1})$ are independent and identically distributed.
Let $\eta(T):=\max\{k: R_k \leq T \}$ denote the amount of regenerative cycles up to time $T$ and let $(Y_k)_{k \in \mathbb{N}}$ be defined as
\begin{equation}
\label{triv_bound}
Y_k= \abs{\sum_{t=R_{k}}^{R_{k+1}-1}g(X_t)},  \end{equation}
where $g(x)=(\abs{f_1(x)-\pi(f_1)}, \dots, \abs{f_d(x)-\pi(f_d)})^T.$ 
Then we see that for $n>N_1(\omega)$  we have that
\begin{align}
\abs{\Psi_n^{-1}\sum_{b_n}^{c_n}    (f(X_k)-\pi(f))} &= \Psi_n^{-1} \abs{\sum_{k=0}^{c_n-b_n} ( f(X_{b_n+u})-\pi(f))} \nonumber\\
\hspace{-0.85cm}&\leq \Psi_n^{-1}\sum_{k=0}^{\kappa_n}  g(X_{b_n+u}) \nonumber\\
\hspace{-1.5cm}&\leq \Psi_n^{-1}\sum_{k=0}^{\bar{\kappa}_n} g(X_{b_n+k}) \nonumber\\
\hspace{-1.5cm}&= \Psi_n^{-1} \sum_{j=\eta(b_n)}^{\eta( b_n+\bar{\kappa}_n)}Y_j + \Psi_n^{-1}   \sum_{t=R_{\eta( b_n+\bar{\kappa}_n)}}^{b_n+\bar{\kappa}_n}\abs{f(X_{t})-\pi(f)} \label{tp_slln}
\end{align}
From (\ref{standardarguments2}) we see that $\eta(T)$ tends to infinity as $T \rightarrow \infty$ and $\lim_{T \rightarrow \infty} \eta(T)/T= 1/\mu_\varrho$ almost surely. Also for every positive sequence $m_T$ that tends to infinity as $T \rightarrow \infty$, we have that $\lim_{T \rightarrow \infty} \eta(m_T)/m_T= 1/\mu_\varrho$ almost surely. By an application of the law of iterated logarithm to \eqref{s2} and \eqref{PP_sip}, we see that 
\[
R_n=n/\mu_\varrho +\bigO_{a.s.}(\sqrt{n \log \log n})
\]
and
\[
L_n=n/\lambda +\bigO_{a.s.}(\sqrt{n \log \log n}).
\]
Consequently, we have that 
have that $b_n=\bigO_{a.s.}(n) $ and  $\eta(b_n)=\bigO_{a.s.}(n)$. Note that $\eta(b_n+\bar{\kappa}_n)$, the amount of regenerations until time $b_n+\bar{\kappa}_n$ is equal to the amount of generation until time $b_n$ and the amount of regenerations in the time interval $(b_n,b_n+\bar{\kappa}_n)$, i.e., 
$\eta(b_n+\bar{\kappa}_n)=\eta(b_n)+\eta(b_n+\bar{\kappa}_n)-\eta(b_n).$  Since $\eta(T)$ is a renewal process it is clear, that the amount of events should be proportional to the time interval and the intensity, i.e., that we should have $\eta(b_n+\bar{\kappa}_n)-\eta(b_n)=O(\bar{\kappa}_n/ \mu_\varrho)$ almost surely. We will now prove this claim. Since we have that $\mathbb{E}_\nu R_1^p< \infty$, by \cite[Theorem 2.4]{weighted_approx} we can construct a Brownian motion $\tilde{B}_2$ such that
\begin{equation}
    \abs{\eta(T)-\frac{T}{\mu_\eta}-\sigma_\eta \tilde{B}_2(T)}=o_{a.s.}(T^{1/p}),
\end{equation}\\
for some constants $\mu_\eta$ and $\sigma_\eta$. Hence for almost all sample paths $\omega$ there exists a $T_1(\omega)$ such that for all $T\geq T_1(\omega)$ we have that
\begin{equation}
    \frac{1}{T^{1/p}} \abs{\eta(T)-\frac{T}{\mu_\eta}-\sigma_\eta \tilde{B}_2(T)} < \varepsilon.
\end{equation}
Since $b_n$ is non-decreasing and tends to infinity almost surely, it follows that for all sample paths $\omega$ there exists a $N_2(\omega)$ such that $\eta(b_n)(\omega)\geq T_1(\omega)$ for all $n \geq N_2(\omega)$ and hence
\begin{equation}
    \frac{1}{b_n^{1/p}} \abs{\eta(b_n)-\frac{b_n}{\mu_\eta}-\sigma_\eta \tilde{B}_2(b_n)} < \varepsilon.
\end{equation}
Since $b_n=O(n)$ almost surely, it follows that
\begin{equation}
\label{b_sip}
\abs{\eta(b_n)-\frac{b_n}{\mu_\eta}-\sigma_\eta \tilde{B}_2(b_n)}=o_{a.s.}(b_n^{1/p})=o_{a.s.}(n^{1/p}).  
\end{equation} 
Then by the triangle inequality, we obtain
\begin{align}    \eta(b_n+\bar{\kappa}_n)-\eta(b_n)
    & \leq \abs{\eta(b_n+\bar{\kappa}_n) - (b_n+\bar{\kappa}_n)/\mu_\eta -\sigma_\eta  \tilde{B}_2(\eta(b_n)+\bar{\kappa}_n)}
    \label{bb11}\\
   &  + \bar{\kappa}_n/\mu_\eta+\abs{-\eta(b_n) + b_n /\mu_\eta  +\sigma_\eta  \tilde{B}_2(b_n)} \label{bb22}\\
   & + \sigma_\eta \abs{\tilde{B}_2(b_n+\bar{\kappa}_n)-\tilde{B}_2(b_n)}\\
    &\leq \bar{\kappa}_n/ \mu_\eta + o_{a.s.}(n^{1/p}).
    \label{bb44}
\end{align}
The last inequality follows, since by (\ref{b_sip})  the  term in (\ref{bb11}) and the second term in (\ref{bb22}) are $o(n^{1/p})$. Furthermore, by \cite[Theorem 2]{biga_csorgo} we have that for any $a_n \ll n$ that
\begin{align}
    \limsup_{n\rightarrow \infty} \sup_{0\leq s\leq a_n} \frac{\abs{\tilde{B}_2(n+s)-\tilde{B}_2(n)}}{\left[a_n (\log(n/a_n)+\log \log n)\right]^{1/2}}&= 1\quad \textrm{a.s.}\end{align}
    Since we have $\bar{\kappa}_n = \bigO
    (n^{1/p}),$ it follows that
\begin{align}
     \sup_{0\leq s\leq\bar{\kappa}_n} \abs{\tilde{B}_2(n+s)-\tilde{B}_2(n)}= \bigO_{a.s.}\left( n^{1/2p} \log (n) \right). \quad 
\end{align}
    Moreover, since $\eta(b_n)=\bigO_{a.s.}(n)$ and almost surely non-decreasing we also have that 
\begin{align}
\label{incr_bb2}
     \sup_{0\leq s\leq \bar{\kappa}_n} \abs{\tilde{B}_2(\eta(b_n)+s)-\tilde{B}_2(\eta(b_n))}= o_{a.s.}\left( \eta(b_n)^{1/2p} \right) = o_{a.s.}\left( n^{1/2p} \right).
\end{align}
Hence, the inequality in (\ref{bb44}) follows and  we have shown that $\eta(b_n+\bar{\kappa}_n)-\eta(b_n)\leq \bar{\kappa}_n\mu_\eta+o\left( n^{1/p} \right)$ almost surely. Therefore there exists a $K>0$ such that for almost all sample paths there exits an $N_3(\omega)$ sufficiently large such that $\eta(b_n+\bar{\kappa}_n)-\eta(b_n) < K n^{1/p}$ almost surely. For notational convenience let $\tilde{a}_n$ be defined as $K n^{1/p}$. 
 Since $(Y_k)_{k\geq0}$ form an i.i.d sequence we have by Lemma \ref{moments_regen_exp} and Proposition \ref{GA_iid}  that there exists a Brownian motion $B_3$ such that 
\begin{align}
\label{bap3}
    \abs{\sum_{k=0}^n Y_k - n\mu_Y-\Sigma_Y^{1/2} B_3(n)}= \bigO_{a.s.}(\tilde{C}_{\xi,d}\Psi_n).
\end{align}
where $\mu_Y$ and $\Sigma_Y^{1/2}$ denote the mean and square root of the covariance matrix of $Y_1$ respectively.
It immediately follows that  we also have
\begin{align}
\label{bap5}
    \abs{\sum_{k=0}^{\eta(b_n)} Y_k - \eta(b_n)\mu_Y-\Sigma_Y^{1/2}B_3(\eta(b_n))}= \bigO_{a.s.}(\tilde{C}_{\xi,d} \Psi_{\eta(b_n)}) = \bigO_{a.s.}(\tilde{C}_{\xi,d} \Psi_n).
\end{align}
By the triangle inequality, we obtain
\begin{align}    
\abs{\sum_{k=\eta(b_n)}^{\eta(b_n)+\tilde{a}_n}Y_k}
    & \leq \abs{\sum_{k=0}^{\eta(b_n)+\tilde{a}_n}Y_k - (\eta(b_n)+\tilde{a}_n)\mu_Y -\Sigma_Y^{1/2}  B_3(\eta(b_n)+\tilde{a}_n))}
    \label{bb1}\\
    &+ \tilde{a}_n \mu_Y+\abs{-\sum_{k=0}^{\eta(b_n)}Y_k 
    + \eta(b_n)\mu_Y +\Sigma_Y^{1/2}  B_3(\eta(b_n))}\\
    &+ \abs{\Sigma_Y^{1/2} }\abs{B_3(\eta(b_n)+\tilde{a}_n)-B_3(\eta(b_n))}
    \label{bb2}\\
    &\leq \tilde{a}_n \mu_Y + \bigO_{a.s.}(\tilde{C}_{\xi,d} \Psi_n).
\end{align}
The last inequality follows, since by (\ref{bap5})  both the term in (\ref{bb1}) and the second term in (\ref{bb2}) are $o(C_{\xi,N,d}n^{1/p})$ almost surely.
By again applying \cite[Theorem 2]{biga_csorgo} to every coordinate of $B_3$ we see that 
\begin{align}
\label{browniandifference2}
\big\lvert B_3(\eta(b_n)+\tilde{a}_n)-B_3(\eta(b_n)) \big\lvert &=   \left( \sum_{i=1}^d \left( B_{3i}(\eta(b_n)+\tilde{a}_n)-B_{3i}(\eta(b_n))\right)^2\right)^{1/2}\\
&=\bigO_{a.s.} ( {d}^{1/2}  n^{1/2p}\log n).
\end{align}
Furthermore, by (\ref{incr_bb2}) the last inequality also follows.
Hence it follows that \begin{equation}
\label{ll_c1}
\mathbb{P}_\nu  \left( \limsup_{n \rightarrow \infty} \frac{1}{C_{\xi,N,d}  n^{1/p}} \abs{\sum_{k=\eta(b_n)}^{\eta(b_n+a_n)}Y_k} \leq K \mu_Y\right) =1 .     \end{equation}
Hence the first term in the upper bound (\ref{tp_slln}) is  $O(1)$ almost surely. For the second term, we see that by a Borel--Cantelli argument that is the same as the one given to obtain (\ref{block_asymp}) that
\begin{align}
    Y_n=\sum_{R_n}^{R_{n+1}}\abs{f(X_k)-\pi(f)}=\bigO_{a.s.} (\alpha^{-1/p}d^{1/2}C_\xi^{1/p}n^{1/p}). 
\end{align} 
Therefore
\begin{align}\sum_{R_{\eta(b_n+\bar{\kappa}_n)}}^{b_n+\bar{\kappa}_n} \hspace{-0.5cm}\abs{f(X_{R_n+s})-\pi(f)}ds &\leq \sum_{R_{\eta(b_n+\bar{\kappa}_n})}^{R_{\eta(b_n+\bar{\kappa}_n)+1}} \hspace{-0.1cm}\abs{f(X_{R_n+s})-\pi(f)}ds\\&=Y_{\eta(b_n+\bar{\kappa}_n)}\\
&=\bigO_{a.s.}\left( \alpha^{-1/p}d^{1/2}C_{\xi}^{1/p}
(\eta(b_n+\bar{\kappa}_n))^{1/p} \right)   \\
&=\bigO_{a.s.}\left( d^{1/2}C_{\xi}^{1/p} (n+C(\omega) C_{\xi,N,d} \Psi_n)^{1/p}\right)\\
&=\bigO_{a.s.}\left(d^{1/2} C_{\xi}^{1/p}n^{1/p} \right)
\end{align}
 
Hence our claim (\ref{tp_slln}) follows, and consequently we have also shown (\ref{claim2}).
\noindent
 Combining (\ref{s1}), (\ref{claim1}), and (\ref{claim2}) it follows that

\begin{equation}
\abs{\sum_{k=1}^{\gamma L(n)}\xi_k-\beta \gamma L(n)+ \beta \varrho n - {V}_\xi^{\frac{1}{2}}B_1(n)}=\bigO_{a.s.}\left(\tilde{C}_{\xi,d}n^{1/p} \right)    \label{s3}
\end{equation}
\noindent
Let $(\Gamma_s)_{s \geq 0}$ be defined as
$\Gamma_0:=0$ and $\Gamma_s:=L^{-1}(s)$, the generalised inverse of the Poisson process. Taking $n'= \Gamma_n$ in (\ref{s3}) and subsequently making the substitution $n=n'/\gamma$, it follows that 
\begin{equation}
\abs{\sum_{k=1}^{n}\xi_k-\beta  n+ {\beta \varrho } \Gamma_{n/ \gamma} - {V}_\xi^{\frac{1}{2}}B_1(\Gamma_{n/ \gamma})}=\bigO_{a.s.}\left(\tilde{C}_{\xi,d}{\Gamma_n}^{{1} / {p}}\right)=\bigO_{a.s.}\left(\tilde{C}_{\xi,d} n^{{1} /{p}} \right)     \label{s4}
\end{equation}
\noindent
Since $\Gamma_n$ has a Gamma distribution, it follows from the Koml\'os--Major--Tusn\'ady approximation \citep[Theorem 1] {kmt_1} that there exists a Brownian motion $B_4$ such that
\begin{equation}
    \abs{\Gamma_n-\frac{n}{\lambda}-\frac{1}{\lambda} B_4(n)}=\bigO_{a.s.} (\log n).
    \label{s5}
\end{equation}

\noindent
Since the Poisson process $N$ is constructed deterministically from $B_2$ we have that $N$ and its corresponding event time process $\Gamma$ are independent of $B_1.$ Moreover, the components of a standard  $d$-dimensional Brownian motion are all independent. Therefore by a componentwise application of Lemma \ref{mer_lemma} it follows that there exists a standard $d$-dimensional Brownian motion $B_5$ independent of $N$ and $\Gamma$ such that 
\begin{equation}
\label{r1}
\abs{B_1(n)-\frac{1}{\sqrt{\lambda}} B_5(L(n))}=    \left(\sum_{i=1}^d \lvert B_{1i}(n)-\frac{1}{\sqrt{\lambda}} B_{5i}(L(n))\rvert^2 \right)^{1/2}=\bigO_{a.s.} (\sqrt{d} \log n).
\end{equation}
\noindent
Furthermore, by \eqref{r1}, there exists an almost surely finite random variable $C$ such that for almost all $\omega$ we have that $\textrm{for all } n\geq N_0\equiv N_0(\omega)$ we have that 

\begin{equation}
\label{split1}
\left. \limsup_{n \rightarrow \infty} \frac{1}{d \log n }\abs{B_1(n)-\frac{1}{\sqrt{\lambda}} B_5(L(n))}\right. < C(\omega)
\end{equation}
\noindent
 Since $\Gamma_n$ is an increasing process and tends to infinity, we have that for almost every  $\omega$ there exists a  $N'_0\equiv N'_0(\omega)$ such that $\Gamma_{n}(\omega)\geq N_0$ for all $n\geq N'_0$. Hence we obtain from (\ref{split1}) that
 \begin{equation}
 \label{split2}
 \limsup_{n \rightarrow \infty} \frac{1}{d  \log  \Gamma_n }\abs{B_1(\Gamma_n)-\frac{1}{\sqrt{\lambda}} B_4(n)}< C(\omega)   \quad \textrm{a.s. },
\end{equation}
\noindent
where we used that $L(\Gamma(n))=n.$ Therefore we see that
\begin{equation}
\abs{B_1(\Gamma_n)-\frac{1}{\sqrt{\lambda}} B_5(n)}=\bigO_{a.s.} (d \log \Gamma_n)= \bigO_{a.s.} (d \log n).
\label{s6}
\end{equation}
\noindent
Here the last equality follows since we can apply the law of iterated logarithm for Brownian motion to (\ref{s5}) which gives
$$\Gamma_n=\frac{n}{\lambda}+\bigO_{a.s.} (\sqrt{n \log \log n}).$$
Applying the obtained approximations given in (\ref{s5}) and (\ref{s6}) to (\ref{s4}) we see that
\begin{equation}
\abs{\sum_{k=1}^n \xi_k-\left(\frac{ {V}_\xi^{\frac{1}{2}}}{\sqrt{\lambda \gamma}}B_5(t)- \frac{\beta \varrho}{\lambda \sqrt{\gamma}}B_4(t)\right)} = \bigO_{a.s.}\left(\tilde{C}_{\xi,d}\Psi_n \right).
\end{equation}
\noindent
Note that since $B_4$ is independent of all components of $B_5$  we have that 
\begin{equation}
W_t=\Sigma_f^{-1}\left(\frac{ {V}_\xi^{\frac{1}{2}}}{\sqrt{\lambda \gamma}}B_5(t)- \frac{\beta \varrho}{\lambda \sqrt{\gamma}}B_4(t)\right)
\end{equation}
\noindent
is a standard $d$-dimensional Brownian motion since
\begin{equation}
\frac{ \tilde{V}_\xi}{\gamma \lambda}+  \frac{\beta \beta^T \varrho^2}{\gamma \lambda^2}=\frac{\Varnu(\xi_1)}{\mu_\varrho}=\Sigma_f,
\end{equation}
\noindent
where the last equality follows from \cite[Theorem 9]{sigman}. 
Note that by the same argument as given in Theorem \ref{theorem_approx_discrete_multi}, if we assume that the initial distribution is $\pi$, it can shown that the initial cycle is asymptotically negligible. Furthermore, in the case where the polynomial drift condition holds, Lemma \ref{moments_regen_poly} gives us the required moment conditions and the proof of the Theorem follows analogously. Similarly, the weak approximation case follows by the same argument.

\end{proof}

\begin{remark}
\label{initial_remark_proof}
Note that we can extend Theorem \ref{theorem_approx_discrete_one} and Theorem \ref{theorem_approx_discrete_multi} to arbitrary initial distributions, provided that the first cycle, until the first draw from the small measure $\nu$, is asymptotically negligible. For $g:\mathbb{E}\rightarrow \mathbb{R}$  we can define the norm $\abs{g}_V:=\sup_{x\in E}\frac{g(x)}{V(x)}$. Then if we assume that a geometric drift condition holds, we have for any $f$ such that $\sup_{i=1,\cdots,d}\abs{f_i}_V<\infty$  by an application of the Comparison theorem; \cite[Theorem 14.2.2]{meyn_tweedie_2012} that
\begin{align}
\abs{\sum_{t=0}^{R_1} f(X_k)} &\leq d \sup_{i=1,\cdots,d}\sum_{t=0}^{R_1} \abs{f_i(X_k)}\\
&\leq d\frac{\sup_{i=1,\cdots,d}\abs{f_i}_V}{1-\lambda}\left(V(x)+b \mathbb{E}_xR_1 \right),
\end{align}
where we used the fact that if $\abs{g}_V<\infty$ then $g(x)\leq V(x) \abs{g}_V$. By Lemma \ref{moment_bound_exp} and an application of Jensen's inequality, we have that

\begin{align}
\label{exp_initial}
\abs{\sum_{t=0}^{R_1} f(X_k)} \leq d\ \frac{\sup_{i=1,\cdots,d}\abs{f_i}_V}{1-\lambda}\left(V(x)+b \log_r\left( \frac{\alpha G(r,x)}{1-(1-\alpha)r^{a}}\right)\right),
\end{align}

where

\begin{equation*}
    G(r,x) =V(x) \mathbbm{1}_C(x)+ r(\lambda \upsilon_C+b)\mathbbm{1}_{C^c}(x)
\end{equation*}
and
$$a=1+\left(\ln\frac{\lambda \upsilon_V+b-\alpha}{1-\alpha}\right)/(\ln(\lambda^{-1})).$$

Similarly, if we assume that a polynomial drift condition holds, from Lemma \ref{moment_bound_poly} we see that,
$$\mathbb{E}_x[  {\tau_C}] \leq \frac{1}{(1-\eta) c} \left(V^{1-\eta}(x)+(b^\eta+b_0)\mathbbm{1}_C(x)\right).$$
   
A similar argument with the comparison theorem gives us
\begin{align}
\label{poly_initial}
\abs{\sum_{t=0}^{R_1} f(X_k)} \leq d\ \frac{\sup_{i=1,\cdots,d}\abs{f_i}_{V^{\eta}}}{c}\left(V(x)+ \frac{1}{(1-\eta) c} \left(V^{1-\eta}(x)+(b^\eta+b_0)\right)\right),
\end{align}
provided that $\sup_{i=1,\cdots,d}\abs{f_i}_{V^{\eta}}<\infty$.
\end{remark}

\subsection{Proofs of Section \ref{section:4}}

\subsubsection{Proof of Theorem \ref{multi_bm_theorem}}

\begin{proof}
    In \cite[Theorem 2]{multivariate_output} it is shown that for every $i,j$ we have 
    \[    \abs{\widehat{\Sigma}^{BM}_{T_{ij}}- \Sigma_f}=\bigO_{a.s.}\left(\left(\frac{\ell_{T,d}}{T}\right)^{1/2}\right)+ \bigO_{a.s.} \left(\bar{\psi}_N \psi_d\Psi_T\log(T)\ell_{T,d}^{-1/2}
    \right) + \bigO_{a.s.} \left(\frac{\bar{\psi}^2_N \psi^2_d\Psi^2_T\log(T)}{T}
    \right). 
    \]
Since we have that 
\begin{align*}
\abs{\widehat{\Sigma}^{BM}_{T}- \Sigma_f} &=\sqrt{ 
\sum_{i=1}^d\sum_{j=1}^d\abs{\widehat{\Sigma}^{BM}_{T_{ij}}- \Sigma_{f_{ij}}}^2}\\
&=\bigO_{a.s.}\left(d\left(\frac{\ell_{T,d}}{T}\right)^{1/2}\right)+ \bigO_{a.s.} \left(\bar{\psi}_N d\psi_d\Psi_T\log(T)\ell_{T,d}^{-1/2}
    \right) + \bigO_{a.s.} \left(\frac{\bar{\psi}^2_N d\psi^2_d\Psi^2_T\log(T)}{T}
    \right)\\
    &=o_{a.s.}(1).    
\end{align*}
Firstly, note that if $$T = \left({\bar{\psi}_N d\psi_d }\right)^{\frac{2p_0}{(p_0-2)}(1+\bar{\delta})},$$

 we have that for $T\rightarrow \infty$  that

$$\bigO_{a.s.} \left(\frac{\bar{\psi}^2_N d\psi^2_d\Psi^2_T\log^2(T)}{T}\right)=\bigO_{a.s.} \left(d^{-1}\left({\bar{\psi}_N d\psi_d }\right)^{-2\bar{\delta}} \log^3 (\bar{\psi}_Nd\psi_d)\right)=o_{a.s.}(1).$$

Note that in order to find the optimal batch size $\ell_T=\floor{T^\alpha}$ such that $\abs{\widehat{\Sigma}^{BM}_{T}- \Sigma_f}$ tends to zero at the fastest rate, we equate the error terms
$\bigO_{a.s.}\left(d\left(\frac{\ell_{T,d}}{T}\right)^{1/2}\right)$ and $ \bigO_{a.s.} \left(\bar{\psi}_N d\psi_d\Psi_T\log(T)\ell_{T,d}^{-1/2}
    \right)$. This gives us that the optimal batch size, up to 
 a logarithmic factor, should be of an asymptotic magnitude
    \begin{align*}
        \ell_T &\asymp {\bar{\psi}_N \psi_d\Psi_T T^{1/2}}\\
         & \asymp d^{-(p_0-2)/(2p_0(1+\bar{\delta}))} T^{1/2+1/p_0+(p_0-2)/(2p_0(1+\bar{\delta}))} \\
         & \asymp d^{-(p_0-2)/(2p_0(1+\bar{\delta}))} \left(\bar{\psi}_N {d}\psi_d
    \right)^{1+\frac{p_0+2}{p_0-2}(1+\bar{\delta})},  
    \end{align*}
which gives us 

\[
\abs{\widehat{\Sigma}^{BM}_{T}- \Sigma_f}= \bigO_{a.s.}\left(  \sqrt{\frac{\bar{\psi}_Nd^2\psi_d}{T^{1/2-1/p_0}}}\right).
\]
With our choice for the simulation time $T = \left({\bar{\psi}_N d\psi_d }\right)^{\frac{2p_0}{(p_0-2)}(1+\bar{\delta})}$ we see that

\begin{align*}
\abs{\widehat{\Sigma}^{BM}_{T}- \Sigma_f}&= \bigO_{a.s.}\left( \sqrt{d}  T^\frac{-\bar{\delta}(p-2)}{4p_0(1+\bar{\delta})}\right)=\bigO_{a.s.}\left(\sqrt{d} \left({\bar{\psi}_N d\psi_d }\right)^{-\bar{\delta}/2}\right)\\
&=\bigO_{a.s.}\left( {\bar{\psi}_N }^{-\bar{\delta}/2} d^{1/2-(a+1)\bar{\delta}/2}\right)=o_{a.s.}(1),  
\end{align*}
where the last equality follows since $\bar{\delta}> 1/(1+a).$
\end{proof}

\begin{remark}
In the case that $\Psi_T=T^{1/4+1/4(p_0-1)}$ and $\psi_d=d^a$ for some given $a>0$ we require 
\[
T=\left(\psi_N d^{1/4} \psi_d \right)^{\frac{p_0-1}{p_0-2}4(1+\bar{\delta})}, 
\]
with $\delta>1/(1+a)$ and consequently
 \begin{align*}
        \ell_T & \asymp {\bar{\psi}_N \psi_d\Psi_T T^{1/2}}\\
         & \asymp T^{\frac{3}{4}+ \frac{1}{4(p_0-1)}+\frac{(p_0-2)}{4(p_0-1)(1+\bar{\delta})}}
    \end{align*}
\end{remark}

\subsubsection{Proof of Theorem \ref{FVSR_optimal_rate}}

\begin{proof}
To prove the first claim, we first show that the difference between termination rules based on the volumes of the ellipsoids based on the estimated and asymptotic covariance matrix tends to zero at an appropriate rate. 
Note that the volume of the confidence ellipsoid is given by
\begin{equation}\label{volume_ellips}
\Vollie(C_T(\alpha))= T^{-d/2}q^{d/2}_{\alpha,d}\frac{2\pi^{d/2}}{d \Gamma(d/2)} 
 \det(\widehat{\Sigma}^{1/2}_T).
\end{equation}
Furthermore under Assumption \ref{assumption_covariance}, we have that $\sigma_0^{1/2d}I_d\preccurlyeq \Sigma_f^{1/2d}\preccurlyeq \sigma_d^{1/2d}I_d,$ where the matrix inequalities hold in the positive semi-definite sense. Furthermore, by Jacobi's formula, we have that \begin{equation}
\label{det_derivative}
\frac{\partial \det(\Sigma)}{\partial \Sigma}=\det(\Sigma)\Sigma^{-1}.    
\end{equation}
Note that the choice $\Sigma=\sigma_d^{1/2d}I_d$ maximises the function $\det(\Sigma)\abs{\Sigma^{-1}}$ subject to the constraint $\sigma_0^{1/2d}I_d\preccurlyeq \Sigma \preccurlyeq \sigma_d^{1/2d}I_d$. Hence
 we have the following Lipschitz property for the determinant on this domain
\begin{equation}\label{det_lips}
\abs{\det(\widehat{\Sigma}^{1/2d}_T)-\det( {\Sigma}_f^{1/2d})} \leq \sigma_d^{1/2-1/2d} \sqrt{d}\abs{\widehat{\Sigma}^{1/2d}_T- {\Sigma}_f^{1/2d}}.   
\end{equation}

 By Ando--van Hemmen's inequality given in Lemma \ref{Ando_inequality}, we have under Assumption \ref{assumption_covariance}, that
    \begin{equation}
    \label{square_root_lips}        
  \abs{\widehat{\Sigma}^{1/2d}_T- {\Sigma}_f^{1/2d}} \leq \left(\frac{1}{\sigma_0}\right)^{1-1/2d} \abs{\widehat{\Sigma}_T-{\Sigma}_f}.
  \end{equation}
By Theorem \ref{multi_bm_theorem} it follows that 
$\abs{\widehat{\Sigma}_T- {\Sigma}_f}=R_T$ with
\begin{align}
\label{conv_rate_variance}
R_T=\bigO_{a.s.}\left(d\left(\frac{\ell_{T,d}}{T}\right)^{1/2}\right)+ \bigO_{a.s.} \left(\bar{\psi}_N d\psi_d\Psi_T\log(T)\ell_{T,d}^{-1/2}
    \right) + \bigO_{a.s.} \left(\frac{\bar{\psi}^2_N d\psi^2_d\Psi^2_T\log(T)}{T}
    \right).
\end{align}
Combining \eqref{det_lips} and \eqref{square_root_lips}, and since $\det(\Sigma_f^{1/2d})\geq \sigma_0^{1/2}$ we see that 
\begin{equation}
    \label{ratio_sigma}     \abs{\frac{\det(\widehat{\Sigma}^{1/2d}_T)}{\det(\Sigma^{1/2d}_f)}-1}
    \le \left(\frac{\sigma_d}{\sigma_0}\right)^{1/2-1/2d} \frac{\sqrt{d}}{\sigma_0} R_T\leq \left(\frac{\sigma_d}{\sigma_0}\right)^{1/2} \frac{\sqrt{d}}{\sigma_0} R_T.
\end{equation}

Under Assumption \ref{assumption_covariance}, we have by the Gershgorin circle theorem that $\sigma_d$ is bounded by
\begin{equation}
\label{eigenvalue_bound}
\sigma_d \leq \sup_{i\in \{1,\dots,d\}}\Sigma_{f_{ii}}+ \sup_{i\in \{1,\dots,d\}}\sum_{j \neq i}\abs{\Sigma_{f_{ij}}} \lsim d^2.    
\end{equation}
Let $c_{\alpha,d}:=q^{d/2}_{\alpha,d}\frac{2\pi^{d/2}}{d \Gamma(d/2)}$, then from  \eqref{volume_ellips}, \eqref{ratio_sigma}, and \eqref{eigenvalue_bound} we see that 
\begin{align}  
\sqrt{T}\frac{\Vollie(C_T(\alpha))^{1/d}}{c_{\alpha,d}^{1/d}\det(\Sigma_f)^{1/2d}}&= 
 \frac{\det(\widehat{\Sigma}^{1/2d}_T)}{\det(\Sigma^{1/2d}_f)} \nonumber \\
 &= 1+\left(\frac{\sigma_d}{\sigma_0}\right)^{1/2} \frac{\sqrt{d}}{\sigma_0}\bigO_{a.s.}(R_T) \nonumber \\
 &= 1+ \bigO_{a.s.}(d^{3/2} R_T). \nonumber 
 \end{align}

Let \[\tilde{T}^* (\varepsilon,d,N):=\left({\bar{\psi}_N d^{3}\psi_d }\right)^{\frac{2p_0}{(p_0-2)}(1+\bar{\delta}_1)} \left(\frac{1}{\varepsilon}\right)^{\frac{4p_0}{(p_0-2)}(1+\bar{\delta}_2)}\wedge e^{\frac{8p_0}{p_0-2}}\],

then for  $T\geq  \tilde{T}^* (\varepsilon,d,N)
                       $ we have that 
      %
      %
      %
%
%
                       \begin{align}         
                       \label{error_term_RT1}
                       d^{3/2}R_T&=\bigO_{a.s.}\left(\bar{\psi}_N^{-\bar{\delta}_1/2}\psi_d^{-\bar{\delta}_1/2}d^{1-3/2\bar{\delta}_1}\varepsilon^{1+\bar{\delta}_2}\log^2(\tilde{T}^* (\varepsilon,d,N))\right) \nonumber \\
                       &=\bigO_{a.s.}\left(\bar{\psi}_N^{-\bar{\delta}_1/2}d^{-1/2}\varepsilon^{1+\bar{\delta}_2}\log^2(\tilde{T}^* (\varepsilon,d,N))\right) \nonumber\\
                        &=\bigO_{a.s.}\left(\bar{\psi}_N^{-\bar{\delta}_1/2}d^{-1/2}\varepsilon^{1+\bar{\delta}_2}\left(\log^2({\bar{\psi}_N d^{3}\psi_d }) +\log^2(1/\varepsilon)\right)\right)\nonumber \\
                        &=o_{a.s.}\left( \log^2(\bar{\psi}_N d^3 \psi_d )\bar{\psi}_N^{-\bar{\delta}_1/2}d^{-1/2} \varepsilon\right),
                       \end{align}  
   where the  second equality follows since $\bar{\delta}_1> 3/(3+a)$ and subsequent inequalities follow by definition of $\tilde{T}^* (\varepsilon,d,N)$, the fact that $\log^a(x)/x^b$ is decreasing on $x\geq e^{a/b}$ for any $a,b>0$, and some basic computations. Therefore we see that
\begin{equation}
    \label{volume_asymptotics}   
    \sqrt{T}\frac{\Vollie(C_T(\alpha))^{1/d}}{c_{\alpha,d}^{1/d}\det(\Sigma_f)^{1/2d}}=1+o_{a.s.}\left( \log^2(\bar{\psi}_N d^3 \psi_d )\bar{\psi}_N^{-\bar{\delta}_1/2}d^{-1/2} \varepsilon\right)   
\end{equation}
and consequently 
\begin{equation}
\label{volume_asymptotics2}    
\sqrt{T}\frac{\Vollie(C_T(\alpha))^{1/d}}{c_{\alpha,d}^{1/d}\det(\Sigma_f)^{1/2d}} \xrightarrow{a.s.} 1 \ \textrm{as}\ \varepsilon \downarrow 0.
 \end{equation}

 The remainder of the first part of the proof now follows by the argument of \cite[Theorem 1]{glynn}.
Let $V(T)=\Vollie(C_T(\alpha))^{1/d}+a(T)$, then by definition of $T_1(\varepsilon)$ we have that $V(T_1(\varepsilon)-1)>\varepsilon$ and that there exists a random variable $Z(\varepsilon) \in [0,1]$ such that $V(T_1(\varepsilon)+Z(\varepsilon))\leq \varepsilon$. This gives us
\[
\limsup_{\varepsilon \downarrow 0} \varepsilon T_1^{1/2}(\varepsilon) \leq \limsup_{\varepsilon \downarrow 0}    V(T_1(\varepsilon)-1) T^{1/2}_1(\varepsilon).
\]
Since $T_1(\varepsilon) \rightarrow \infty$ almost surely as $\varepsilon$ tends to zero it follows from  \eqref{volume_asymptotics} that 
\[
\limsup_{\varepsilon \downarrow 0}\frac{\varepsilon T^{1/2}_1(\varepsilon) }{\left( c_{\alpha,d} 
 \det({\Sigma}^{1/2}_f)\right)^{1/d}} \leq
 \limsup_{\varepsilon \downarrow 0}\frac{V(T_1(\varepsilon)-1) T^{1/2}_1(\varepsilon) }{\left( c_{\alpha,d} 
 \det({\Sigma}^{1/2}_f)\right)^{1/d}}=1+ \  o_{a.s.}\left( \log^2(\bar{\psi}_N d^3 \psi_d )\bar{\psi}_N^{-\bar{\delta}_1/2}d^{-1/2} \varepsilon\right).
\]

By a similar argument, we also have that 
\[
\liminf_{\varepsilon \downarrow 0} \frac{\varepsilon T_1^{1/2}(\varepsilon)}{(c^2_{\alpha,d} \det({\Sigma_f}))^{1/2d}} \geq \liminf_{\varepsilon \downarrow 0}    \frac{V(T_1(\varepsilon)+Z(\varepsilon)) T_1^{1/2}(\varepsilon)}{(c^2_{\alpha,d} \det({\Sigma_f}))^{1/2d}} = 1+ \  o_{a.s.}\left( \log^2(\bar{\psi}_N d^3 \psi_d )\bar{\psi}_N^{-\bar{\delta}_1/2}d^{-1/2} \varepsilon\right).
\]
This proves the first part of the Theorem. To prove the second claim,
 we first bound the difference between the confidence ellipsoids based on the estimated and asymptotic covariance matrix. Recall that $\hat{\pi}_T(f)=T^{-1}\sum_{t=1}^Tf(X_t)$ and let\\
\begin{align*}
&E_T:=\abs{T(\hat{\pi}_T(f)-\pi(f))^T\Sigma_f^{-1}(\hat{\pi}_T(f)-\pi(f))-T(\hat{\pi}_T(f)-\pi(f))^T\widehat{\Sigma}_T^{-1}(\hat{\pi}_T(f)-\pi(f))}
\end{align*}
Then by Cauchy--Schwarz's inequality we have that
\begin{equation}
E_T\leq T\abs{\widehat{\Sigma}_T^{-1}-\Sigma_f^{-1}} \abs{\hat{\pi}_T(f)-\pi(f)}^2.    
\end{equation}

Since $\Sigma^{-1}_f(\widehat{\Sigma}_T-\Sigma_f)\widehat{\Sigma}^{-1}_T=\Sigma_f^{-1}-\widehat{\Sigma}_T^{-1}$, we have by sub-multiplicativity of the Frobenius norm that 
\[
T\abs{\widehat{\Sigma}_T^{-1}-\Sigma_f^{-1}} \abs{\hat{\pi}_T(f)-\pi(f)}^2\leq T \abs{\widehat{\Sigma}_T-\Sigma_f}\abs{\widehat{\Sigma}_T^{-1}} \abs{\Sigma_f^{-1}}\abs{\hat{\pi}_T(f)-\pi(f)}^2.
\]
From our assumed Gaussian approximation result and the law of iterated logarithm, we see that
\begin{align*}
T\abs{\hat{\pi}_T(f)-\pi(f)}^2&\leq 2T\abs{\hat{\pi}_T(f)-\pi(f)-T^{-1}\Sigma_f^{1/2}W_T}^2+ 2T^{-1}\abs{\Sigma_f^{1/2}W_T}^2.
\end{align*}
By a coordinate-wise application of the law of iterated logarithm, we have that 
\begin{equation}
    \label{LITL}
\abs{\Sigma_f^{1/2}W_T}^2 \leq \tr(\Sigma_f) \abs{W_T}^2= \bigO_{a.s} ( \tr(\Sigma_f) d{T \log \log T}).
\end{equation}

This gives us
\begin{align*}
T\abs{\hat{\pi}_T(f)-\pi(f)}^2& = \bigO_{a.s.}\left( \frac{\bar{\psi}^2_N \psi^2_d \Psi_T^2}{T}\right) + \bigO_{a.s.}\left(d \tr(\Sigma_f) {\log \log T}\right)
\end{align*}

Let $\hat{\sigma}_1$ and $\sigma_1$ denote the smallest eigenvalues of $\widehat{\Sigma}_T$ and $\Sigma_f$ respectively. Then by the equivalence of the Frobenius and spectral norm that
\begin{align*}
    \abs{\hat{\Sigma}_T^{-1}} \leq \sqrt{d} \abs{\hat{\Sigma}_T^{-1}}_*\leq \frac{\sqrt{d}}{ \hat{\sigma}_1}.
\end{align*}

Therefore we have that 
\begin{align*}
    \abs{\hat{\Sigma}_T^{-1}}&\leq \frac{\sqrt{d}}{ {\sigma}_1}+\sqrt{d}\abs{\frac{1}{\hat{\sigma}_1}-\frac{1}{{\sigma}_1}}\\
    &\leq \frac{\sqrt{d}}{ {\sigma}_1}+\frac{\sqrt{d}}{\sigma_0^2}\abs{{\hat{\sigma}_1}-{{\sigma}_1}},
\end{align*}
where the last inequality follows since by Assumption \ref{assumption_covariance} we have that $\sigma_1,\hat{\sigma}_1\geq \sigma_0>0.$ By Weyl's Perturbation Theorem, see \cite[Corollary III.2.6]{bhatia2013matrix}, we have that for all Hermitian matrices $\Sigma_1,\Sigma_2$  that
\[
\max_{i} \abs{\sigma_i(\Sigma_1)-\sigma_i(\Sigma_2)}\leq \abs{\Sigma_1-\Sigma_2}_*,
\]
where $\sigma_i(\Sigma)$ denotes the $i$-th eigenvalue of $\Sigma$ for $i=1,\cdots,d.$ Hence the eigenvalues of $\hat{\Sigma}_T$ converge at the same rate as the matrix itself, see also \cite[Theorem 3]{multivariate_consistency}. This gives us $$
    \abs{\hat{\sigma}_1-\sigma_1} \leq \abs{\widehat{\Sigma}_T-\Sigma_f}_*\leq R_T.$$   

Furthermore, we also have that $\abs{\Sigma^{-1}_f}\leq \sqrt{d}/\sigma_1$. Therefore it follows that
\begin{align}
\label{convergence_ellipsoid}
\abs{\widehat{\Sigma}_T^{-1}-\Sigma_f^{-1}} \abs{\hat{\pi}_T(f)-\pi(f)}^2&\leq \abs{\widehat{\Sigma}_T-\Sigma_f}\bigO_{a.s.}\left( d^{2} \frac{\trace(\Sigma_f)}{\sigma_0^2} \log \log T \right) \nonumber \\
&\leq R_T \bigO_{a.s.}\left( d^{2} \frac{\trace(\Sigma_f)}{\sigma_0^2}  \log^{1/2} T \right).
\end{align}

Given the specification of $R_T$  given in \eqref{conv_rate_variance} we have that
\[
E_T\leq E_{1T}+ E_{2T}+ E_{3T},
\]

where

\begin{align}
  E_{1T}&= \bigO_{a.s.}\left(d^3 \frac{\trace(\Sigma_f)}{\sigma_0^2}\left(\frac{\ell_{T,d}}{T}\right)^{1/2} \hspace{-0.3cm}\log^{1/2} (T)\right)\\
  E_{2T}&=\bigO_{a.s.} \left( \frac{\bar{\psi}_N d^3 \trace(\Sigma_f)\psi_d\Psi_T\log^{3/2}(T)}{\sigma^2_0\ell_{T,d}^{1/2}}
    \right) \\
    E_{3T} &=  \bigO_{a.s.} \left(\frac{\bar{\psi}^2_N d^3\trace(\Sigma_f)\psi^2_d\Psi^2_T\log^{3/2}(T)}{\sigma_0^2T}
    \right)
\end{align}
Note that for $T\geq T^* (\varepsilon,d,N) = 
                       \left({\bar{\psi}_N \frac{\tr(\Sigma_f)}{\sigma^2_0}d^{3}\psi_d }\right)^{\frac{2p_0}{(p_0-2)}(1+\bar{\delta}_1)} \left(\frac{1}{\varepsilon}\right)^{\frac{4p_0}{(p_0-2)}(1+\bar{\delta}_2)}\vee e^{\frac{10p_0}{p_0-2}}$ we have that
$E_{3T}=o_{a.s.}(\bar{\psi}_N^{-2\bar{\delta}_1}\psi_d^{-2\bar{\delta}_1}d^{-3-6\bar{\delta}_1}\varepsilon^{4(1+\bar{\delta}_2)}\log^{3/2}(T^*(\varepsilon,d,N)))$.
Note that by the choice of $\ell_T$ it follows that $E_{1T}$ and $E_{2T}$ are almost surely of the same asymptotic magnitude up to a log factor. By the same argument as given in \eqref{error_term_RT1} it follows that both  $E_{1T}$ and $E_{2T}$ are of order $o_{a.s.}(\log^{5/2}(\bar{\psi}_N d^3 \psi_d )\bar{\psi}_N^{-\bar{\delta}_1/2}d^{3/2-(3+a)\bar{\delta}_1/2}\varepsilon)$ 
and consequently, it follows that  \begin{equation}
\label{error_CI}
 E_T= o_{a.s.}(\log^{5/2}(\bar{\psi}_N d^3 \psi_d )\bar{\psi}_N^{-\bar{\delta}_1/2}d^{3/2-(3+a)\bar{\delta}_1/2}\varepsilon).    
\end{equation}
 Now we show that the confidence ellipsoid based on $\Sigma_f$ has the asymptotically correct coverage. Note that 
\begin{align*}
&\abs{T(\hat{\pi}_T(f)-\pi(f))^T\Sigma_f^{-1}(\hat{\pi}_T(f)-\pi(f))-T^{-1}\langle W_T,W_T\rangle}\\
&=\abs{T^{-1/2}\left(\sum_{t=1}^Tf(X_t)-T\pi(f)-\Sigma_f^{1/2}W_T\right)^T\Sigma_f^{-1}T^{-1/2}\left(\sum_{t=1}^Tf(X_t)-T\pi(f)+\Sigma_f^{1/2}W_T\right)}\\
&\leq \frac{\sqrt{d}}{\sigma_0 T}\abs{\sum_{t=1}^Tf(X_t)-T\pi(f)-\Sigma_f^{1/2}W_T}^2+\frac{2\abs{\Sigma^{-1/2}_f}}{T}\abs{\sum_{t=1}^Tf(X_t)-T\pi(f)-\Sigma_f^{1/2}W_T}\abs{W_T},
\end{align*}
where the last inequality follows from  Cauchy--Schwarz and since we have by the equivalence of the Frobenius and spectral norm that $\abs{\Sigma_f^{-1}}\leq \sqrt{d}/\sigma_0$. Moreover, by the assumed weak Gaussian approximation and \eqref{LITL} we obtain 
\begin{align}
\label{convergence_probability}
\abs{T(\hat{\pi}_T(f)-\pi(f))^T\Sigma_f^{-1}(\hat{\pi}_T(f)-\pi(f))-\frac{\langle W_T,W_T\rangle}{T}}&=\bigO_{a.s.}\left( \frac{\bar{\psi_N}^2 \sqrt{d} \psi_d^2 \Psi^2_T}{T}
\right) +\bigO_{p}\left( \frac{{\bar{\psi}_N d^{3/2}\psi_d \Psi'_T}}{\sqrt{T}}
\right) \nonumber \\
&=\bigO_{a.s.}\left( \frac{{\bar{\psi}_N d^{3/2}\psi_d \Psi_T\log^{1/4}(T)}}{\sqrt{T}}
\right) \nonumber \\
&=o_{a.s.}\left( \frac{\log (\bar{\psi}_N d^{3/2}\psi_d) }{(\bar{\psi}_N d^{3/2}\psi_d)^{\bar{\delta}_3} }\varepsilon \right),
\end{align}

with $\Psi'_T=\Psi_T(\log \log T)^{1/2}$ and for all $T\geq \left(\frac{\bar{\psi}_N d^{3/2} \psi_d}{\varepsilon}\right)^{\frac{2p_0}{(p_0-2)}(1+\bar{\delta}_3)}\vee e^{\frac{5p_0}{2(p_0-2)}}$ for any $\bar{\delta}_3>0$. 
Finally,  given  \eqref{error_CI} and \eqref{convergence_probability}, we can use the argument of \cite[Theorem 1]{glynn} to show that at termination time $T_1(\varepsilon)$ the empirical confidence interval also has the correct coverage as $\varepsilon \downarrow 0$.
\end{proof}

\subsubsection{Proof of Theorem \ref{FVSR_one_dep_sip}}
\begin{proof}
The proof follows completely analogous to the proof of Theorem \ref{FVSR_optimal_rate}. Note that we now obtain
\[\abs{T(\hat{\pi}_T(f)-\pi(f))^T\hat{\Sigma}_f^{-1}(\hat{\pi}_T(f)-\pi(f))-T^{-1}\langle W_T,W_T\rangle}=\bigO_{a.s.} \left( \frac{\bar{\psi}_N d^3 \tr(\Sigma_f)\psi_d\log^{2}(T)}{\sigma_1 T^{\frac{p_0-2}{8(p_0-1)}}}
    \right).\]
\end{proof}

\subsubsection{Proof of Corollary \ref{ESS_corr}}
\begin{proof}
By the same argument as the first part of Theorem \ref{FVSR_optimal_rate} it follows that
      \[
    \frac{\textrm{ESS}}{T}= \left(\frac{\lvert{\Gamma}_f\rvert}{\lvert{\Sigma}_f\rvert}\right)^{1/d} + \bigO_{a.s.}\left(\frac{\sigma_d}{\sigma_0}\sqrt{d}\abs{\hat{\Sigma}_T-\Sigma_f}\right).
    \]
The claim now follows from Theorem \ref{multi_bm_theorem}.
\end{proof}

\section*{Acknowledgements}
We would like to thank Joris Bierkens and Fabian Mies for helpful discussions on the results. This work is part of the research programme ‘Zigzagging through computational barriers’ with project number 016.Vidi.189.043, which is financed by the Dutch Research Council (NWO).

\bibliography{References,main}

\end{document}